\documentclass[journal,onecolumn,12pt]{IEEEtran}
\usepackage{times,amsmath,amsthm,epsfig,amssymb,graphicx,amsfonts}
\usepackage{hyphenat}
\usepackage{psfrag}
\usepackage{ifpdf}
\usepackage{cite}
\usepackage{array}
\usepackage{multirow}
\usepackage{graphicx}
\usepackage{color}
\usepackage{subfigure}
\usepackage{epsfig}
\usepackage{citesort}
\usepackage{setspace}
\usepackage{latexsym}
\usepackage{amssymb}
\usepackage{amsmath}
\usepackage{flexisym}

\title{On the Capacity Regions of Two-Way Diamond Channels}
\author{Mehdi Ashraphijuo, Vaneet Aggarwal and Xiaodong Wang \thanks{ The authors are with the Electrical Engineering Department, Columbia University, New York, NY 10027, email: \{mehdi,wangx\}@ee.columbia.edu, vaneet@alumni.princeton.edu. }}
\newtheorem{theorem}{Theorem}

\newtheorem{corollary}{Corollary}

\newtheorem{remark}{Remark}
\newtheorem{example}{Example}

\begin{document}
\maketitle
\begin{abstract}
In this paper, we study the capacity regions of two-way diamond channels. We show that for a linear deterministic model the capacity of the diamond channel in each direction can be simultaneously achieved for all values of channel parameters, where the forward and backward channel parameters are not necessarily the same. We divide the achievability scheme into three cases, depending on the forward and backward channel parameters. For the first case, we use a reverse amplify-and-forward strategy in the relays. For the second case, we use four relay strategies based on the reverse amplify-and-forward with some modifications in terms of replacement and repetition of some stream levels. For the third case, we use two relay strategies based on performing two rounds of repetitions in a relay. The proposed schemes for deterministic channels are used to find the capacity regions within constant gaps for two special cases of the Gaussian two-way diamond channel. First, for the general Gaussian two-way relay channel with a simple coding scheme the smallest gap is achieved compared to the prior works. Then, a special symmetric Gaussian two-way diamond model is considered and the capacity region is achieved within four bits.
\end{abstract}

\


\newpage
\section{Introduction}


Two-way communication between two nodes was first studied by Shannon \cite{Shannon}.  There have been many attempts recently to demonstrate two-way communications experimentally \cite{Chen:1998aa,Khandani,Bliss:2007,Radunovic:2009aa,Aryafar,tech_report,Bharadia}. The two-way relay channel where two nodes communicate to each other in the presence of a single relay, has been widely studied \cite{b1,b2,i4,i5,i6,i7,i8,Rankov,s2,s5,s8,s16,s20,s22,Avestim-t2,s4,s12,s13}. In this paper, we will consider the two-way diamond channel, where two nodes communicate to each other in the presence of two relays.

Some achievable rate regions for the two-way relay channel are based on strategies like decode-and-forward, compress-and-forward, and amplify-and-forward \cite{Avestim-t2,s4,s7,s12,s13,s14,s19}. The capacity region of the two-way half-duplex relay channel, where the relay decodes the message is characterized in \cite{Oechtering}. Network coding type techniques have been proposed \cite{Katti,Hausl,Baik} in order to improve the transmission rate. While inferior to traditional routing at low signal-to-noise-ratios (SNR), it was shown that network coding achieves twice the rate of routing at high SNR \cite{Katti2}. The authors of \cite{i7} considered the half-duplex two-way relay channel with unit channel gains, and found that a combination of a decode-and-forward strategy using lattice codes and a joint decoding strategy is asymptotically optimal at high SNR. The authors of \cite{Avestim-t1,Avestim-t2} studied the capacity of the full-duplex two-way relay channel with two users and one relay, and found that a rate within three bits for each user to the capacity can be simultaneously achieved by both users. The result was further extended in \cite{i6,s21}, where lattice codes were used to bring the gap down from three bits to one bit for some special case of channel gains.

The diamond channel was first introduced in \cite{th1}, and consists of one transmitter, two relays and a receiver. In this paper, we study the capacity of the full-duplex two-way diamond channel. In \cite{d1}, several techniques, i.e., amplify-and-forward, hybrid decode-amplify-and-forward with linear combination, hybrid decode-amplify-and-forward with multiplexed coding, decode-and-forward, and partial decode-and-forward, have been considered for achievability in a Gaussian diamond reciprocal channel with half-duplex nodes, and it is shown that these techniques achieve DoF of at most 1, 0.25, 0.25, 0.25, and 0.75, respectively. The two-way half-duplex $K$-relay channel has been studied using the amplify-and-forward strategy at the relays \cite{d2,d5,d4,th2}.

The design of relay beamformers based on minimizing the transmit power subject to the received signal-to-noise ratio constraints was considered in \cite{th2}. Furthermore, achievability schemes using time-sharing are investigated in \cite{d3} for a symmetric reciprocal diamond channel with half-duplex nodes and the inner and outer bounds are compared using simulations. However, we show that the achievability scheme in \cite{d3} has an unbounded gap from the capacity. We note that, to the best of our knowledge, none of the prior works gave a capacity achieving strategy for a two-way diamond channel.

In this paper, we consider a linear deterministic model which was proposed in \cite{Avestimehr}, and has been shown to lead to approximate capacity results for Gaussian channels in \cite{detex,Avestim-t2,Tsedete,Mohajer,J3,dete1,dette,Alireza}. We study the capacity region of a two-way linear deterministic diamond channel where the forward and the backward channel gains are not necessarily the same. We find that the capacity in each direction can be simultaneously achieved. Thus, each user can transmit at a rate which is not affected by the fact that the relays receive the superposition of the signals.

In order to achieve the capacity in each direction separately, we develop new transmission strategies by the transmitters and the relays. The strategies proposed for the one-way diamond channel in \cite{Avestimehr} do not directly work for two-way channels. The reason is that they are dependent on the channel parameters in the forward direction; but for two-way channels we need a strategy that is optimal for both directions. For the special case when the diamond channel reduces to a two-way relay channel (channel gains to and from one of the relays are zeros), our proposed strategy reduces to a reverse amplify-and-forward strategy, where the relay reverses the order of the received signals to form the transmitted signal. The proposed strategy in this case is different from the one in \cite{Avestim-t2} for two-way relay channels, since the relay strategy in  \cite{Avestim-t2} depends on the channel parameters, while ours simply reverses the order of the input. On the other hand, the transmission strategy at the source nodes in our approach is dependent on the channel parameters unlike that in \cite{Avestim-t2}. Thus, the proposed strategy in this paper makes the relay strategy simpler by compensating in the transmission strategy at the source nodes. This proposed simple relay strategy leads to a novel strategy for Gaussian channels. Thus, we extend the achievability scheme to Gaussian channels, and obtain a simpler approach to achieving capacity for a two-way relay channel compared with that in \cite{i6,s21}.


For a general two-way diamond channel, we give different strategies based on the parameters of both the forward and backward channels. Depending on the forward and backward channel gains we consider four cases; these cases are further subdivided. Two special cases are Cases 3.1.2 and 4.1.2. Our first main result is that if neither the forward, nor the backward channel is of one of these two cases, then the proposed reverse amplify-and-forward strategy at the relays is optimal.

We next consider the case that exactly one of the forward and backward channels is of Case 3.1.2 or 4.1.2. Without loss of generality, we assume that the forward channel is of one of the two mentioned cases. For each of these two cases, we give four new strategies at the relay which involve various modifications to the reverse amplify-and-forward strategy, such as repeating some of the streams on multiple levels or changing the order of transmission at some levels at one of the relays. Furthermore, the transmission strategy for the forward direction is rather straightforward by simply sending capacity number of bits at the lowest levels. We show that all these modified strategies achieve the capacity in the forward direction. The choice of the strategies then depends on the parameters in the backward direction. We show that for each case of the backward channel, at least one of the four proposed strategies achieves the capacity for the backward direction. Finally, the case when both the forward and backward channels are of Case 3.1.2 or 4.1.2 is considered. Here, a modified form of the relay strategies proposed above is used to achieve the capacity in both directions.

As an extension to the Gaussian model, first we consider the general Gaussian two-way relay channel and show that the proposed achievability scheme leads to a smaller gap to the cut-set outer bound compared to the previous works \cite{Avestim-t2}. Noting that the treatment for linear deterministic model involves many cases, extending all of the cases of deterministic channel to the Gaussian channel model is challenging. Thus, we consider a special case where the forward and backward channels are Gaussian versions of Case 1 in the linear deterministic model. We take the symmetric case where channel gains from the nodes to each relay are equal and also channel gains from each relay to the nodes are equal. For this special case, under certain conditions, we obtain the achievable rate of each direction that is within four bits of the capacity. The achievability scheme employs lattice codes, and is the first leading to an approximate capacity result for two-way diamond channels.

The remainder of this paper is organized as follows. Section II introduces the model for a two-way linear deterministic diamond channel and presents the main capacity result that shows the capacity in each direction can be achieved. Sections III, IV and V present the proofs for various cases of the channel parameters. Sections VI introduces the model for a two-way Gaussian diamond channel and describes our results on the capacity regions. The results include achieving the capacity of each direction within one bit for a two-way relay channel (if the upper-bound for two directions are equal) and otherwise achieving within one bit for the direction with the lower upper-bound and within two bits for the other direction. The results also include achieving the capacity of each direction within four bits for a special case of two-way diamond channel. Section VII concludes the paper. The detailed proofs of various results in Sections III, IV and V are given in Appendices \ref{apdx_sc1}, \ref{apdx_sc2} and \ref{apdx_sc3}, respectively.

\section{Capacity Region of Deterministic Two-Way Diamond Channel}\label{sec:det}

\subsection{Deterministic Two-Way Diamond Channel Model}

The linear deterministic channel model was proposed in \cite{Avestimehr} to focus on signal interactions instead of the additive noise, and to obtain insights for the Gaussian channel. As shown in  Figure \ref{fig:Example0d}, a two-way diamond channel consists of two nodes (denoted by $A$ and $B$) who wish to communicate to each other through two relays (denoted by $R_1$ and $R_2$). We use non-negative integers $n_{Ak}$, $n_{Bk}$, $n_{kA}$, and $n_{kB}$,  to represent the channel gains from node $A$ to $R_k$, node $B$ to  $R_k$, $R_k$ to node $A$, and $R_k$ to node $B$, respectively, for $k \in\{1, 2\}$. In this paper, the links in the direction from $A$ to $B$ are said to be in the forward direction and those from $B$ to $A$ are in the backward direction.

\begin{figure}[htbp]
\centering
	\includegraphics[width=15cm]{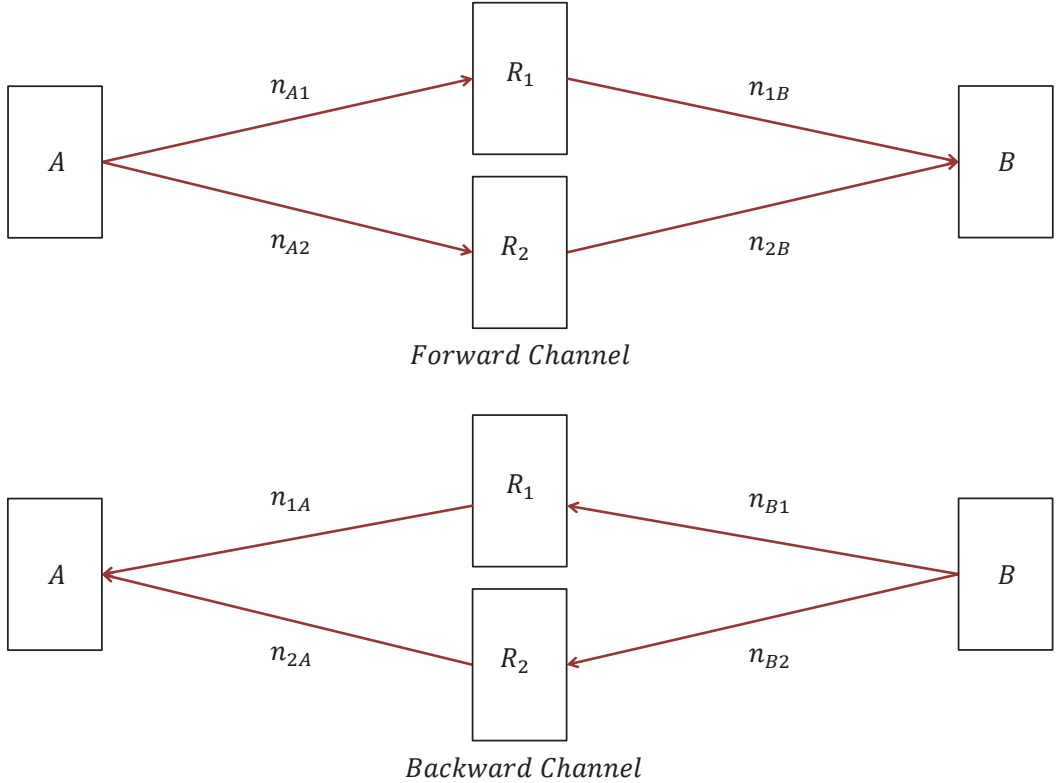}
\caption{A deterministic two-way diamond channel.}
\label{fig:Example0d}
\end{figure}

Let us define $q_{AR}\triangleq\max_k \{n_{Ak}\}$, $q_{RB} \triangleq \max_k \{n_{kB}\}$, $q_{BR}\triangleq\max_k \{n_{Bk}\}$, $q_{RA}\triangleq\max_k \{n_{kA}\}$,  $q^I_k\triangleq\max \{n_{Ak},n_{Bk}\}$, and $q^O_k\triangleq\max \{n_{kA},n_{kB}\}$ for $k \in\{1, 2\}$. Furthermore, denote the channel input at transmitter $u$, for $u\in \{A, B\}$, at time $i$ as $X_{u,i} = [X_{u,i}^{q_{uR}}, \cdots, X_{u,i}^2, X_{u,i}^1]^T \in {\mathsf {F}_2}^{q_{uR}}$, such that $X_{u,i}^1$ and $X_{u,i}^{q_{uR}}$ represent the least and the most significant bits of the transmitted signal, respectively. Also, we define $X^R_{uk,i} = [X_{u,i}^{q_{uR}}, \cdots, X_{u,i}^{q_{uR}-n_{uk}+2}, X_{u,i}^{q_{uR}-n_{uk}+1}, \underset{{q^I_{k}-n_{uk}}}{\underbrace{0,...,0}}]^T$, for $k \in \{1, 2\}$. At each time $i$, the received signal at $R_k$ is given by
\begin{equation}
Y_{k,i}=D_{q^I_{k}}^{q^I_k-n_{Ak}}X^R_{Ak,i}+D_{q^I_{k}}^{q^I_k-n_{Bk}}X^R_{Bk,i} \ \ {\text {mod 2}},
\end{equation}
where $D_{q^I_{k}}$ is a $q^I_{k}\times q^I_{k}$ shift matrix as Eq. (9) in \cite{Avestimehr}. Also if we have $Y_{k,i} = [Y_{k,i}^{q^I_{k}}, \cdots, Y_{k,i}^2, Y_{k,i}^1]^T$, define $V_{k,i} = [0, \cdots, 0, Y_{k,i}^{\min(q^I_{k},q^O_{k})}, \cdots, Y_{k,i}^2, Y_{k,i}^1]^T$, for $k \in \{1, 2\}$,  where the first $(q^O_{k}-q^I_{k})^+$ elements of $V_{k,i}$ are zero.


Furthermore, define $T_{k,i}\triangleq f_{k,i}(V_{k,1},...,V_{k,i-1})$ where $f_{k,i}:{\left(R^{q^O_{k}}\right)}^{i-1}\rightarrow R^{q^O_{k}}$ is a function at $R_k$ which converts $V_{k,1},...,V_{k,i-1}$ to the output signal at time $i$. We represent $T_{k,i}$'s elements as $T_{k,i} = \left[T_{k,i}^1, T_{k,i}^2, \cdots, T_{k,i}^{q^O_{k}}\right]^T$. Also, we define $T^{\prime}_{ku,i} = [T_{k,i}^1, T_{k,i}^2, \cdots, T_{k,i}^{n_{ku}}, \underset{q_{Ru}-n_{ku}}{\underbrace{0,...,0}}]^T$ for $u \in \{A, B\}$. At each time $i$, the received signal at the receivers $u\in \{A, B\}$ is given by
\begin{equation}
Y_{u,i}=\sum_{k=1}^{2} D_{q_{Ru}}^{q_{Ru}-n_{ku}}T^{\prime}_{ku,i} \ \ {\text {mod 2}}.
\end{equation}
Source $u$ picks a message $W_u$ that it wishes to communicate to $\bar{u}$ ($u,\bar{u}\in \{A,B\}$, $u\neq\bar{u}$), and transmits signal at each time $i$ which is a function of $W_u$ and $Y^{i-1}_{u}=\{Y_{u,i-1},Y_{u,i-2},...,Y_{u,1}\}$. Each destination $\bar u$ uses a decoder, which is a mapping $g_{\bar u}: R^m\times |W_{\bar u}|\rightarrow \{1,...,|W_u|\}$ from the $m$ received signals and the message at the receiver to the source message indices ($|W_u|$ is the number of messages of node $u$ that can be chosen). We say that the rate pair ($R_A\triangleq \frac{\log|W_A|}{m}, R_B\triangleq \frac{\log|W_B|}{m}$) is achievable if the probability of error in decoding both messages by their corresponding destinations can be made arbitrarily close to 0 as $m\rightarrow\infty$. The capacity region is the convex hull of all the achievable rate pairs $(R_A,R_B)$.

\subsection{Capacity of Two-Way Linear Deterministic Diamond Channel}

In this subsection, we state the main result that the cut-set bound for the diamond channel in each direction can be simultaneously achieved, thus giving the capacity region for the two-way linear deterministic diamond channel. It can be seen from Figure \ref{fig:Example0d} that $\max\{n_{A1},n_{A2}\}$ and $\max\{n_{1B},n_{2B}\}$ are cut-set bounds on the transmissions from $A$ and to $B$, respectively. Moreover, $n_{A1}+n_{2B}$ and $n_{A2}+n_{1B}$ are cut-set bounds on the sum of the two paths for the transmission from $A$ to $B$. The same observation can be made for the other direction.

\begin{theorem}\label{thm_det}
For the  two-way linear deterministic diamond channel, the capacity region is given as follows:
\begin{eqnarray}\label{r-d}
{R}_{A}&\le& C_{AB} \triangleq \min\{\max\{n_{A1},n_{A2}\},\max\{n_{1B},n_{2B}\},n_{A1}+n_{2B},n_{A2}+n_{1B}\},\label{r-d1}\\
{R}_{B}&\le& C_{BA} \triangleq \min\{\max\{n_{B1},n_{B2}\},\max\{n_{1A},n_{2A}\},n_{B1}+n_{2A},n_{B2}+n_{1A}\}\label{r-d2}.
\end{eqnarray}
\end{theorem}

We note that the outer-bound is the cut-set bound, and thus the proof is straightforward. We will prove the achievability of the rate pair $(C_{AB},C_{BA})$.


We consider four main cases and several subcases depending on the forward channel parameters as follows.

{\bf Case 1:} $C_{AB}=n_{A2}+n_{1B}$.

{\bf Case 2:} $C_{AB}=n_{A1}+n_{2B}$.

{\bf Case 3:} $C_{AB}=\max\{n_{A1},n_{A2}\}$. We call it {\bf Type 1}, if $\max\{n_{A1},n_{A2}\}=n_{A1}$, and {\bf Type 2} otherwise. For {\bf Type i}, where $i,j\in\{1,2\}, i\neq j$, we have:

\indent \indent {\bf Case 3.1:} $n_{iB}< C_{AB}$. We divide it into two sub-cases:

\indent \indent \indent {\bf Case 3.1.1:} $n_{jB}\ge n_{Aj}+n_{iB}$.

\indent \indent \indent {\bf Case 3.1.2:} $n_{jB} < n_{Aj}+n_{iB}$.

\indent \indent {\bf Case 3.2:} $n_{iB}\ge C_{AB}$.

{\bf Case 4:} $C_{AB}=\max\{n_{1B},n_{2B}\}$ We call it {\bf Type 1}, if $\max\{n_{1B},n_{2B}\}=n_{1B}$, and {\bf Type 2} otherwise. For {\bf Type i}, where $i,j\in\{1,2\}, i\neq j$, we have:

\indent \indent {\bf Case 4.1:} $n_{Ai}< C_{AB}$. We divide it into two sub-cases:

\indent \indent \indent {\bf Case 4.1.1:} $n_{Aj}\ge n_{jB}+n_{Ai}$.

\indent \indent \indent {\bf Case 4.1.2:} $n_{Aj}< n_{jB}+n_{Ai}$.

\indent \indent {\bf Case 4.2:} $n_{Ai}\ge C_{AB}$.

Similarly we divide the backward channel into four main cases and several subcases where the case definition is obtained by interchanging $A$ and $B$ in the forward direction cases. For instance, Case 1 in the backward direction is $C_{BA}=n_{B2}+n_{1A}$.

%
%
%
%
%
%
%
%
%
%
%
%


We divide the proof into three parts, depending on the cases in which forward and backward channel gain parameters lie. The first part is when neither the forward channel nor the backward channel is of Case 3.1.2 or 4.1.2 (Section \ref{neitherr}). The second part is when exactly one of the forward and backward channels is of Case 3.1.2 or 4.1.2 (Section \ref{one}). And finally the third part is when both the forward and backward channels are of Case 3.1.2 or 4.1.2 (Section \ref{bothare}).

\section{Neither the forward channel nor backward channel is of Case 3.1.2 or 4.1.2}\label{neitherr}

In this scenario, we use a reverse amplify-and-forward strategy in the relays to achieve the rate pair $(C_{AB}, C_{BA})$. Assume a particular relay (say $R_i$) gets $n_{Ai}$ levels from node $A$ and $n_{Bi}$ levels from node $B$ and transmits $q^O_i$ levels, as shown in Figure \ref{fig:relay} for $n_{Ai}=3$, $n_{Bi}=6$, and $q^O_i=7$. It receives $Y_{A1}={\left[a_{n_{Ai}},...,a_1\right]}^T$ from node $A$ and $Y_{B1}={\left[b_{n_{Bi}},...,b_1\right]}^T$ from node $B$. Then it sends out the following signal to nodes $A$ and $B$
\begin{eqnarray}
X_{R_i} = \left[ \begin{array}{l}
a_{1}\\
\vdots\\
a_\min(n_{Ai},q^O_i)\\
0_{(q^O_i-n_{Ai})^+}
\end{array}\right] + \left[ \begin{array}{l}
b_{1}\\
\vdots\\
b_\min(n_{Bi},q^O_i)\\
0_{(q^O_i-n_{Bi})^+}
\end{array}\right] \ \ {\text {mod 2}}.
\end{eqnarray}
\begin{figure}[htbp]
\centering
	\includegraphics[width=8cm]{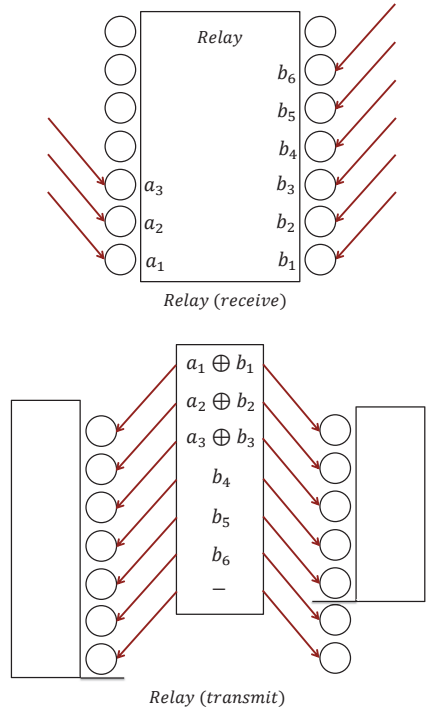}
\caption{Reverse amplify-and-forward as a two-way relay function.}
\label{fig:relay}
\end{figure}
%
%
%
%
%
%
We call this relay strategy as ``Relay Strategy 0'' (also called reverse amplify-and-forward). We will keep the strategy at the relays the same, and for different cases use different strategies for transmission at nodes $A$ and $B$. Since we need to show that the rate pair $(C_{AB},C_{BA})$ is achievable, it is enough to show that there is a transmission strategy for node $A$ such that with the above relay strategy, node $B$ is able to decode the data in a one-way diamond channel because any interference by node $B$ on the received signal can be canceled by node $B$ which knows the interfering signal (Showing it for one direction is enough since the same arguments hold for the other). Thus, we only consider one-way diamond channel for this case. We further consider the case when $n_{A1}, n_{A2}, n_{1B}, n_{2B}>0$ since otherwise the diamond channel reduces to a relay channel or no connection between the nodes $A$ and $B$, and in both cases it is easy to see that node $A$ sending $C_{AB}$ bits on the lowest levels achieves this rate in the forward direction.

Appendix \ref{apdx_sc1} proves that there is a transmission strategy for each of the cases (except for Case 3.1.2 or 4.1.2) such that the above relay strategy achieves the capacity for one-way diamond channel.

\begin{example}
Consider the case $(n_{A1},n_{A2},n_{1B},n_{2B},n_{B1},n_{B2},n_{1A},n_{2A})=(6,2,3,7,6,3,4,8)$. With these parameters, the forward channel is of Case 1, and the backward channel is of Case 3.1.1 Type 1. We use the transmission strategies corresponding to these cases given in Appendix \ref{apdx_sc1}, and shown in Figure \ref{fig:Example1.1} that the desired messages can be decoded by both nodes $A$ and $B$.
\end{example}

\begin{figure}[htbp]
\centering
\subfigure[Transmission to relays.]{
	\includegraphics[width=8.5cm]{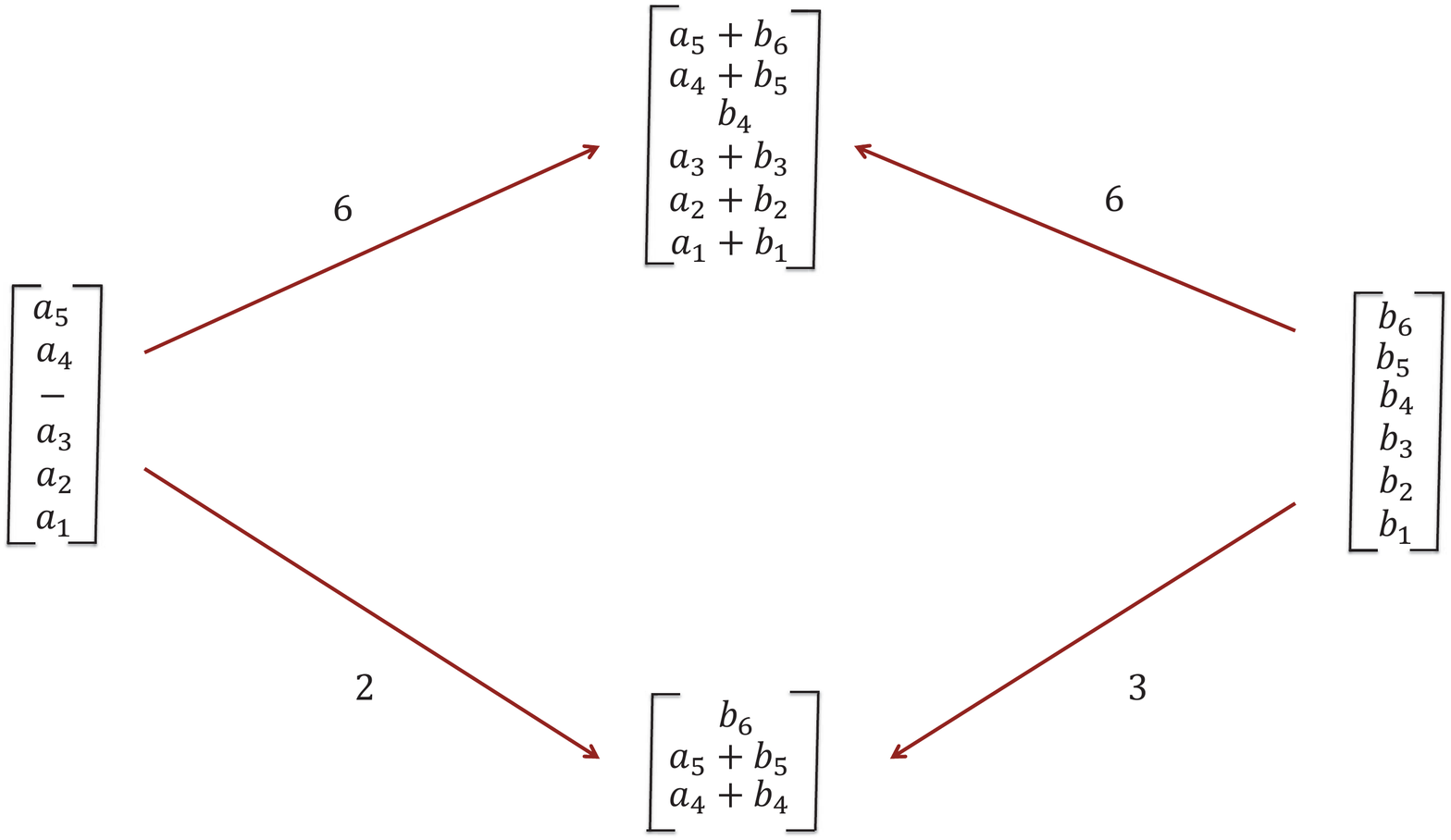}
\label{fig:subfig111}
}
\subfigure[Reception from relays.]{
	\includegraphics[width=8.5cm]{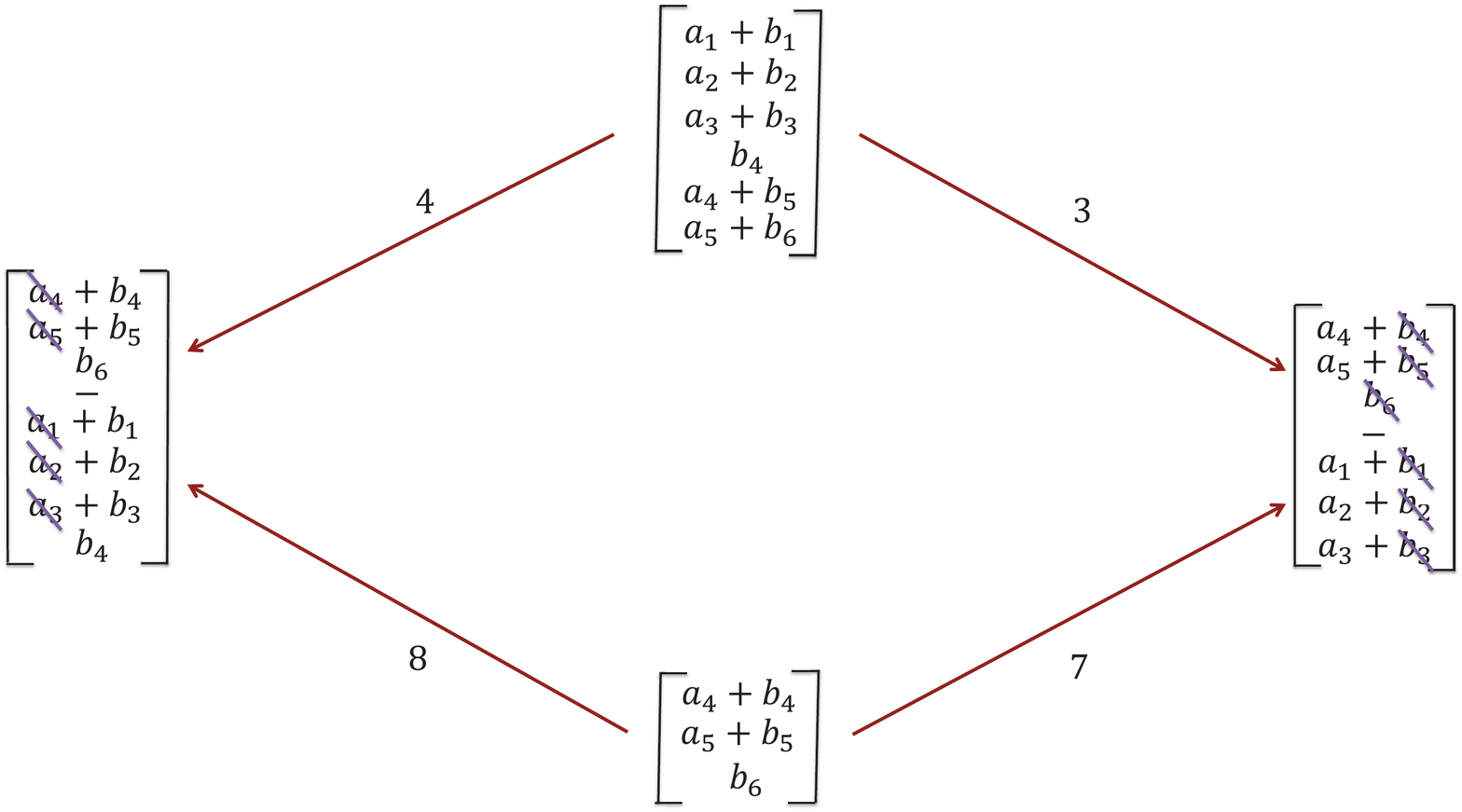}
\label{fig:subfig222}
}
\caption[Optional caption for list of figures]{Example for $(n_{A1},n_{A2},n_{1B},n_{2B},n_{B1},n_{B2},n_{1A},n_{2A})=(6,2,3,7,6,3,4,8)$.}
\label{fig:Example1.1}
\end{figure}

\section{Exactly one of the forward and backward channels is of Case 3.1.2 or 4.1.2}\label{one}

We assume that the forward channel is of Case 3.1.2 or 4.1.2 without loss of generality. The other case where the backward channel is of Case 3.1.2 or 4.1.2 can be proven symmetrically. Since we need to show that the rate pair $(C_{AB},C_{BA})$ is achievable, we will describe a few relay strategies for which the same transmission strategy is used at node $A$ such that node $B$ is able to decode the corresponding message. Furthermore, we will show that at least one of these strategies is optimal for the backward channel for each case of the backward channel parameters. As before we consider the case when $n_{A1}, n_{A2}, n_{1B}, n_{2B}>0$. In the remainder of this section, we assume that the forward channel is of Case 3.1.2. The case that the forward channel is of Case 4.1.2 is treated in Appendix \ref{apdx_sc2}.





When the forward channel is of Case 3.1.2, node $A$ uses the same transmission strategy as Case 3.1.1 in Appendix \ref{apdx_sc1}, i.e., it transmits ${[a_{C_{AB}},...,a_{1}]}^T$. Also, the transmission strategy for node $B$ depends on the channel gains in the backward direction of the channel, and is the same as that used in Appendix \ref{apdx_sc1} for each set of parameters.

For the relay strategy, we will choose one of the four strategies explained in the following depending on the backward channel parameters. We will prove that all of these strategies are optimal for the forward channel for any set of parameters.

The parameters associated with each relay strategy proposed here are only based on the forward channel gains, and we will show that at least one of the proposed strategies is optimal for each choice of the backward channel parameters. Note that using Relay Strategy 0 in both relays, node $B$ cannot necessarily decode the message if the forward channel is of Case 3.1.2 or 4.1.2, when the above transmission strategy is used by node $A$. An example is illustrated in Figure \ref{fig:Exampleno} when the forward channel is of Case 3.1.2 and the parameters are $n_{A1}=4$, $n_{A2}=3$, $n_{1B}=3$ and $n_{2B}=5$.

\begin{figure}[htbp]
\centering
	\includegraphics[width=15cm]{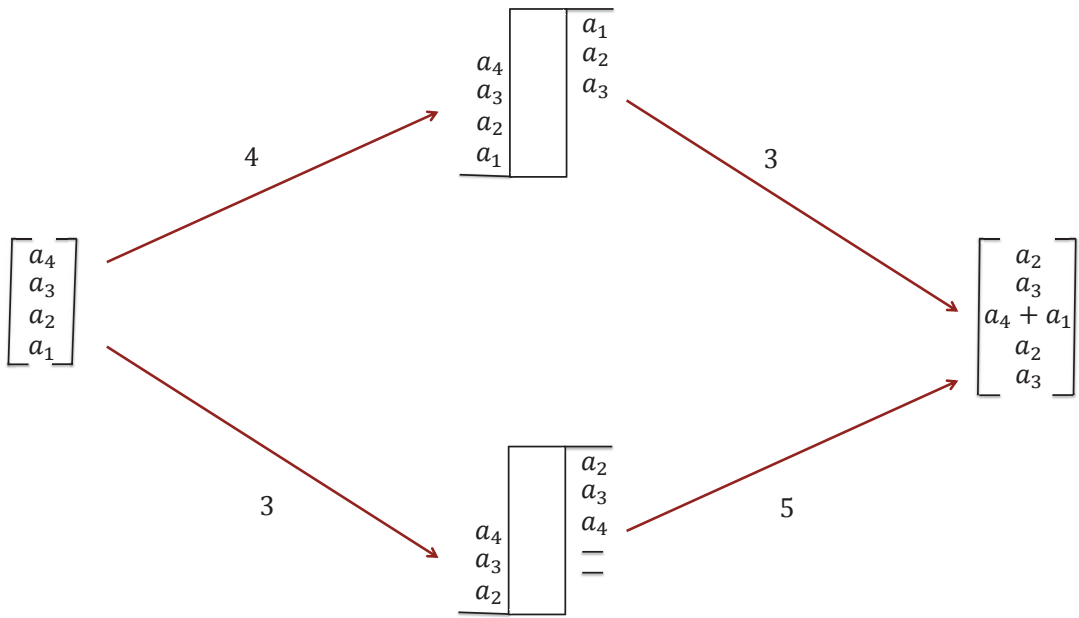}
\caption{An example that Relay Strategy 0 does not work.}
\label{fig:Exampleno}
\end{figure}

\begin{remark} All relay strategies in this subsection and in Appendix \ref{apdx_sc2}, are defined with respect to the forward channel parameters (and in favor of the forward channel direction\footnote{In the sense that the strategies are designed so  that the forward communication achieves the capacity.}) because we assumed that the forward channel is either of Case 3.1.2 or 4.1.2 and the backward channel is not of these cases. We note that Relay Strategy 0 is symmetric and is not dependent on the channel gains in any direction. In Section \ref{bothare}, we will generalize some of these strategies to be based on the parameters of both the forward and backward channels.
\end{remark}


\subsection{ \bf Relay Strategy 1:}

If the forward channel is of Case 3.1.2 Type $i$, then Relay Strategy 0 is used at $R_i$, and Relay Strategy 1 is used at $R_{\bar{i}}$, where $i,\bar{i}\in\{1,2\}, i\neq \bar{i}$. Here, we define Relay Strategy 1 at $R_2$ (forward channel of Case 3.1.2 Type $1$), while that for $R_1$ can be obtained by interchanging roles of relays $R_1$ and $R_2$ (interchanging 1 and 2 and forward channel of Case 3.1.2 Type $2$). As shown in Figure \ref{fig:rs1}, if $R_2$ receives a block of $n_{2B}$ bits, first it will reverse them as in Relay Strategy 0 and then changes the order of the first $n_{1B}-(n_{A1}-n_{A2})$ streams\footnote{In the following relay strategies, we divide the streams into multiple sub-streams. The number of streams in each sub-stream is a non-negative number when the forward channel is of Case 3.1.2 Type 1.} with the next $n_{A1}-n_{1B}$ streams.

\begin{figure}[htbp]
\centering
	\includegraphics[width=10cm]{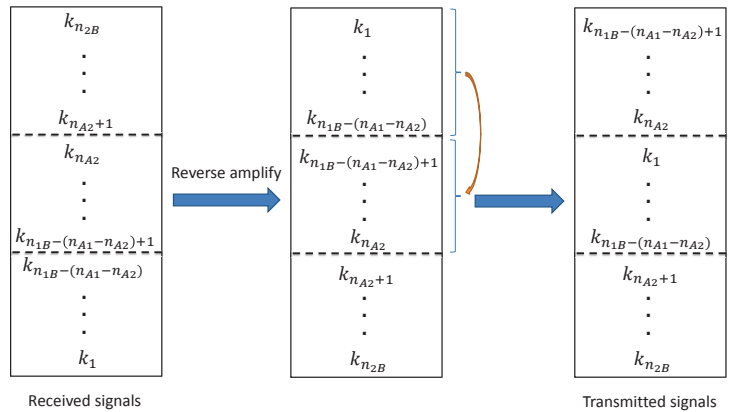}
\caption{Relay Strategy 1 at $R_2$.}
\label{fig:rs1}
\end{figure}

Node $A$ transmits ${[a_{C_{AB}},...,a_{1}]}^T$. The received signals can be seen in Figure \ref{fig:s1s}. We use $(R_i, B_j)$ to denote block number $j$ from $R_i$. Bits that are not delivered to node $B$ from $R_1$ using Relay Strategy 0, ($a_{n_{1B}+1}, ..., a_{n_{A1}}$), are all sent at the highest levels from $R_2$ to node $B$ and thus are decoded with no interference (block $(R_2,B_1)$). The remaining bits can be decoded by starting from the lowest level of reception in $B$ ($a_{n_{1B}}$ in block $(R_1,B_4)$) and removing the effect of the decoded bits and going up.

\begin{figure}[htbp]
\centering
	\includegraphics[width=15cm]{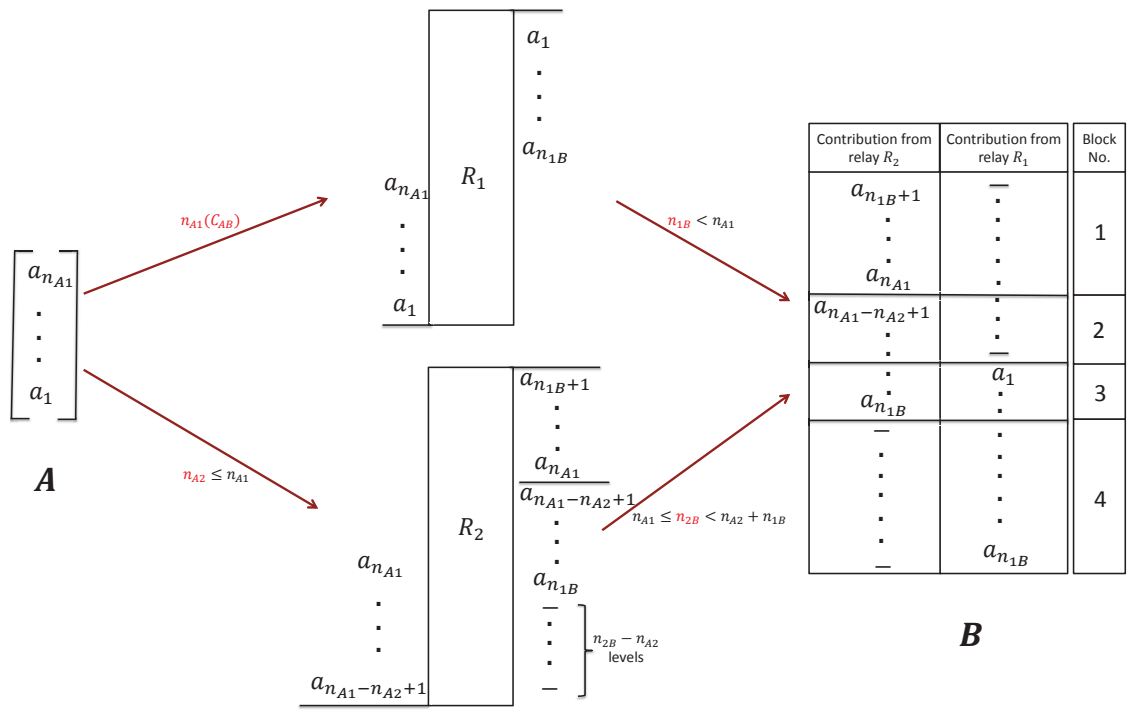}
\caption{Received signals by using Relay Strategy 1 when the forward channel is of Case 3.1.2 Type 1.}
\label{fig:s1s}
\end{figure}

\begin{example}
Consider the case $(n_{A1},n_{A2},n_{1B},n_{2B},n_{B1},n_{B2},n_{1A},n_{2A})=(6,4,5,7,6,5,1,7)$. With these parameters, the forward channel is of Case 3.1.2 Type 1, and the backward channel is of Case 3.1.1 Type 1. For the backward channel, we use the transmission strategy corresponding to Case 3.1.1 given in Appendix \ref{apdx_sc1} (transmit ${[b_{C_{BA}}, \cdots, b_1]}^T$) and for the forward channel, we transmit ${[a_{C_{AB}},...,a_{1}]}^T$, as explained at the beginning of this section. Also, $R_1$ uses Relay Strategy 0, and  $R_2$ uses Relay Strategy 1. Figure \ref{fig:Example.S.1} illustrates that the desired messages can be decoded by both nodes $A$ and $B$.
\end{example}

\begin{figure}[htbp]
\centering
\subfigure[Transmission to relays.]{
	\includegraphics[width=8.5cm]{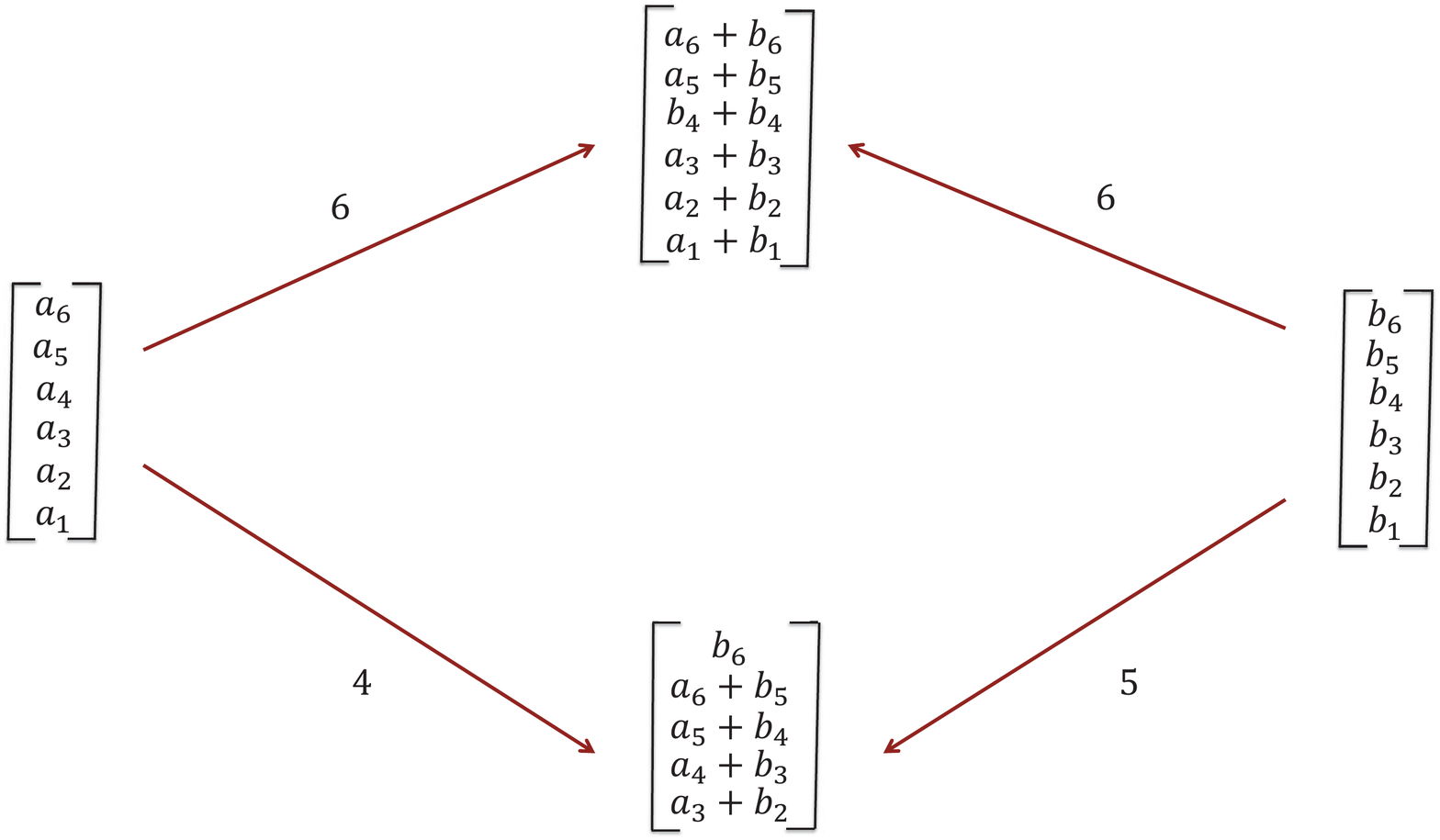}
\label{fig:subfigS11}
}
\subfigure[Reception from relays.]{
	\includegraphics[width=8.5cm]{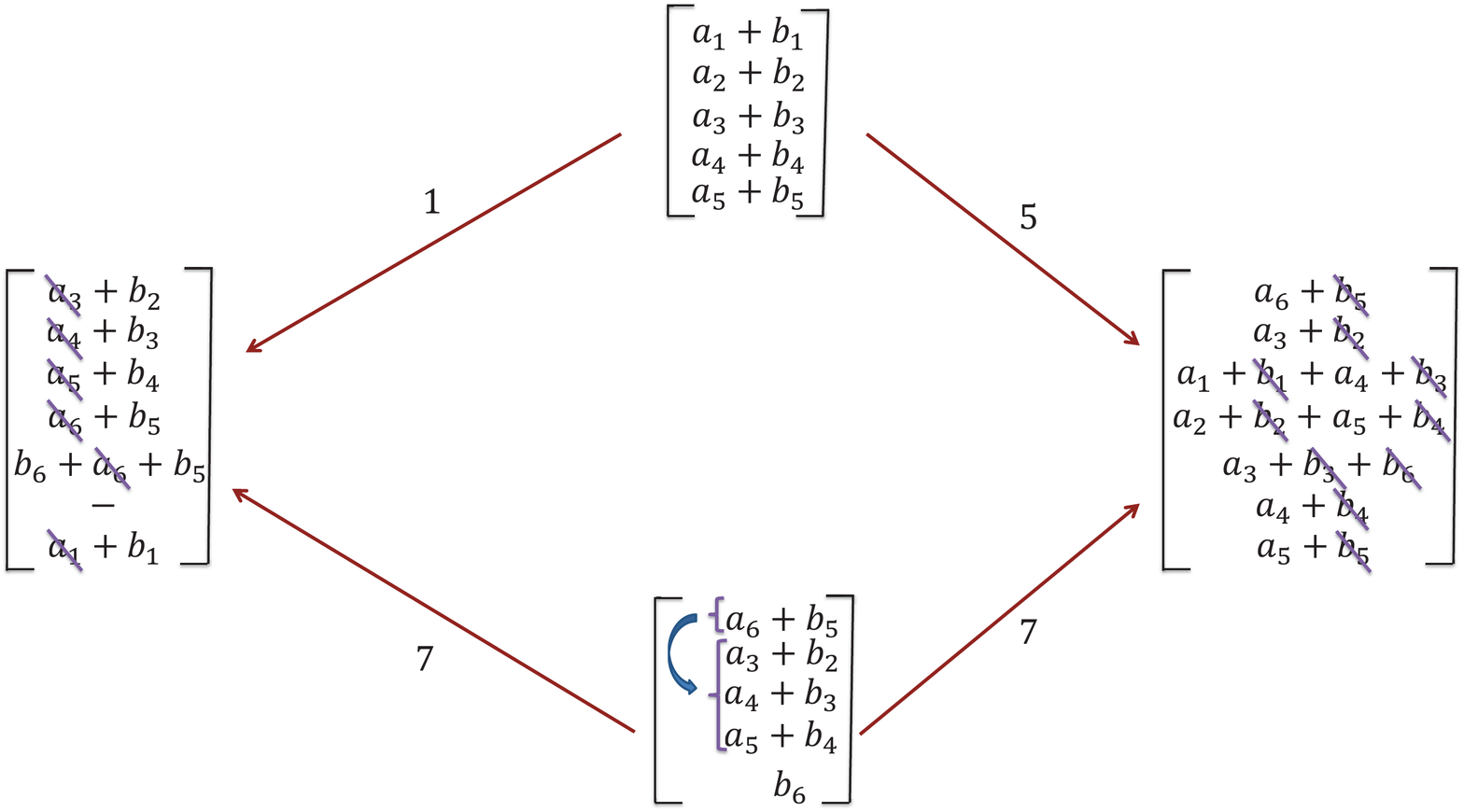}
\label{fig:subfigS12}
}
\caption[Optional caption for list of figures]{Example for $(n_{A1},n_{A2},n_{1B},n_{2B},n_{B1},n_{B2},n_{1A},n_{2A})=(6,4,5,7,6,5,1,7)$ using Relay Strategy 1.}
\label{fig:Example.S.1}
\end{figure}

\subsection{\bf Relay Strategy 2:}

If the forward channel is of Case 3.1.2 Type $i$, then Relay Strategy 0 is used at $R_i$, and Relay Strategy 2 is used at $R_{\bar{i}}$, where $i,\bar{i}\in\{1,2\}, i\neq \bar{i}$. Here, we define Relay Strategy 2 at $R_2$ (forward channel of Case 3.1.2 Type $1$), while that for $R_1$ can be obtained by interchanging roles of $R_1$ and $R_2$ (interchanging 1 and 2 and forward channel of Case 3.1.2 Type $2$). It is similar to Relay Strategy 0 with the only difference that $R_2$ repeats a part of the top $n_{A2}$ streams after reverse-amplify-and-forward, as explained below in nine separate scenarios, based on the parameters of the forward channel. We note that the repetition of streams is based on the received signal at the relay. However, we describe below only the forward direction to show that the messages can be decoded.


\begin{figure}[htbp]
\centering
	\includegraphics[width=10cm]{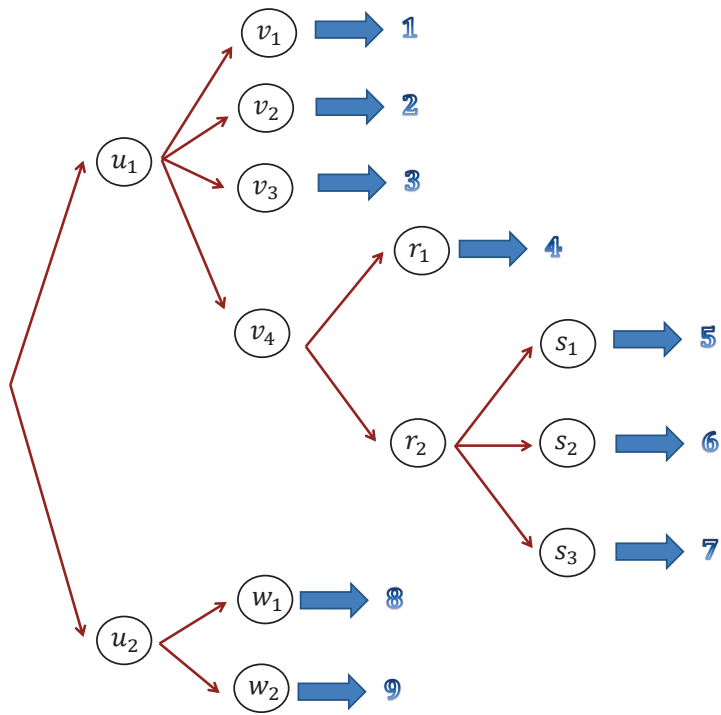}
\caption{Dividing the 4-dimensional space consisting of $(n_{A1},n_{A2},n_{1B},n_{2B})$ into nine subspaces.}
\label{fig:sup1}
\end{figure}

As shown in Figure \ref{fig:sup1}, we define the partition of the four-dimensional space $(n_{A1},n_{A2},n_{1B},n_{2B})$ into nine parts that lead to different received signal structures in node $B$, as shown in Figures \ref{fig:1}-\ref{fig:6}, respectively. Specifically, $\{u_1,u_2\}=\{n_{2B}+(n_{A1}-n_{A2}) \le n_{A2}+n_{1B},n_{A2}+n_{1B}< n_{2B}+(n_{A1}-n_{A2})\}$, $\{v_1,v_2,v_3,v_4\}=\{n_{1B}\le(n_{A1}-n_{A2})+(n_{2B}-n_{1B}),n_{1B}-(n_{A1}-n_{A2})\le(n_{A1}-n_{A2})+(n_{2B}-n_{1B})< n_{1B},n_{A2}-(n_{2B}-n_{1B})\le(n_{A1}-n_{A2})+(n_{2B}-n_{1B})< n_{1B}-(n_{A1}-n_{A2}),(n_{A1}-n_{A2})+(n_{2B}-n_{1B})< n_{A2}-(n_{2B}-n_{1B})\}$, $\{w_1,w_2\}=\{n_{1B}-(n_{A1}-n_{A2}) \le n_{2B}-n_{1B},n_{1B}-(n_{A1}-n_{A2}) > n_{2B}-n_{1B}\}$, $\{r_1,r_2\}=\{2(2(n_{1B}-n_{2B}+n_{A2})-n_{A1})+n_{2B}-n_{A1}\le 2 n_{A2}-n_{A1}+n_{1B}-n_{2B},2(2(n_{1B}-n_{2B}+n_{A2})-n_{A1})+n_{2B}-n_{A1}> 2 n_{A2}-n_{A1}+n_{1B}-n_{2B}\}$ and $\{s_1,s_2,s_4\}=\{n_{A1}-n_{A2}\ge n_{1B}-2(n_{A1}-n_{A2}+n_{2B}-n_{1B}),n_{2B}-n_{A2}\ge n_{1B}-2(n_{A1}-n_{A2}+n_{2B}-n_{1B})> n_{A1}-n_{A2},n_{1B}-2(n_{A1}-n_{A2}+n_{2B}-n_{1B})> n_{2B}-n_{A2}\}$.

\begin{enumerate}
  \item $(u_1,v_1)$: Figure \ref{fig:1} depicts the received signal at node $B$ (ignoring the effect of transmitted signal from $B$) assuming that both relays use Relay Strategy 0. The repetitions will be described below to show that messages can be decoded with the proposed strategies.
\begin{figure}[htbp]
\centering
	\includegraphics[width=8cm]{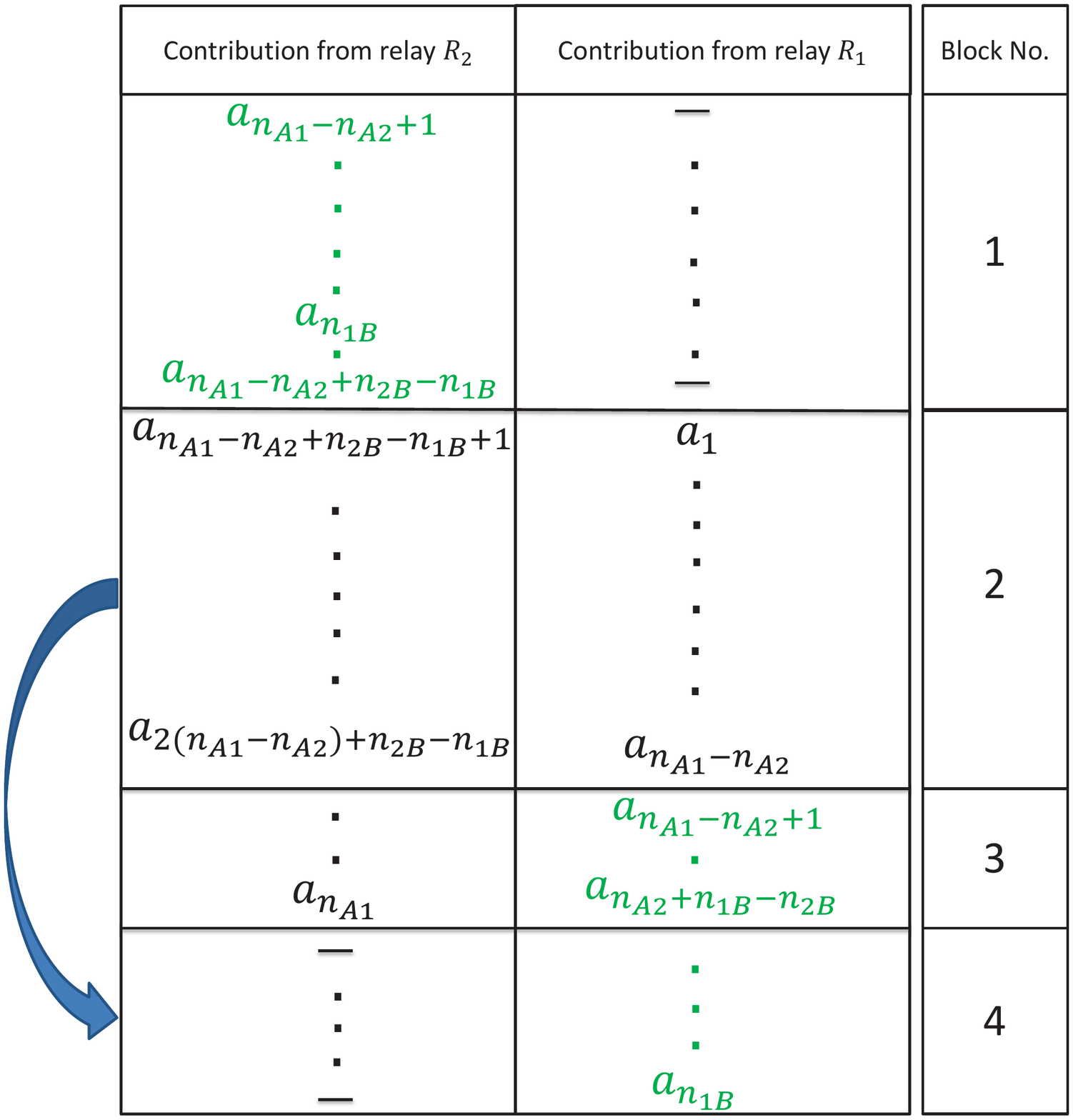}
\caption{The received signals at node $B$ (ignoring the effect of transmitted signal from $B$) assuming that both relays use Relay Strategy 0 for channel parameters of case $(u_1,v_1)$.}
\label{fig:1}
\end{figure}

      $R_2$ repeats the streams in block $(R_2,B_2)$ on block $(R_2,B_4)$. Using this strategy, block $(R_2,B_1)$ will be decoded from the top levels of the received signal from $R_2$ since there is no interference from the other relay. Then, subtract the corresponding signals (blocks $(R_1,B_3)$ and $(R_1,B_4)$). Furthermore, block $(R_2,B_4)$ can be decoded from repetitions because their interference is already decoded. Then, subtract the corresponding signals (block $(R_2,B_2)$). Consequently, block $(R_1,B_2)$ are decoded because their interference (block $(R_2,B_2)$) was decoded earlier. Finally, block $(R_2,B_3)$ can be decoded because all its interference signals have been decoded.

  \item $(u_1,v_2)$: As shown in Figure \ref{fig:2},
\begin{figure}[htbp]
\centering
	\includegraphics[width=8cm]{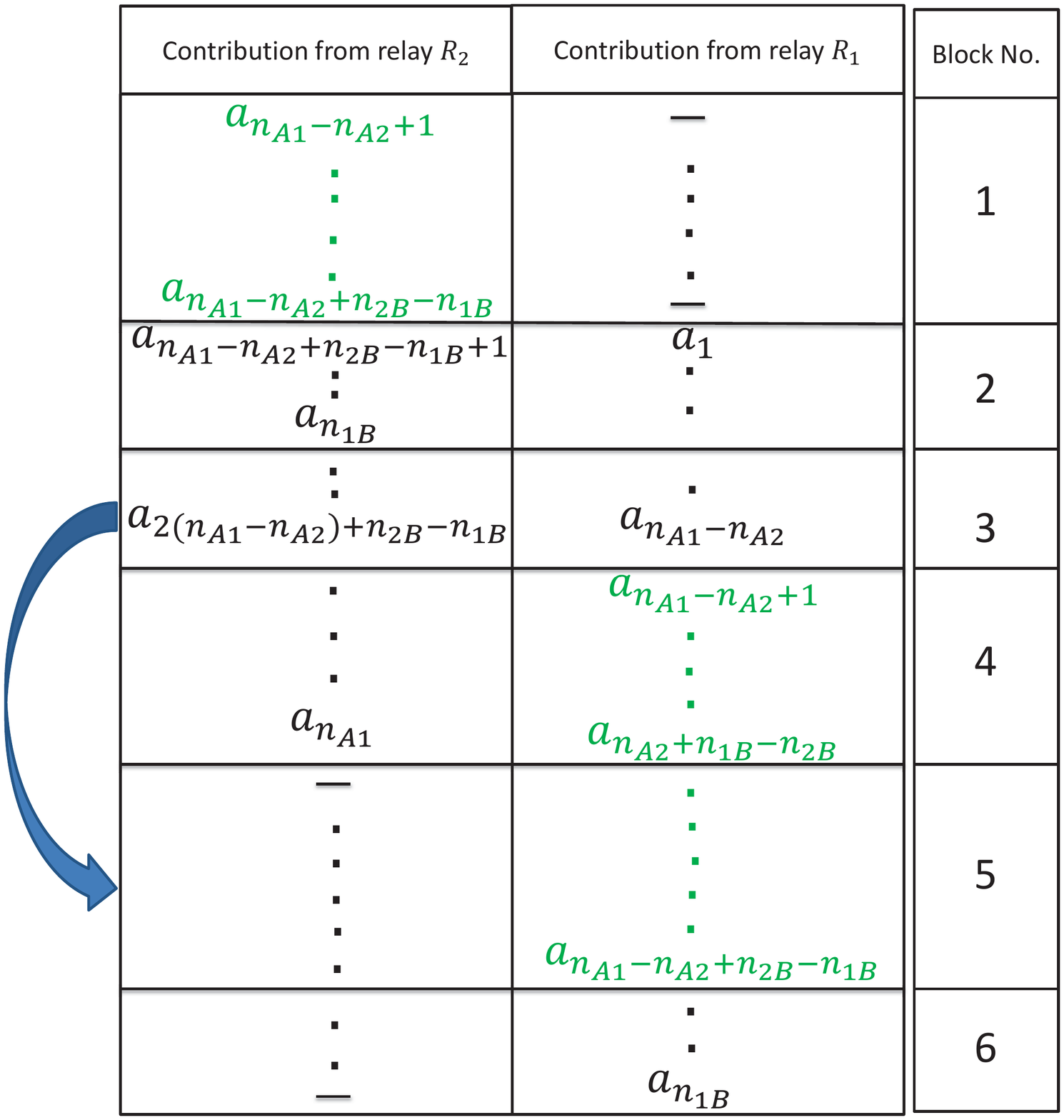}
\caption{The received signals at node $B$ (ignoring the effect of transmitted signal from $B$) assuming that both relays use Relay Strategy 0 for channel parameters of case $(u_1,v_2)$.}
\label{fig:2}
\end{figure}
      $R_2$ repeats block $(R_2, B_3)$ on block $(R_2, B_5)$. The decoding order is
      \begin{eqnarray}
&&{\text {decode \& subtract\ }} (R_2,B_1) \rightarrow {\text {subtract\ }} (R_1,B_4) \& (R_1,B_5) \rightarrow {\text {decode \& subtract\ }} (R_1,B_6) \rightarrow \nonumber\\
&&{\text {decode \& subtract\ }} (R_2,B_5) \rightarrow {\text {subtract\ }} (R_2,B_3) \rightarrow {\text {decode \& subtract\ }}(R_1,B_3) \rightarrow\nonumber\\
&& {\text {subtract\ }} (R_2,B_2) \rightarrow {\text {decode \& subtract\ }} (R_1,B_2).\nonumber
\end{eqnarray}

  \item $(u_1,v_3)$: As shown in Figure \ref{fig:3},
\begin{figure}[htbp]
\centering
	\includegraphics[width=8cm]{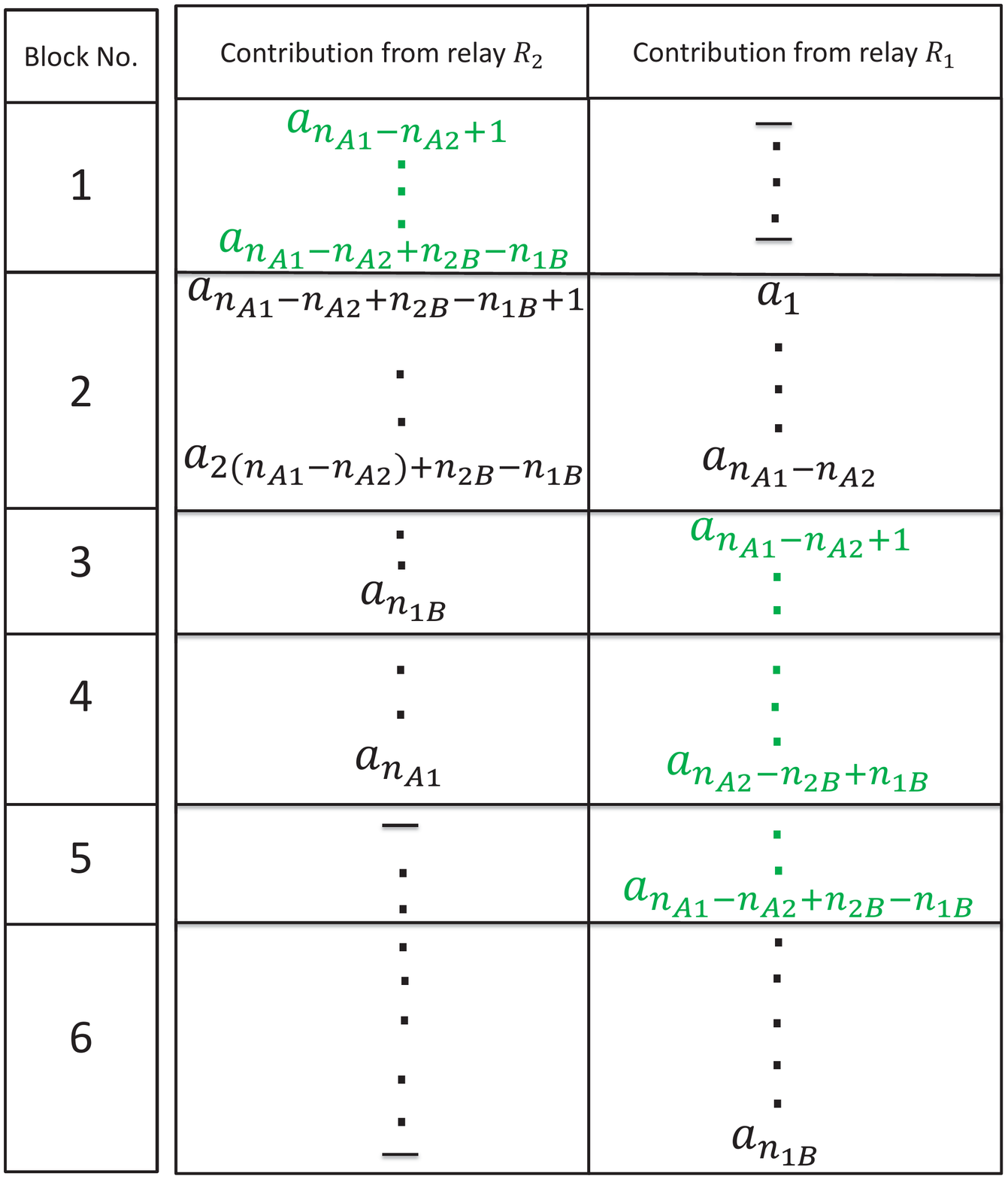}
\caption{The received signals at node $B$ (ignoring the effect of transmitted signal from $B$) assuming that both relays use Relay Strategy 0 for channel parameters of case $(u_1,v_3)$.}
\label{fig:3}
\end{figure}
    this case does not need repetition. The decoding order is
      \begin{eqnarray}
&&{\text {decode \& subtract\ }} (R_2,B_1) \rightarrow {\text {subtract\ }} (R_1,B_3) \& (R_1,B_4) \& (R_1,B_5) \rightarrow \nonumber\\
&&{\text {decode \& subtract\ }} (R_2,B_3) \& (R_2,B_4) \rightarrow {\text {decode \& subtract\ }} (R_1,B_6) \rightarrow {\text {subtract\ }} (R_2,B_2) \rightarrow \nonumber\\
&& {\text {decode \& subtract\ }} (R_1,B_2).\nonumber
\end{eqnarray}

  \item $(u_1,v_4,r_1)$: As shown in Figure \ref{fig:4.1},
\begin{figure}[htbp]
\centering
	\includegraphics[width=8cm]{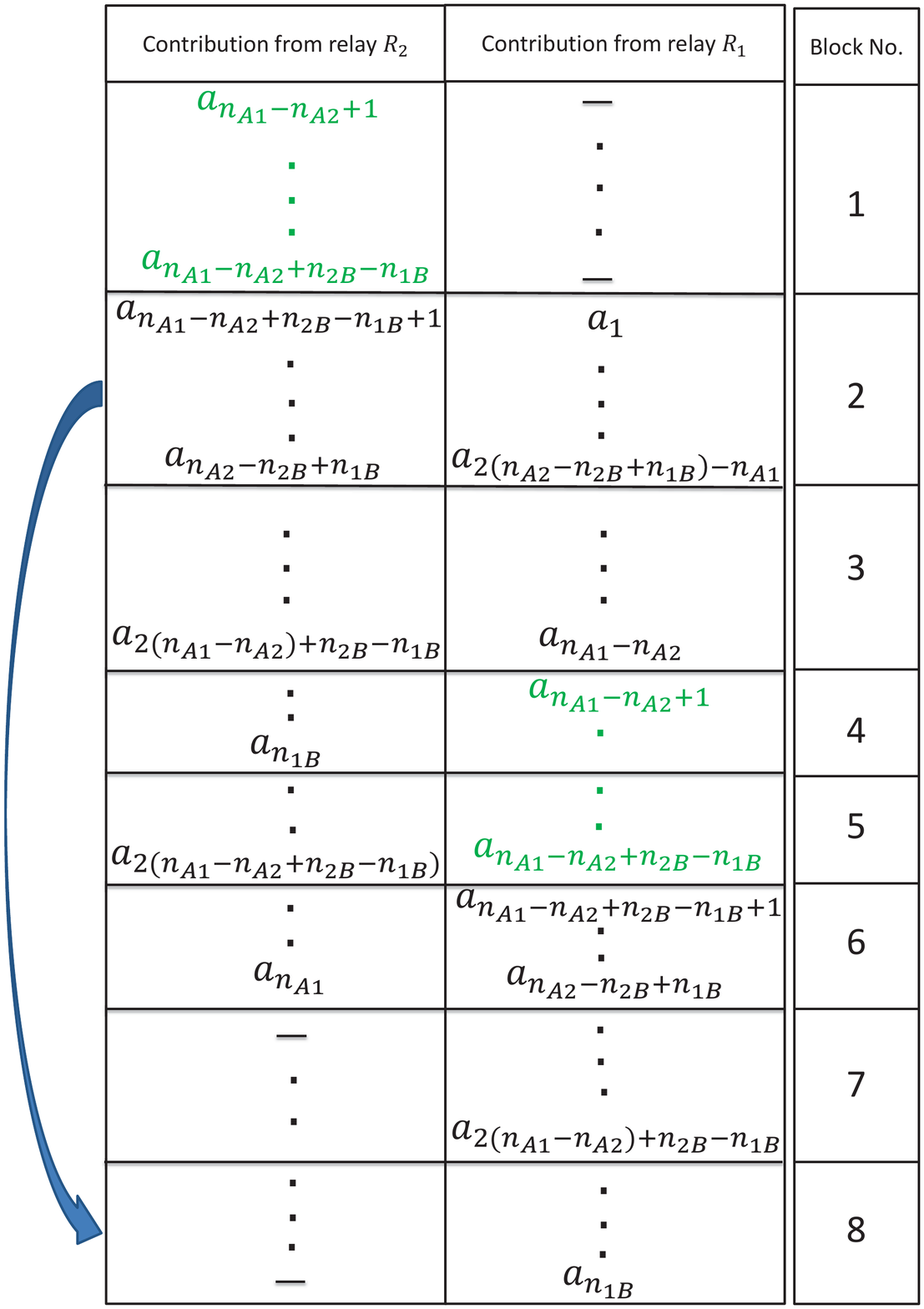}
\caption{The received signals at node $B$ (ignoring the effect of transmitted signal from $B$) assuming that both relays use Relay Strategy 0 for channel parameters of case $(u_1,v_4,r_1)$.}
\label{fig:4.1}
\end{figure}
      $R_2$ repeats block $(R_2, B_2)$ on block $(R_2, B_8)$. The decoding order is
      \begin{eqnarray}
&& {\text {decode \& subtract\ }} (R_2,B_1) \rightarrow {\text {subtract\ }}  (R_1,B_4) \& (R_1,B_5) \rightarrow \nonumber\\
&& {\text {decode \& subtract\ }} (R_2,B_4) \& (R_2,B_5) \rightarrow  {\text {decode \& subtract\ }} (R_1,B_8) \rightarrow \nonumber\\
&& {\text {decode \& subtract\ }} (R_2,B_8) \rightarrow {\text {subtract\ }}(R_2,B_2) \& (R_1,B_7) \rightarrow  {\text {decode \& subtract\ }} (R_2,B_6) \rightarrow \nonumber\\
&& {\text {decode \& subtract\ }} (R_1,B_7) \rightarrow {\text {subtract\ }} (R_2,B_3)\rightarrow {\text {decode \& subtract\ }} (R_1,B_2) \& (R_1,B_3).\nonumber
\end{eqnarray}

  \item $(u_1,v_4,r_2,s_1)$: As shown in Figure \ref{fig:4.21},
\begin{figure}[htbp]
\centering
	\includegraphics[width=8cm]{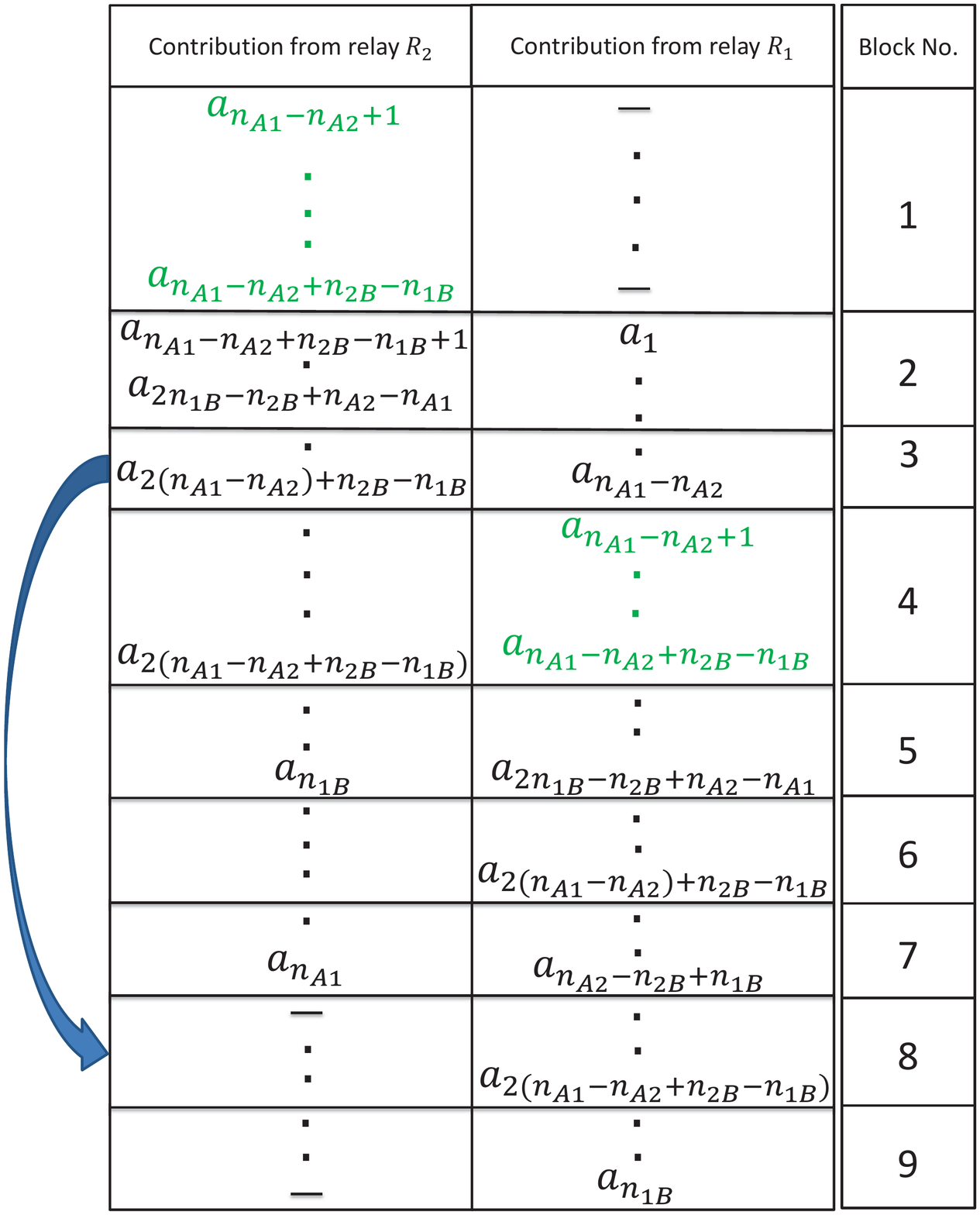}
\caption{The received signals at node $B$ (ignoring the effect of transmitted signal from $B$) assuming that both relays use Relay Strategy 0 for channel parameters of case $(u_1,v_4,r_2,s_1)$.}
\label{fig:4.21}
\end{figure}
      $R_2$ repeats block $(R_2, B_3)$ on block $(R_2, B_8)$. The decoding order is
      \begin{eqnarray}
&& {\text {decode \& subtract\ }} (R_2,B_1) \rightarrow {\text {subtract\ }}(R_1,B_4) \rightarrow {\text {decode \& subtract\ }} (R_2,B_4) \rightarrow \nonumber\\
&& {\text {decode \& subtract\ }} (R_1,B_7) \& (R_1,B_8) \rightarrow  {\text {decode \& subtract\ }} (R_1,B_9) \rightarrow {\text {subtract\ }} (R_2,B_5) \rightarrow \nonumber\\
&& {\text {decode \& subtract\ }} (R_2,B_8) \rightarrow {\text {subtract\ }} (R_2,B_3)\& (R_1,B_6) \rightarrow {\text {decode \& subtract\ }} (R_1,B_5) \rightarrow \nonumber\\
&& {\text {subtract\ }} (R_2,B_2)\rightarrow  {\text {decode \& subtract\ }} (R_2,B_6) \& (R_2,B_7).\nonumber
\end{eqnarray}

  \item $(u_1,v_4,r_2,s_2)$: As shown in Figure \ref{fig:4.22},
\begin{figure}[htbp]
\centering
	\includegraphics[width=8cm]{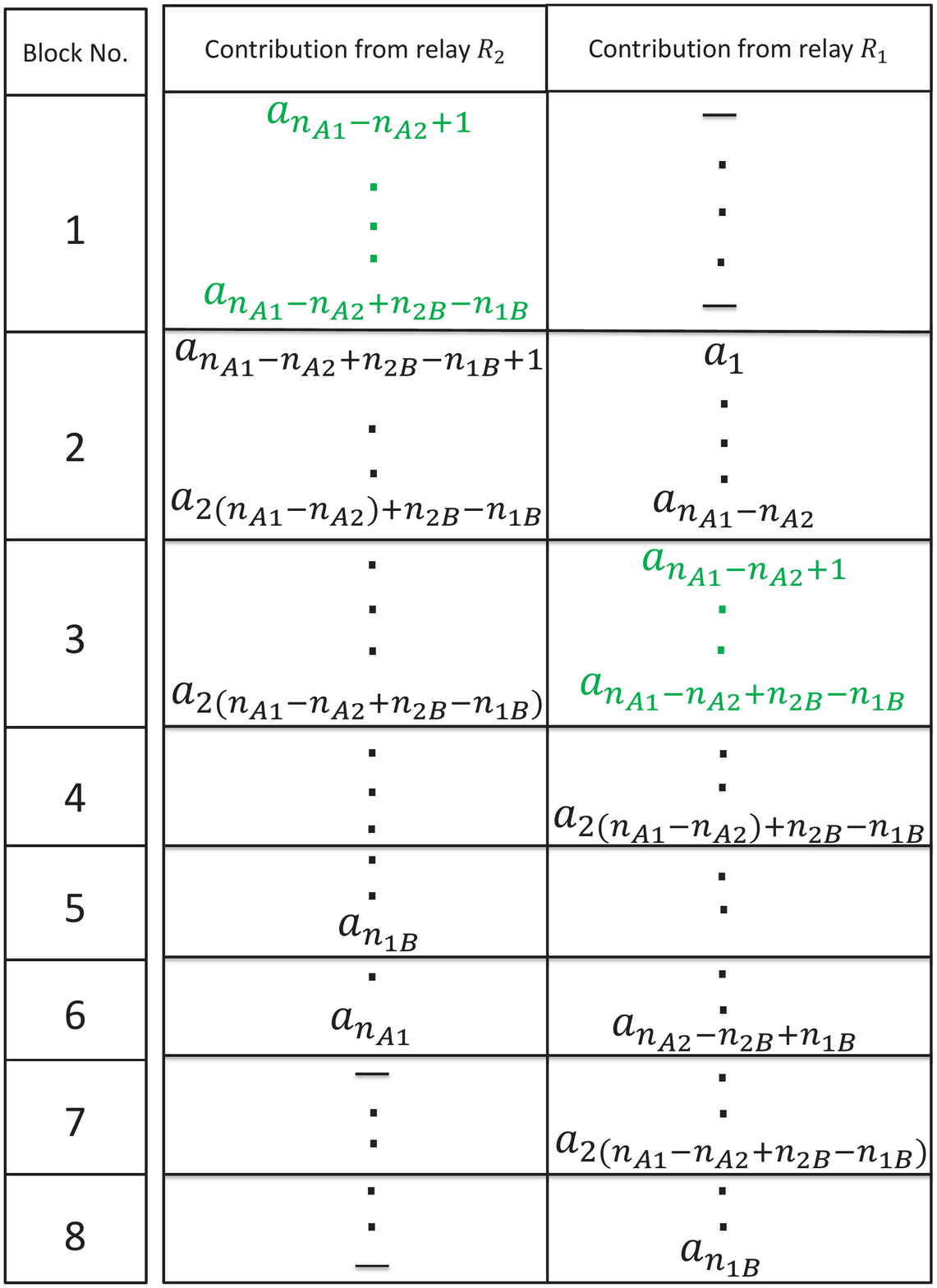}
\caption{The received signals at node $B$ (ignoring the effect of transmitted signal from $B$) assuming that both relays use Relay Strategy 0 for channel parameters of case $(u_1,v_4,r_2,s_2)$.}
\label{fig:4.22}
\end{figure}
      this case does not need repetition. The decoding order is
      \begin{eqnarray}
&&{\text {decode \& subtract\ }} (R_2,B_1) \rightarrow {\text {subtract\ }} (R_1,B_3) \rightarrow \nonumber\\
&&{\text {decode \& subtract\ }}(R_2,B_3) \& (R_1,B_7) \& (R_1,B_8) \rightarrow {\text {subtract\ }} (R_1,B_5) \& (R_1,B_6) \& (R_2,B_4) \rightarrow \nonumber\\
&& {\text {decode \& subtract\ }} (R_2,B_5) \& (R_2,B_6) \rightarrow {\text {decode \& subtract\ }} (R_1,B_7) \rightarrow {\text {subtract\ }} (R_2,B_2) \rightarrow \nonumber\\
&&{\text {decode \& subtract\ }}(R_1,B_2).\nonumber
\end{eqnarray}

  \item $(u_1,v_4,r_2,s_3)$: As shown in Figure \ref{fig:4.23},
\begin{figure}[htbp]
\centering
	\includegraphics[width=8cm]{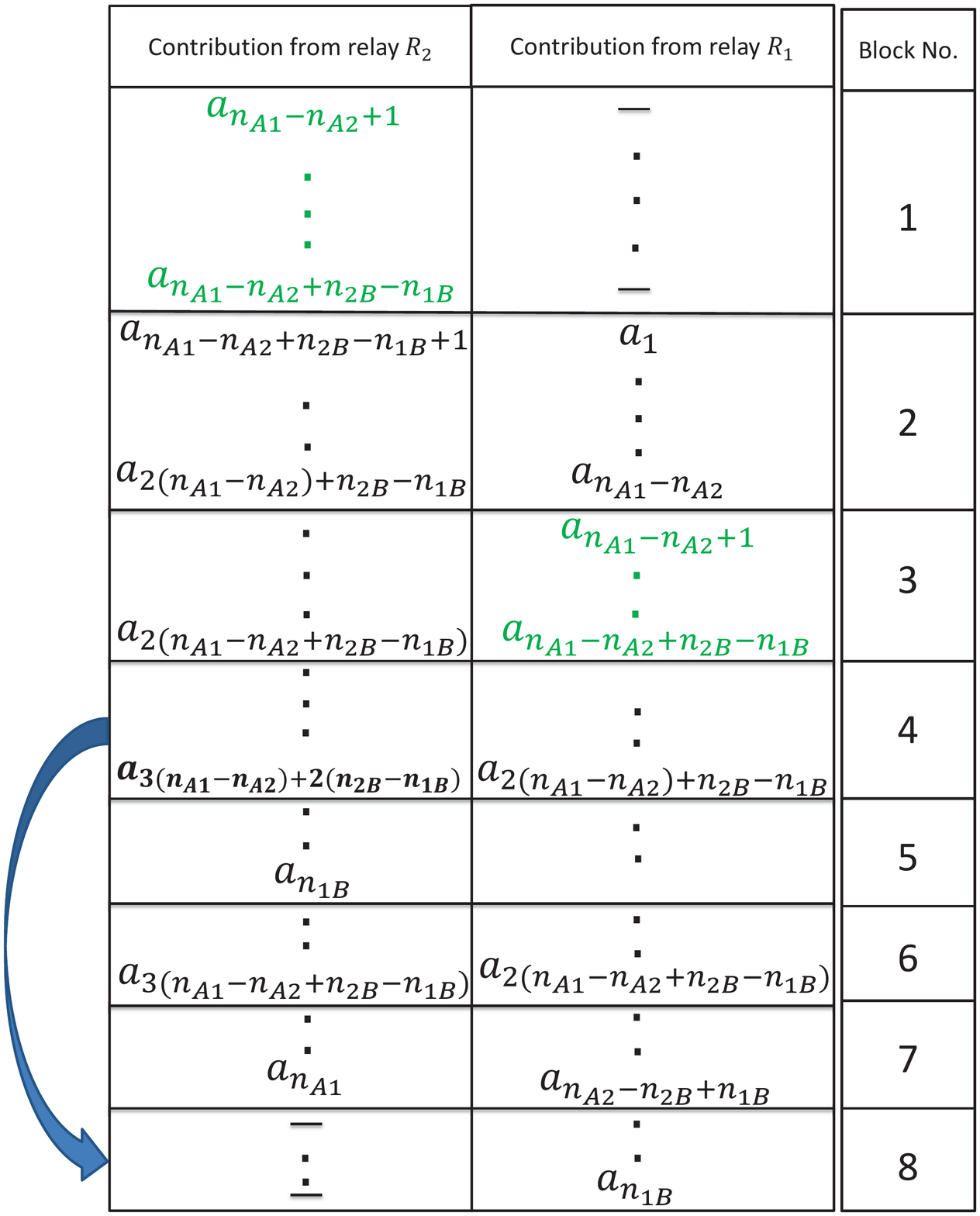}
\caption{The received signals at node $B$ (ignoring the effect of transmitted signal from $B$) assuming that both relays use Relay Strategy 0 for channel parameters of case $(u_1,v_4,r_2,s_3)$.}
\label{fig:4.23}
\end{figure}
      $R_2$ repeats block $(R_2, B_4)$ on block $(R_2, B_8)$. The decoding order is
      \begin{eqnarray}
&& {\text {decode \& subtract\ }} (R_2,B_1) \rightarrow {\text {subtract\ }} (R_1,B_3) \rightarrow {\text {decode \& subtract\ }} (R_2,B_3) \rightarrow \nonumber\\
&& {\text {subtract\ }} (R_1,B_5) \& (R_1,B_6) \rightarrow  {\text {decode \& subtract\ }} (R_2,B_5) \& (R_2,B_6) \rightarrow \nonumber\\
&& {\text {decode \& subtract\ }} (R_2,B_8) \rightarrow {\text {subtract\ }} (R_1,B_7) \& (R_1,B_8) \& (R_2,B_4) \rightarrow \nonumber\\
&& {\text {decode \& subtract\ }} (R_2,B_7) \rightarrow  {\text {decode \& subtract\ }} (R_1,B_4) \rightarrow {\text {subtract\ }} (R_2,B_2) \rightarrow \nonumber\\
&&{\text {decode \& subtract\ }} (R_1,B_2).\nonumber
\end{eqnarray}

  \item $(u_2,w_1)$: As shown in Figure \ref{fig:5},
\begin{figure}[htbp]
\centering
	\includegraphics[width=8cm]{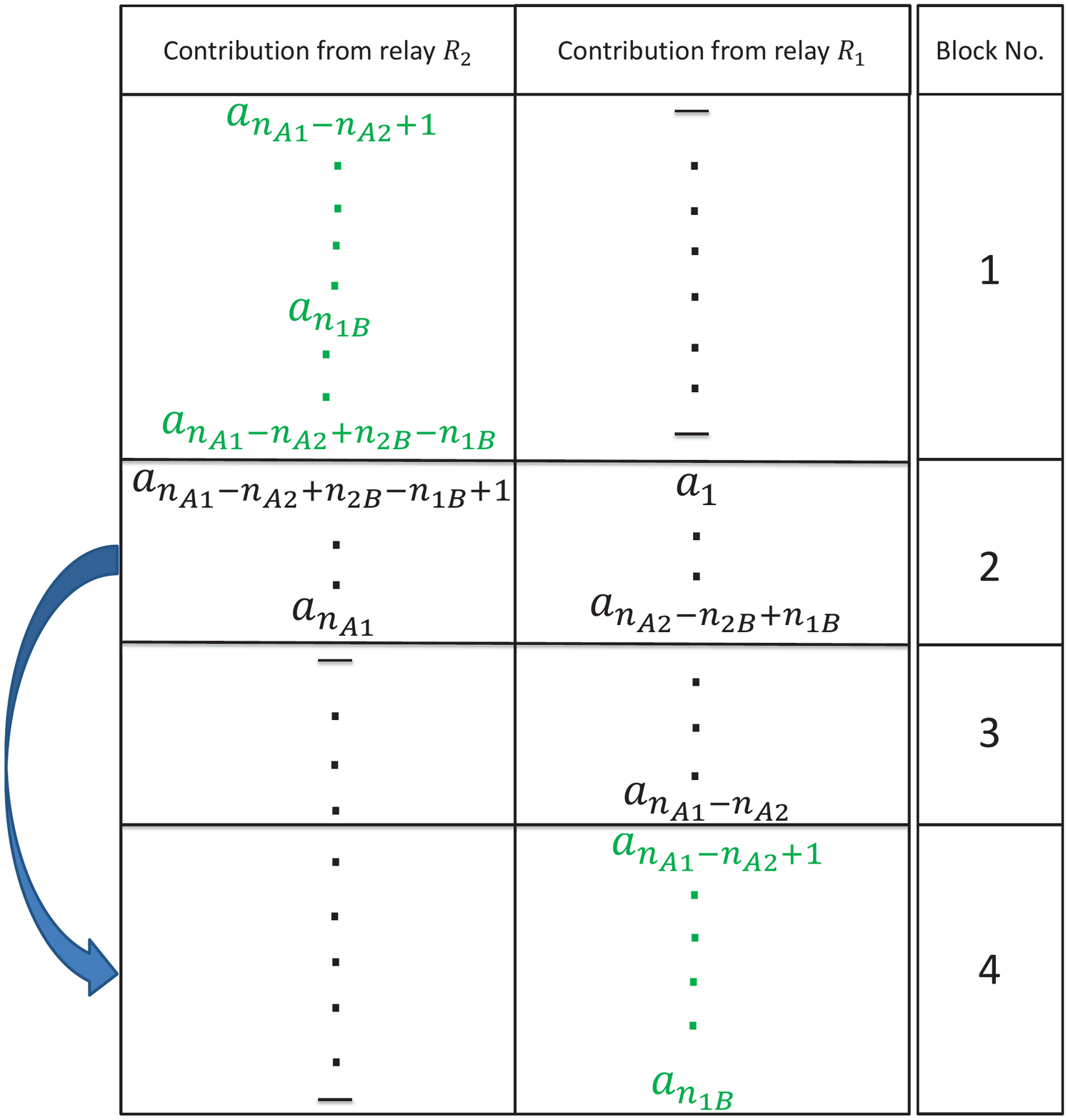}
\caption{The received signals at node $B$ (ignoring the effect of transmitted signal from $B$) assuming that both relays use Relay Strategy 0 for channel parameters of case $(u_2,w_1)$.}
\label{fig:5}
\end{figure}
      $R_2$ repeats block $(R_2, B_2)$ on block $(R_2, B_4)$. The decoding order is
      \begin{eqnarray}
&& {\text {decode \& subtract\ }} (R_2,B_1) \rightarrow {\text {subtract\ }} (R_1,B_4) \rightarrow {\text {decode \& subtract\ }}(R_2,B_4) \rightarrow \nonumber\\
&& {\text {subtract\ }}(R_1,B_2) \rightarrow  {\text {decode \& subtract\ }} (R_1,B_2) \rightarrow {\text {decode \& subtract\ }}(R_1,B_3).\nonumber
\end{eqnarray}

  \item $(u_2,w_2)$: As shown in Figure \ref{fig:6},
\begin{figure}[htbp]
\centering
	\includegraphics[width=8cm]{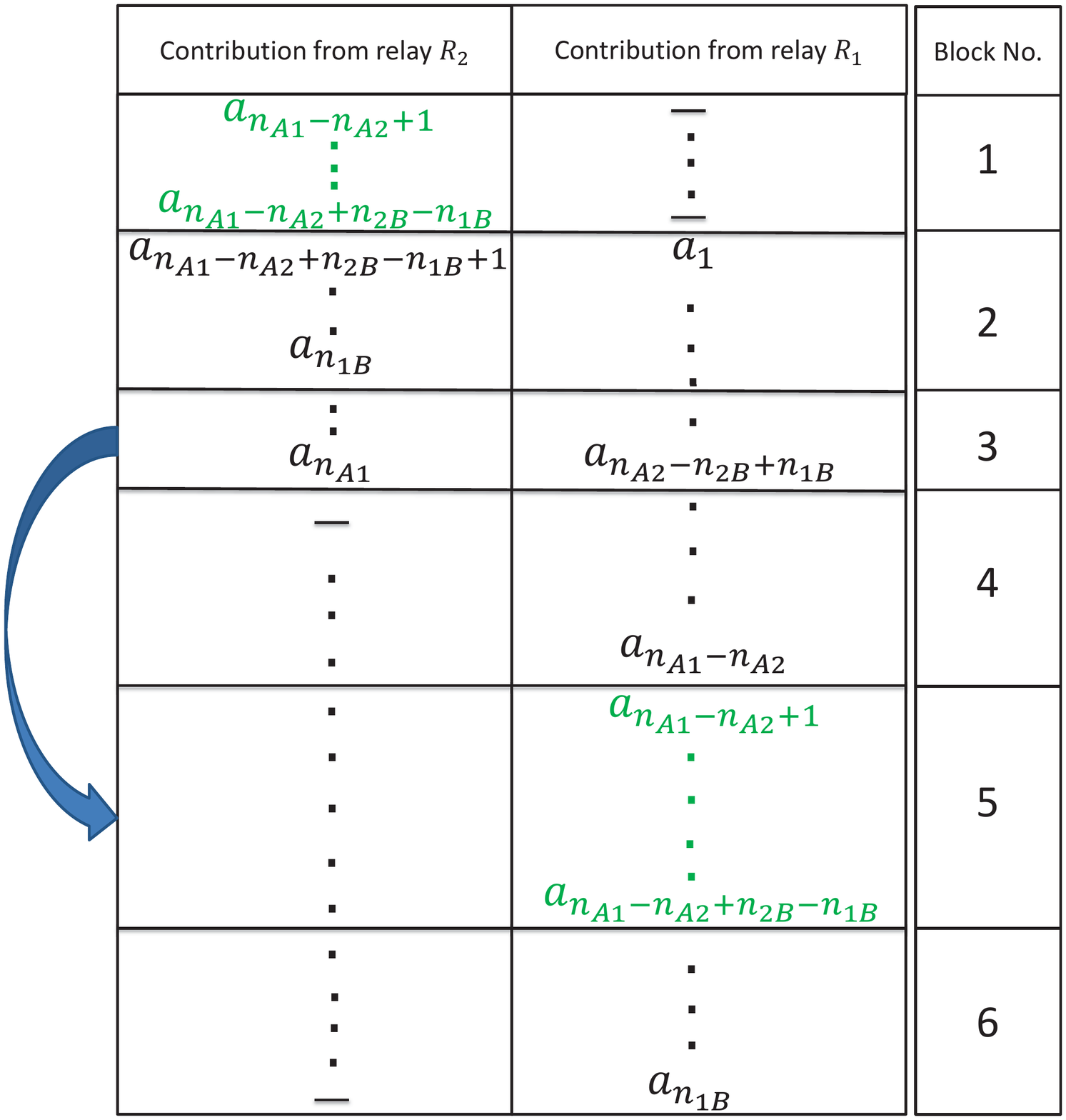}
\caption{The received signals at node $B$ (ignoring the effect of transmitted signal from $B$) assuming that both relays use Relay Strategy 0 for channel parameters of case $(u_2,w_2)$.}
\label{fig:6}
\end{figure}
      $R_2$ repeats block $(R_2, B_3)$ on block $(R_2, B_5)$. The decoding order is
      \begin{eqnarray}
&& {\text {decode \& subtract\ }} (R_2,B_1) \rightarrow {\text {subtract\ }} (R_1,B_5) \rightarrow {\text {decode \& subtract\ }} (R_2,B_5) \rightarrow \nonumber\\
&& {\text {subtract\ }} (R_2,B_3) \rightarrow  {\text {decode \& subtract\ }} (R_1,B_3) \rightarrow {\text {decode \& subtract\ }} (R_1,B_6) \rightarrow \nonumber\\
&& {\text {subtract\ }} (R_2,B_2) \rightarrow  {\text {decode \& subtract\ }} (R_1,B_2) \rightarrow {\text {decode \& subtract\ }} (R_1,B_4).\nonumber
\end{eqnarray}

\end{enumerate}
\begin{remark} In all cases above, we can see that for every $V$ streams that we want to repeat, there are $V+(n_{2B}-n_{A1})$ empty spots available, which makes it flexible to place the $V$ streams.
\end{remark}

\begin{example}
Consider the case $(n_{A1},n_{A2},n_{1B},n_{2B},n_{B1},n_{B2},n_{1A},n_{2A})=(6,4,5,7,6,3,6,4)$. With these parameters, the forward channel is of Case 3.1.2 Type 1, and the backward channel is of Case 1. We use the transmission strategy for node $B$ for Case 1 given in Appendix \ref{apdx_sc1} for the backward channel and transmit ${[a_{C_{AB}},...,a_{1}]}^T$ for the forward channel. Also, $R_1$ uses Relay Strategy 0, and $R_2$ uses Relay Strategy 2. The desired messages can be decoded by both nodes $A$ and $B$, as illustrated in Figure \ref{fig:Example.S.2}.
\end{example}

\begin{figure}[htbp]
\centering
\subfigure[Transmission to relays.]{
	\includegraphics[width=8.5cm]{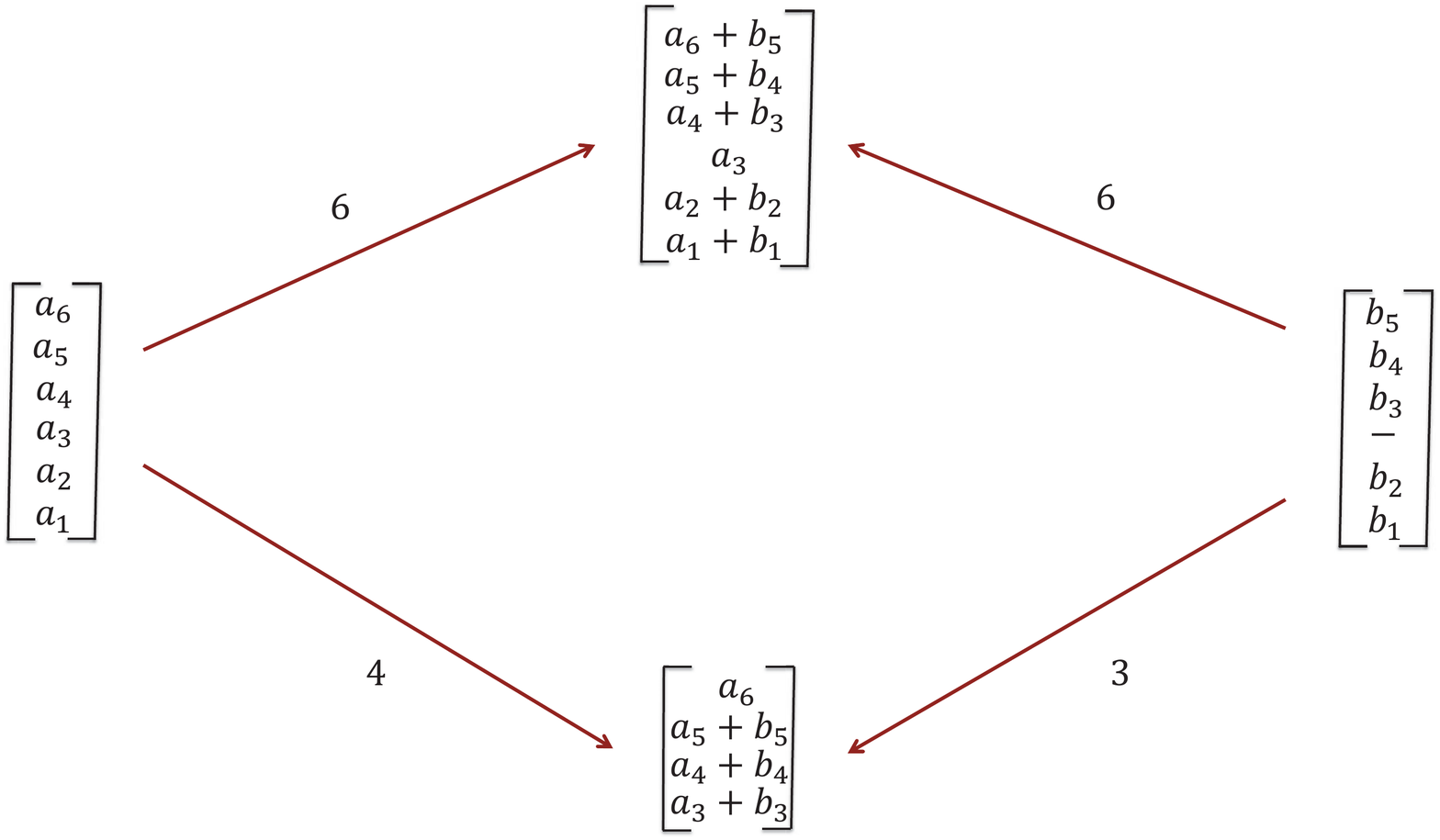}
\label{fig:subfigS2.1}
}
\subfigure[Reception from relays.]{
	\includegraphics[width=8.5cm]{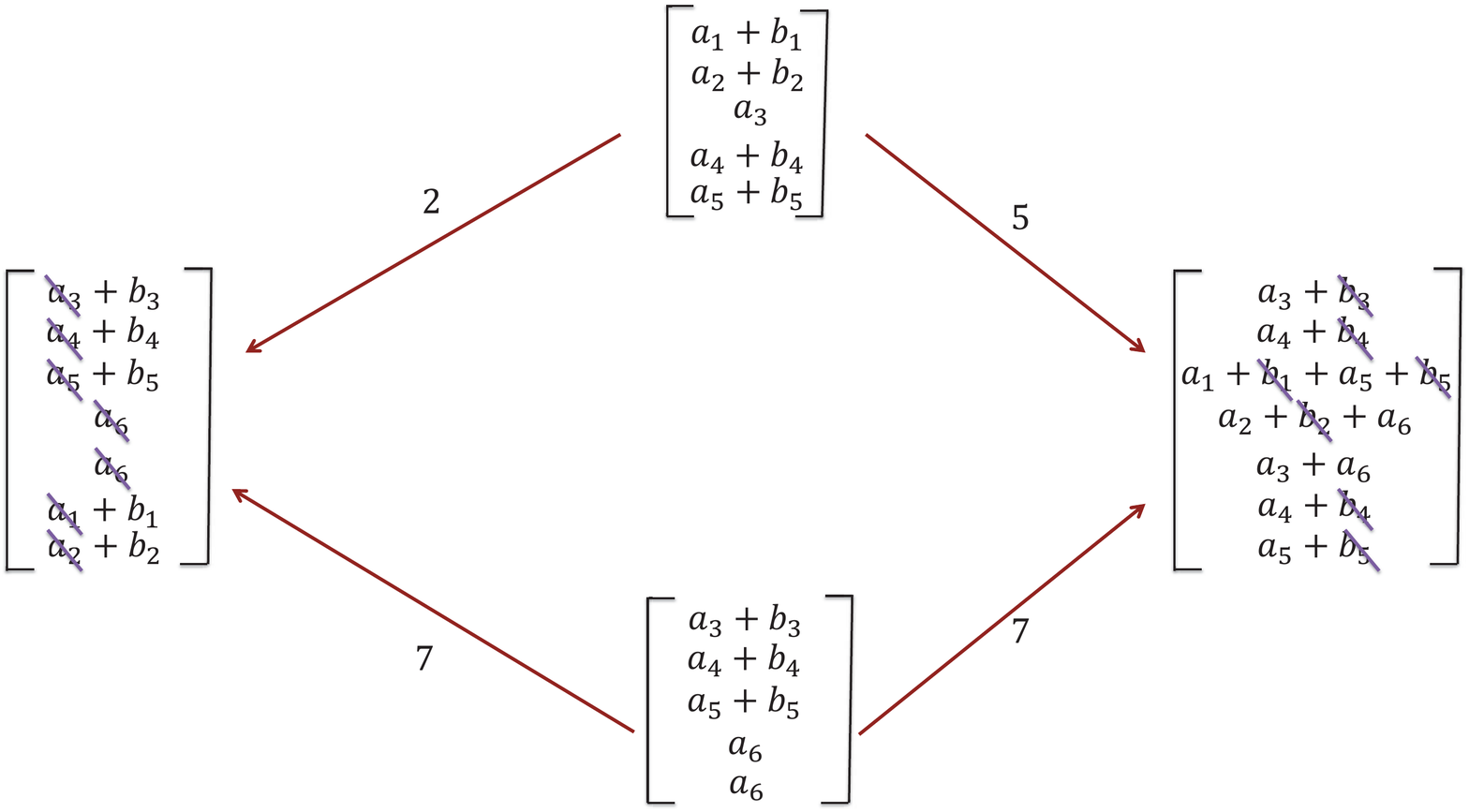}
\label{fig:subfigS2.2}
}
\caption[Optional caption for list of figures]{Example for $(n_{A1},n_{A2},n_{1B},n_{2B},n_{B1},n_{B2},n_{1A},n_{2A})=(6,4,5,7,6,3,6,4)$ using Relay Strategy 2.}
\label{fig:Example.S.2}
\end{figure}

\begin{remark} For $(u_1,v_3)$ and $(u_1,v_4,r_2,s_2)$ Relay Strategy 2 is equivalent to Relay Strategy 0. An example is depicted in Figure \ref{fig:Exampleqq} where $(n_{A1},n_{A2},n_{1B},n_{2B})=(10,8,7,10)$ (Case 3.1.2 Type 1, $(u_1,v_3)$).
\end{remark}

\begin{figure}[htbp]
\centering
	\includegraphics[width=15cm]{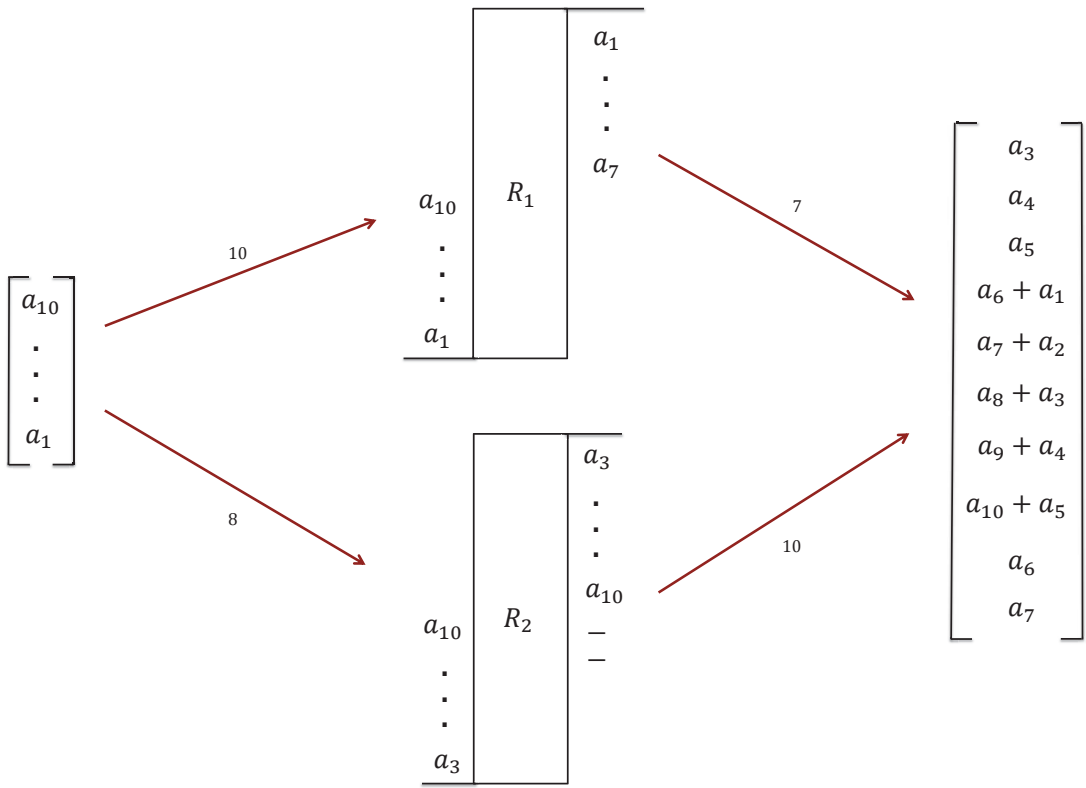}
\caption{An example of the Case 3.1.2 that can be decoded by using Relay Strategy 0 at both relays.}
\label{fig:Exampleqq}
\end{figure}

\subsection{\bf Relay Strategy 3:}

If the forward channel is of Case 3.1.2 Type $i$, then Relay Strategy 0 is used at $R_i$, and Relay Strategy 3 is used at $R_{\bar{i}}$, where $i,\bar{i}\in\{1,2\}, i\neq \bar{i}$. Here, we define Relay Strategy 3 at $R_2$ (forward channel of Case 3.1.2 Type $1$), while that for $R_1$ can be obtained by interchanging roles of relays $R_1$ and $R_2$ (interchanging 1 and 2 and forward channel of Case 3.1.2 Type $2$). As shown in Figure \ref{fig:rs3}, if $R_2$ receives a block of $n_{2B}$ bits, first it will reverse them as in Relay Strategy 0 and then changes the order of the $n_{A2}-(n_{2B}-n_{1B})$ streams right after the first $n_{2B}-n_{1B}$ streams, with the following $n_{2B}-n_{A2}$ streams.

\begin{figure}[htbp]
\centering
	\includegraphics[width=10cm]{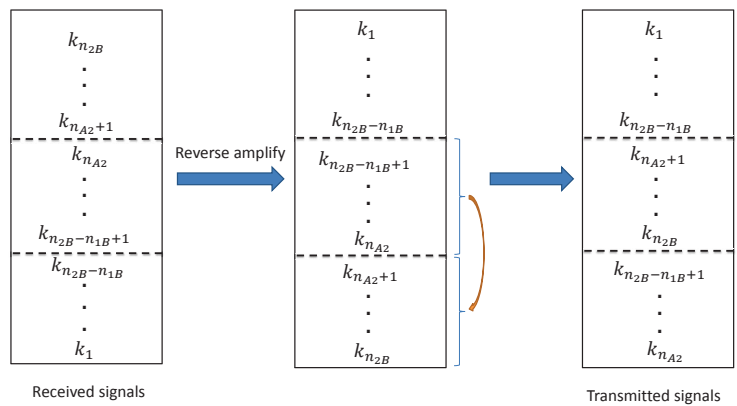}
\caption{Relay Strategy 3 at $R_2$.}
\label{fig:rs3}
\end{figure}

Node $A$ transmits ${[a_{C_{AB}},...,a_{1}]}^T$. The received signals can be seen in Figure \ref{fig:s3s}. The block $(R_2,B_1)$ will be decoded from the top levels of the received signal from $R_2$ without any interference from $R_1$. We then subtract the corresponding signals in blocks $(R_1,B_3)$ and $(R_1,B_4)$. Also, bits that were not delivered to node $B$ from $R_2$ using Relay Strategy 0, ($a_{1}, ..., a_{n_{A1}-n_{A2}}$), are decoded from block $(R_1,B_2)$ without any interference. The remaining bits can be decoded by starting from the highest remaining level ($a_{(n_{A1}-n_{A2})+(n_{2B}-n_{1B})+1}$ in block $(R_2,B_4)$) and removing the effect of the decoded bits.


\begin{figure}[htbp]
\centering
	\includegraphics[width=15cm]{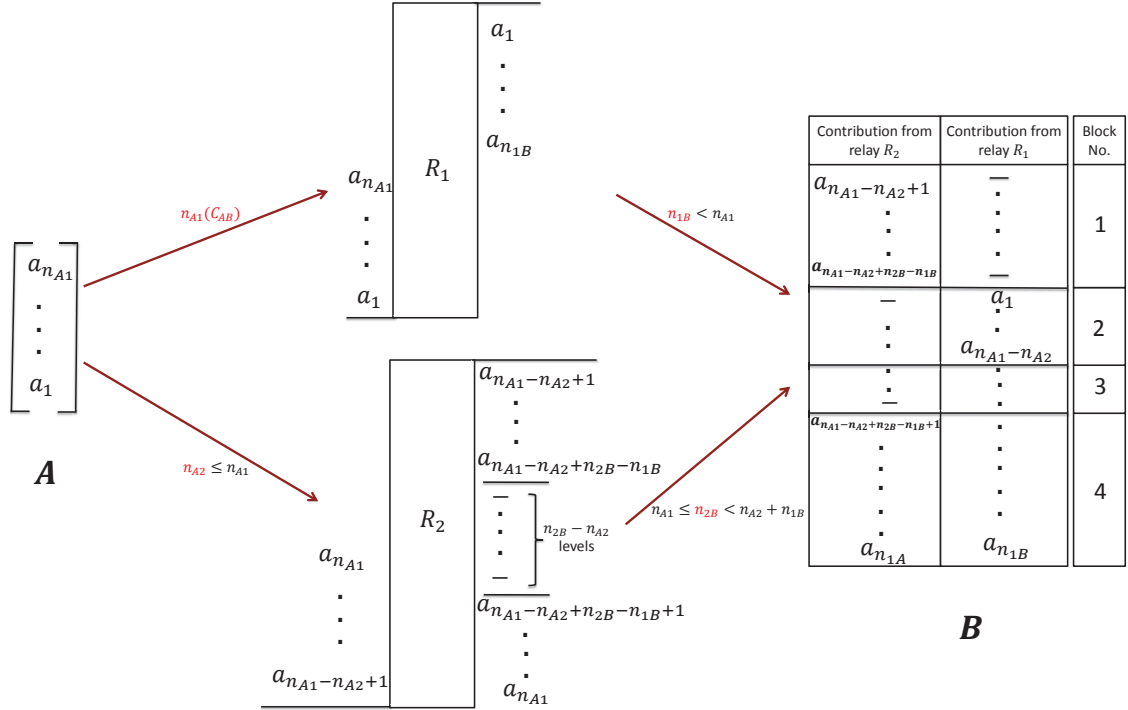}
\caption{Received signals by using Relay Strategy 3 when the forward channel is of Case 3.1.2 Type 1.}
\label{fig:s3s}
\end{figure}

\begin{example}
Consider the case $(n_{A1},n_{A2},n_{1B},n_{2B},n_{B1},n_{B2},n_{1A},n_{2A})=(6,4,6,5,5,7,6,7)$. With these parameters, the forward channel is of Case 3.1.2 Type 1, and the backward channel is of Case 3.2 Type 1. We use the transmission strategy for the backward channel corresponding to this case given in Appendix \ref{apdx_sc1} (transmit $[b_{C_{BA}}, \cdots, b_1]^T$) and transmit ${[a_{C_{AB}},...,a_{1}]}^T$ for the forward channel. Also, $R_1$ uses Relay Strategy 0, and $R_2$ uses Relay Strategy 3. Figure \ref{fig:Example.S.3} illustrates that the desired messages can be decoded by both nodes $A$ and $B$.
\end{example}

\begin{figure}[htbp]
\centering
\subfigure[Transmission to relays.]{
	\includegraphics[width=8.5cm]{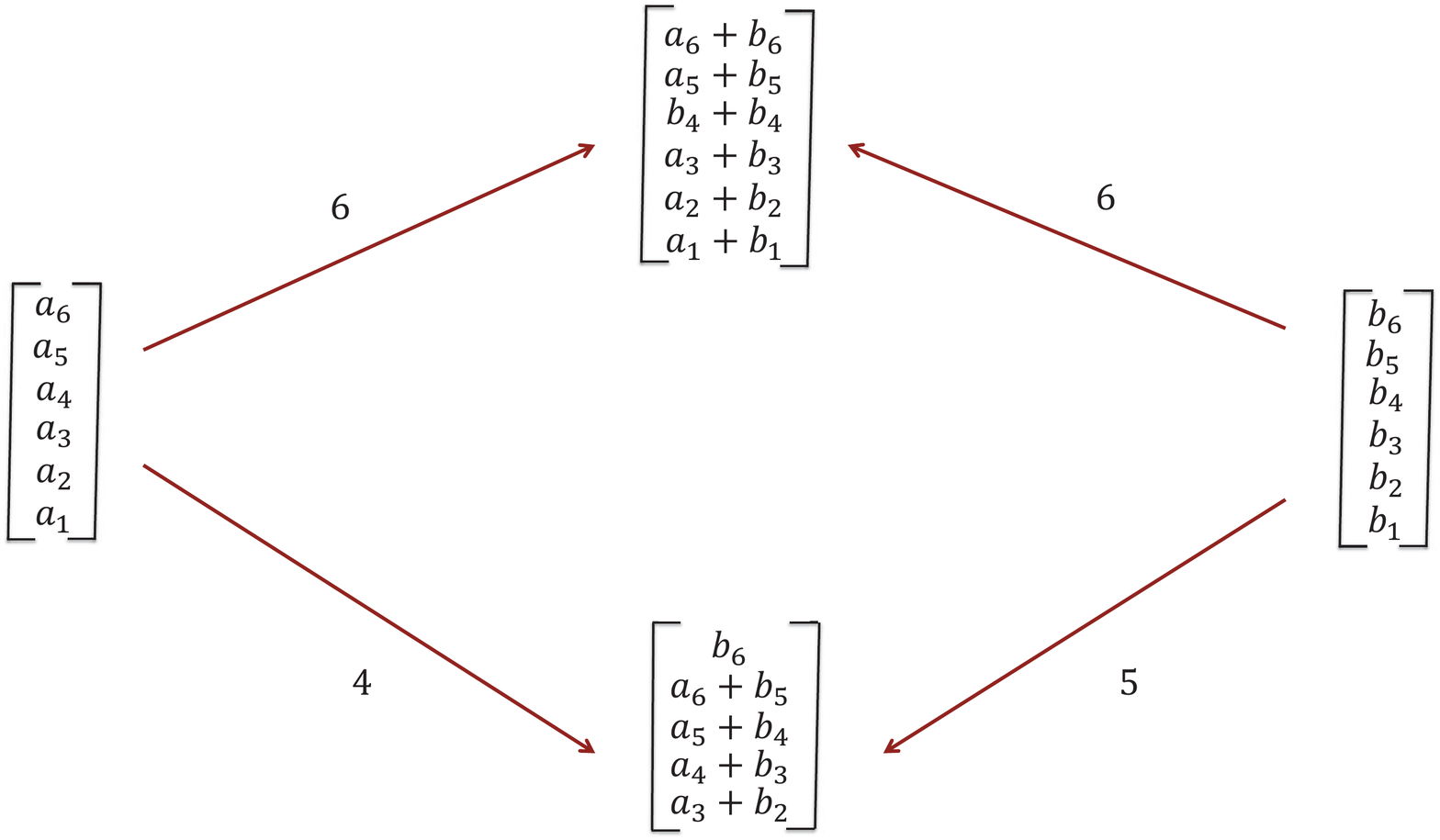}
\label{fig:subfigS3.1}
}
\subfigure[Reception from relays.]{
	\includegraphics[width=8.5cm]{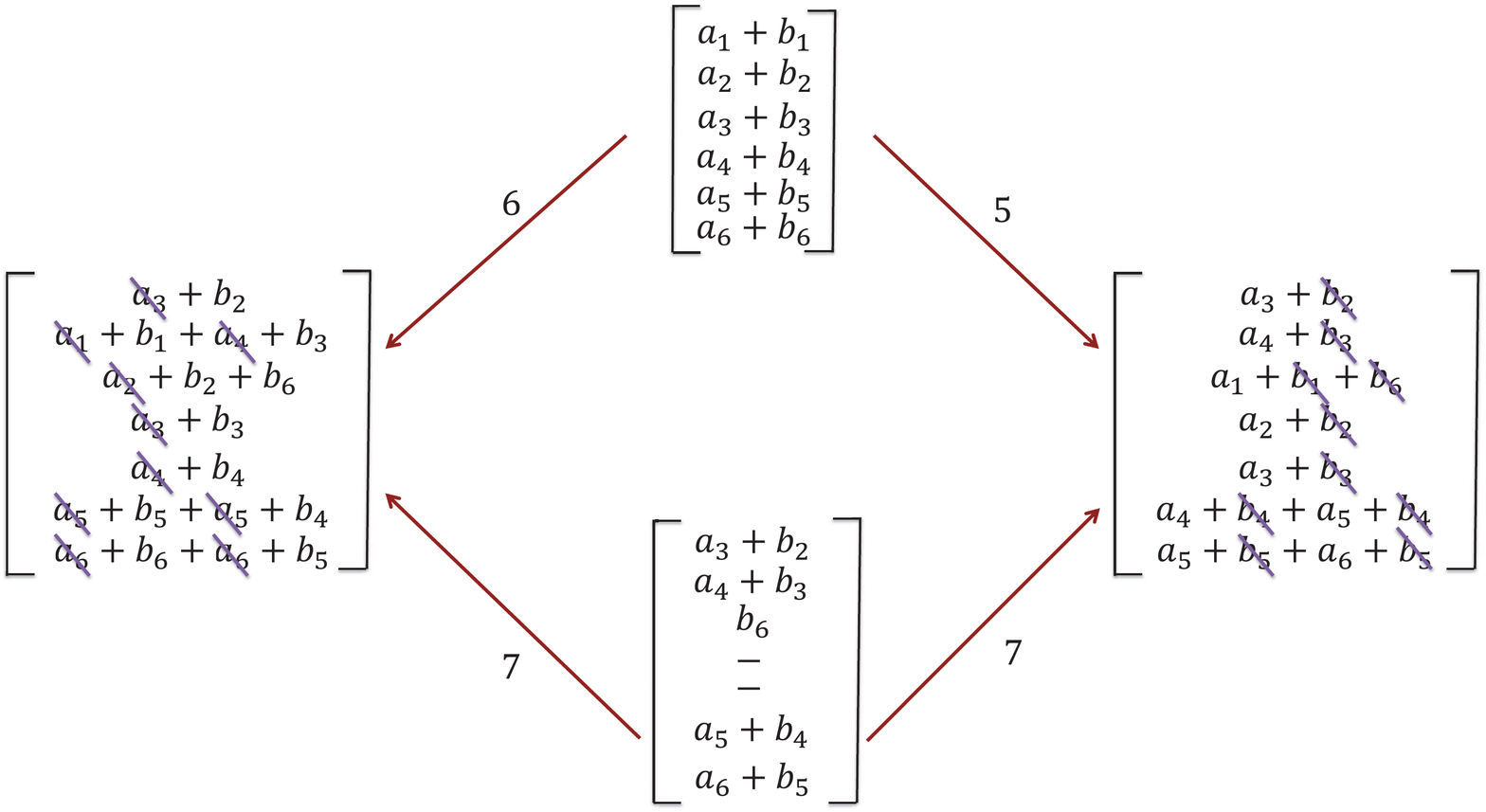}
\label{fig:subfigS3.2}
}
\caption[Optional caption for list of figures]{Example for $(n_{A1},n_{A2},n_{1B},n_{2B},n_{B1},n_{B2},n_{1A},n_{2A})=(6,4,6,5,5,7,6,7)$ using Relay Strategy 3.}
\label{fig:Example.S.3}
\end{figure}

\subsection{\bf Relay Strategy 4:}

If the forward channel is of Case 3.1.2 Type $i$, then Relay Strategy 0 is used at $R_{\bar{i}}$, where $i,\bar{i}\in\{1,2\}, i\neq \bar{i}$, and Relay Strategy 4 is used at $R_{i}$. Here, we define Relay Strategy 4 at $R_1$ (forward channel of Case 3.1.2 Type $1$), while that for $R_2$ can be obtained by interchanging roles of $R_1$ and $R_2$ (interchanging 1 and 2 and forward channel of Case 3.1.2 Type $2$). As shown in Figure \ref{fig:rs4}, if $R_1$ receives a block of $n_{1B}$ bits, first it will reverse them as in Relay Strategy 0 and then changes the order of the first $n_{A1}-n_{A2}$ streams with the next $n_{1B}-(n_{A1}-n_{A2})$ streams.

\begin{figure}[htbp]
\centering
	\includegraphics[width=10cm]{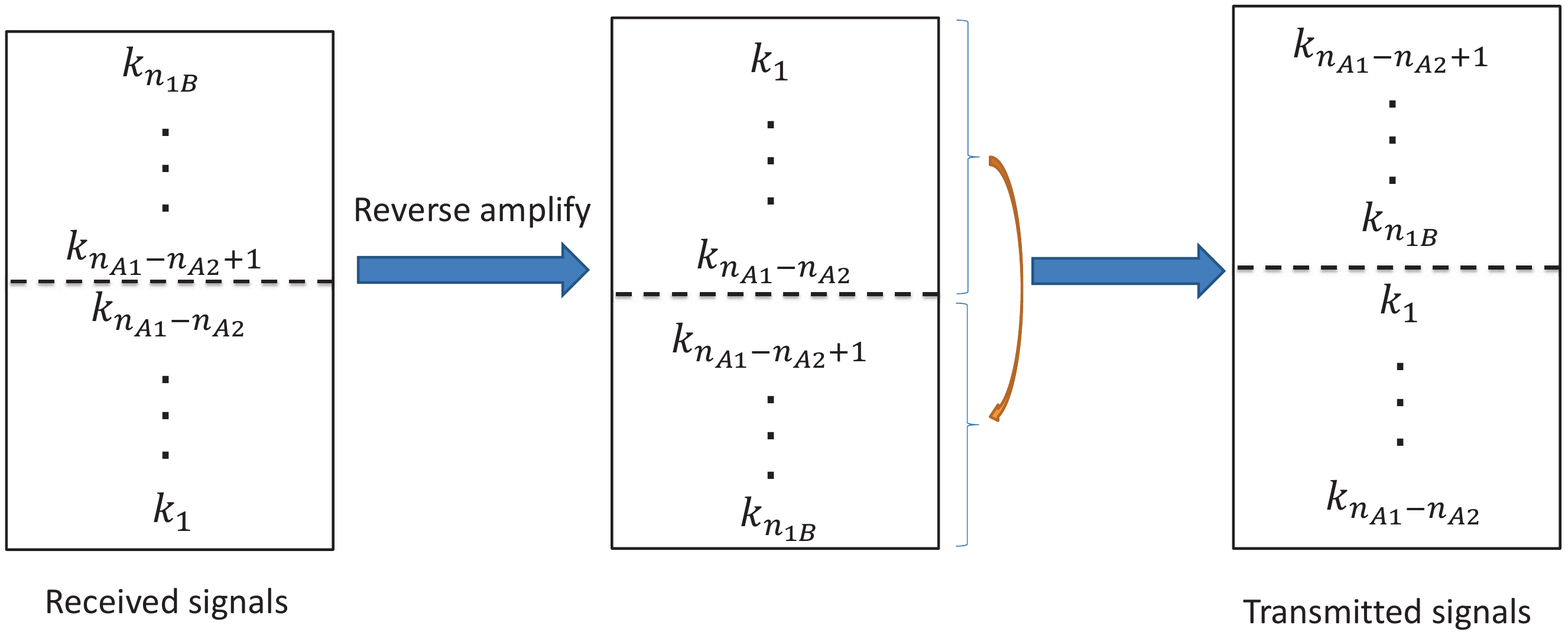}
\caption{Relay Strategy 4 at $R_1$.}
\label{fig:rs4}
\end{figure}

Node $A$ transmits ${[a_{C_{AB}},...,a_{1}]}^T$. The received signals can be seen in Figure \ref{fig:s4s}. Bits that are not delivered to node $B$ from $R_2$ using Relay Strategy 0 in the block $(R_1,B_4)$ are decoded without any interference. The remaining bits can be decoded by starting from the highest level ($a_{n_{A1}-n_{A2}+1}$ in block $(R_2,B_1)$) and removing the effect of the decoded bits.


\begin{figure}[htbp]
\centering
	\includegraphics[width=15cm]{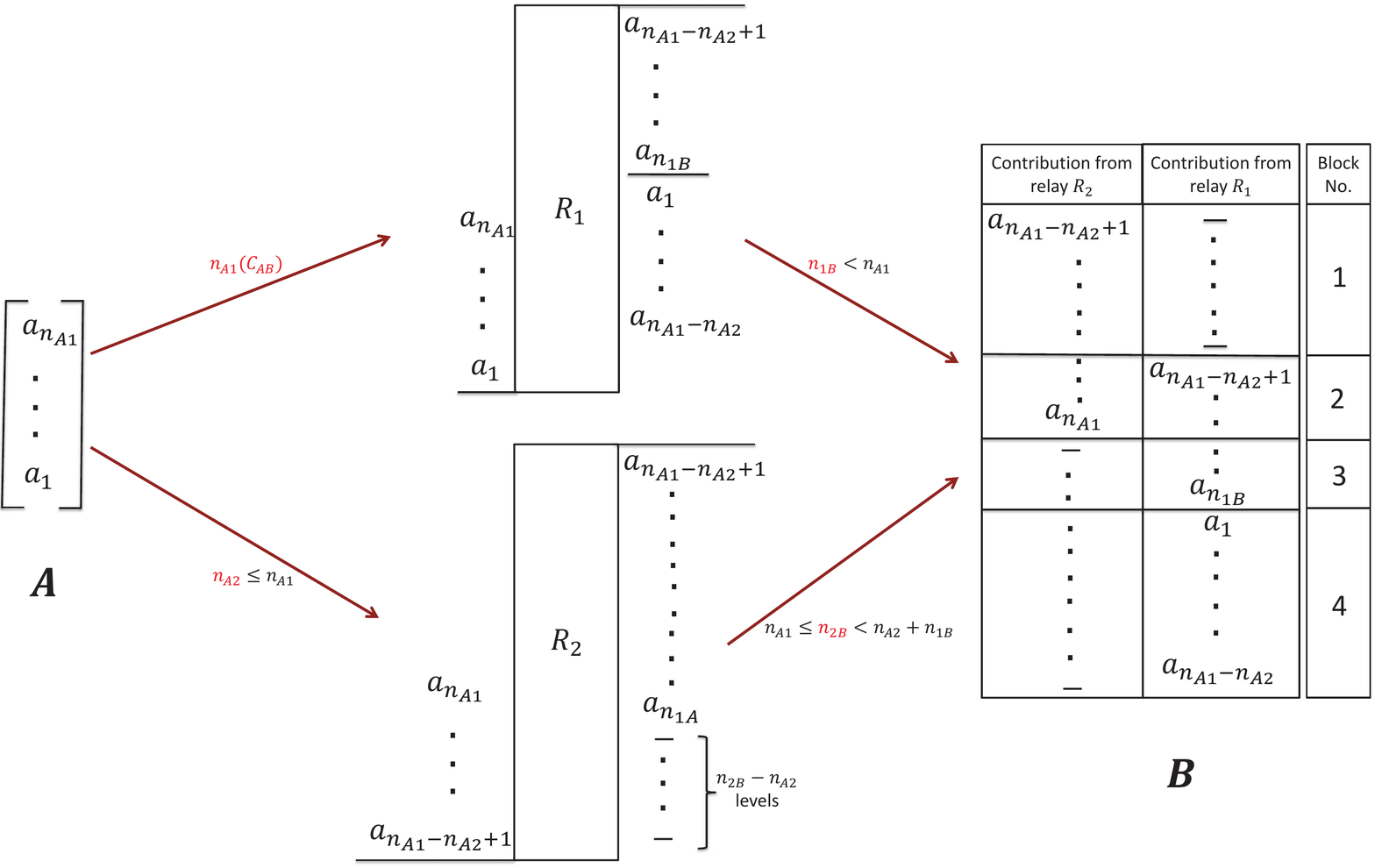}
\caption{Received signals by using Relay Strategy 4 when the forward channel is of Case 3.1.2 Type 1.}
\label{fig:s4s}
\end{figure}

\begin{example}
Consider the case $(n_{A1},n_{A2},n_{1B},n_{2B},n_{B1},n_{B2},n_{1A},n_{2A})=(6,4,6,5,5,7,6,7)$. With these parameters, the forward channel is of Case 3.1.2 Type 1, and the backward channel is of Case 4.2 Type 1. We use the transmission strategy for the backward channel corresponding to this case given in Appendix \ref{apdx_sc1} and transmit ${[a_{C_{AB}},...,a_{1}]}^T$ for the forward channel. Also, $R_2$ uses Relay Strategy 0, and $R_1$ uses Relay Strategy 4. Figure \ref{fig:Example.S.4} illustrates that the desired messages can be decoded by both nodes $A$ and $B$.
\end{example}

\begin{figure}[htbp]
\centering
\subfigure[Transmission to relays.]{
	\includegraphics[width=8.5cm]{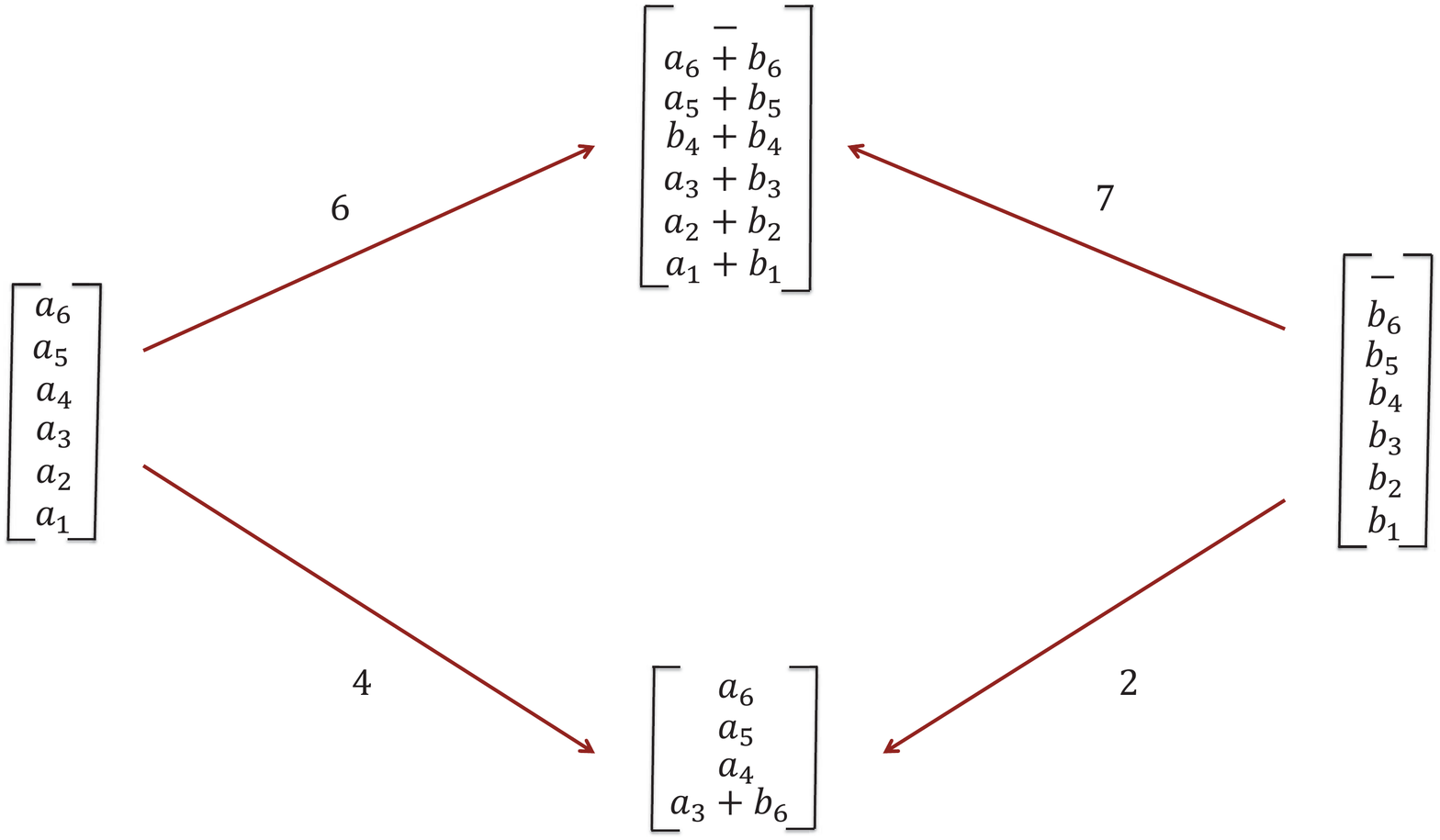}
\label{fig:subfigS4.1}
}
\subfigure[Reception from relays.]{
	\includegraphics[width=8.5cm]{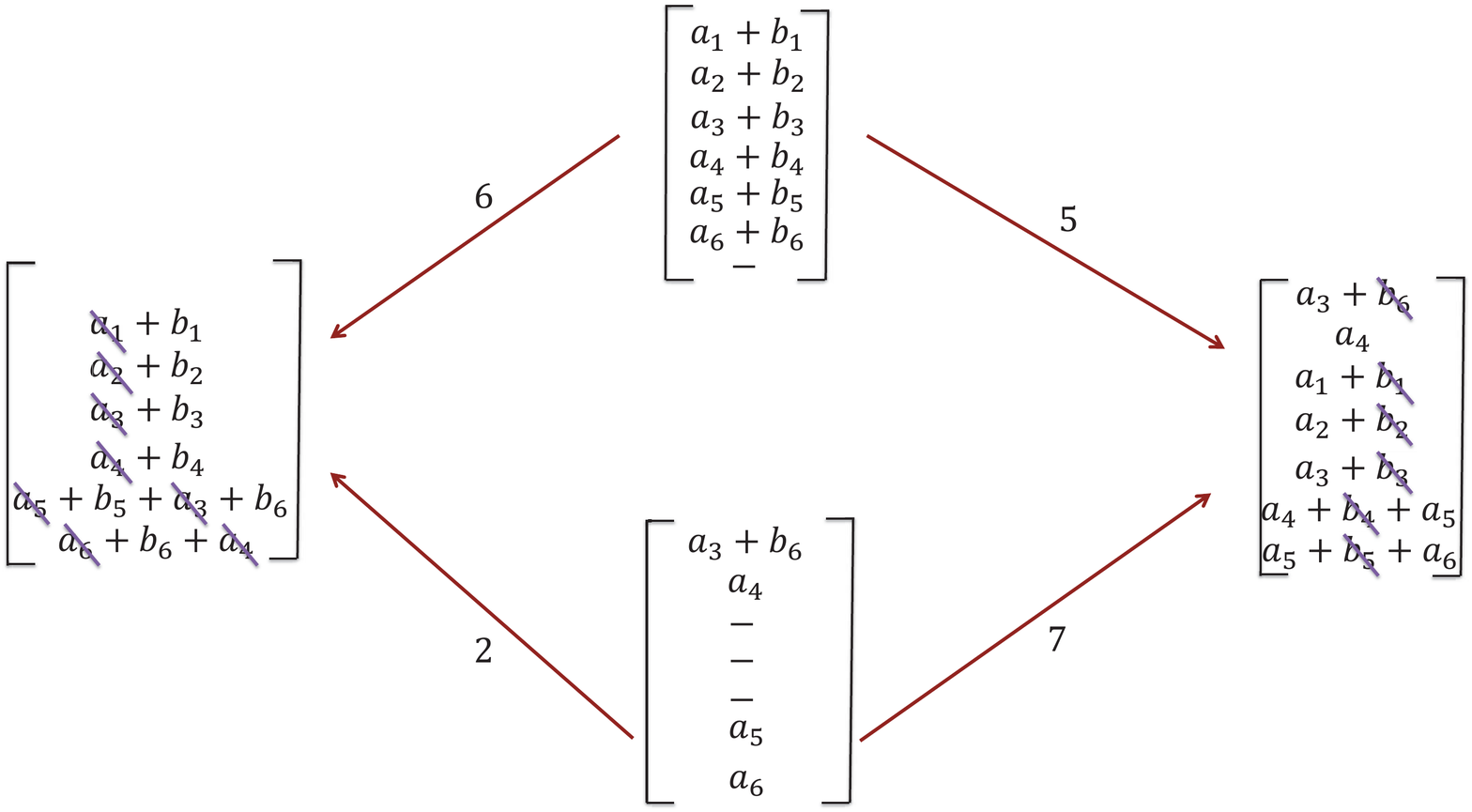}
\label{fig:subfigS4.2}
}
\caption[Optional caption for list of figures]{Example for $(n_{A1},n_{A2},n_{1B},n_{2B},n_{B1},n_{B2},n_{1A},n_{2A})=(6,4,5,7,7,2,6,4)$ using Relay Strategy 4.}
\label{fig:Example.S.4}
\end{figure}












\subsection{Achieving the Optimum Rate}

Now we explain how the above mentioned strategies achieve the optimal rate for any set of parameters on the backward channel.

1. Backward channel is of Case 1:
\begin{itemize}
  \item Forward channel is of Case 3.1.2 Type 1: If $n_{A2}>n_{B2}$, we use Relay Strategy 2 at $R_2$ and Relay Strategy 0 at $R_1$. Otherwise use Relay Strategy 1 at $R_2$ and Relay Strategy 0 at $R_1$. Figure \ref{fig:d1} shows the backward channel when the forward channel is of Case 3.1.2 Type 1. If $n_{A2}>n_{B2}$, $R_2$ repeats from the streams that are already decoded from the highest levels received in $A$, ($(b_{n_{1A}+1},...,b_{n_{1A}+n_{B2}})$ in green in Figure \ref{fig:d1}) on the lower levels, and otherwise it just changes the order of some of the equations at the highest levels received in $A$, ($(b_{n_{1A}+1},...,b_{n_{1A}+n_{B2}})$ in green in Figure \ref{fig:d1}), which does not affect the decoding.

  \item Forward channel is of Case 3.1.2 Type 2: If $n_{A1}>n_{1A}$, we use Relay Strategy 2 at $R_1$ and Relay Strategy 0 at $R_2$. Otherwise use Relay Strategy 1 at $R_1$ and Relay Strategy 0 at $R_2$. If $n_{A1}>n_{1A}$, $R_1$ repeats from the streams ($b_{1},...,b_{n_{1A}}$) received below the noise level in $A$, and otherwise it just changes the order of some of the equations ($b_{1},...,b_{n_{1A}}$).
\end{itemize}

\begin{figure}[htbp]
\centering
	\includegraphics[width=15cm]{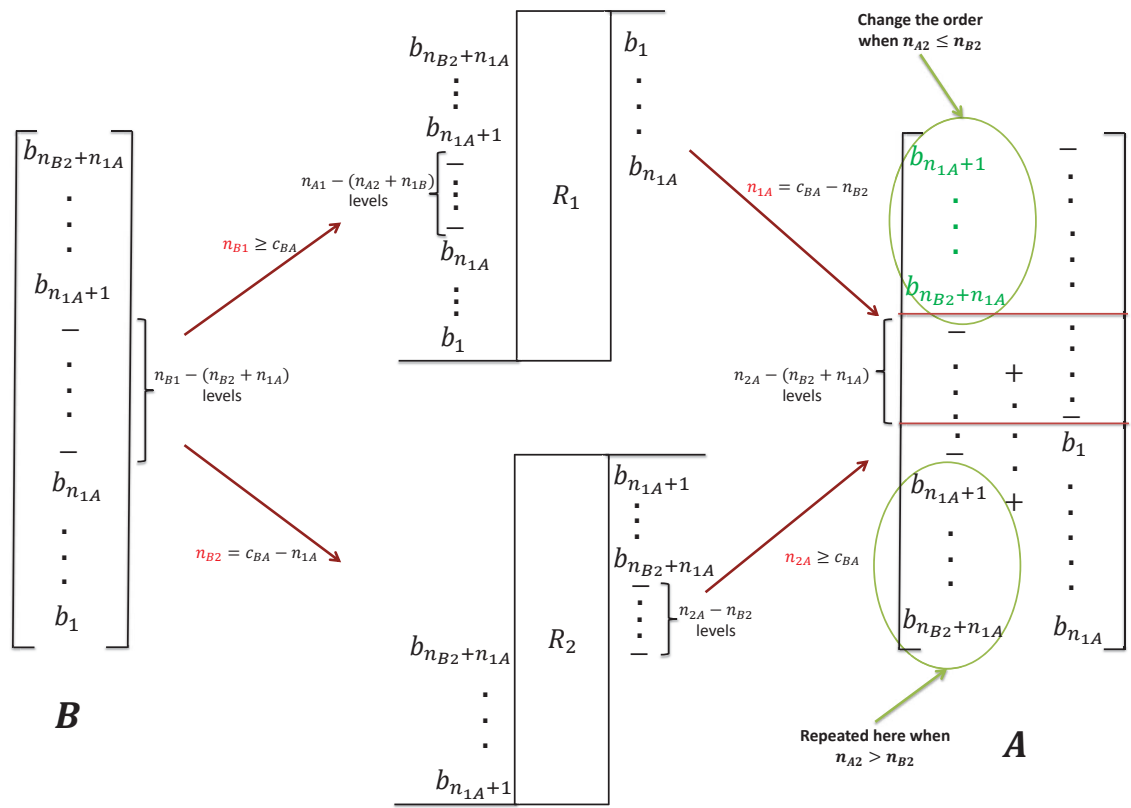}
\caption{Backward channel (Case 1) when forward channel is of Case 3.1.2 Type 1.}
\label{fig:d1}
\end{figure}

2. Backward channel is of Case 2:
\begin{itemize}
  \item Forward channel is of Case 3.1.2 Type 1: If $n_{A2}>n_{2A}$, we use Relay Strategy 1 at $R_2$ and Relay Strategy 0 at $R_1$. Otherwise use Relay Strategy 2 at $R_2$ and Relay Strategy 0 at $R_1$. If $n_{A2}>n_{2A}$, $R_2$ repeats from the streams ($b_{1},...,b_{n_{2A}}$) received below the noise level in $A$, and otherwise it just changes the order of some of the equations ($b_{1},...,b_{n_{2A}}$).

  \item Forward channel is of Case 3.1.2 Type 2: If $n_{A1}>n_{B1}$, we use Relay Strategy 2 at $R_1$ and Relay Strategy 0 at $R_2$. Otherwise use Relay Strategy 1 at $R_1$ and Relay Strategy 0 at $R_2$. If $n_{A1}>n_{B1}$, $R_1$ repeats from the streams that are already decoded from the highest levels received in $A$, ($b_{n_{2A}+1},...,b_{n_{2A}+n_{B1}}$), on the lower levels, and otherwise it just changes the order of some of the equations at the highest levels received in $A$, ($b_{n_{2A}+1},...,b_{n_{2A}+n_{B1}}$).
\end{itemize}

3. Backward channel is of Case 3.1.1: We assume that the backward channel is Type 1. For Type 2 the argument is similar.
\begin{itemize}
  \item Forward channel is of Case 3.1.2 Type 1: If $n_{A2}>n_{B2}$, we use Relay Strategy 2 at $R_2$ and Relay Strategy 0 at $R_1$. Otherwise use Relay Strategy 1 at $R_2$ and Relay Strategy 0 at $R_1$. If $n_{A2}>n_{B2}$, $R_2$ repeats from the streams that are already decoded from the highest levels received in $A$, ($b_{n_{B1}-n_{B2}+1},...,b_{n_{B1}}$), on the lower levels, and otherwise it just changes the order of some of the equations at the highest levels received in $A$, ($b_{n_{B1}-n_{B2}+1},...,b_{n_{B1}}$).

  \item Forward channel is of Case 3.1.2 Type 2: If $n_{A1}>n_{B1}$, we use Relay Strategy 2 at $R_1$ and Relay Strategy 0 at $R_2$. Otherwise use Relay Strategy 1 at $R_1$ and Relay Strategy 0 at $R_2$. If $n_{A1}>n_{B1}$, $R_1$ repeats from the streams ($b_{1},...,b_{n_{1A}}$) received below the noise level in $A$, and otherwise it just changes the order of some of the equations ($b_{1},...,b_{n_{1A}}$).
\end{itemize}

4. Backward channel is of Case 4.1.1: We assume that the backward channel is Type 1. For Type 2 the argument is similar.
\begin{itemize}
  \item Forward channel is of Case 3.1.2 Type 1: If $n_{A2}>n_{2A}$, we use Relay Strategy 2 at $R_2$ and Relay Strategy 0 at $R_1$. Otherwise use Relay Strategy 1 at $R_2$ and Relay Strategy 0 at $R_1$. If $n_{A2}>n_{2A}$, $R_2$ repeats from the streams ($b_{1},...,b_{n_{2A}}$) received below the noise level in $A$, and otherwise it just changes the order of some of the equations ($b_{1},...,b_{n_{2A}}$).

  \item Forward channel is of Case 3.1.2 Type 2: If $n_{A1}>n_{1A}-n_{2A}$, we use Relay Strategy 2 at $R_1$ and Relay Strategy 0 at $R_2$. Otherwise use Relay Strategy 1 at $R_1$ and Relay Strategy 0 at $R_2$. If $n_{A1}>n_{1A}-n_{2A}$, $R_1$ repeats from the streams that are already decoded from the highest levels received in $A$, ($b_{n_{2A}+1},...,b_{n_{1A}}$), on the lower levels, and otherwise it just changes the order of some of the equations at the highest levels received in $A$, ($b_{n_{2A}+1},...,b_{n_{1A}}$).
\end{itemize}

5. Backward channel is of Case 3.2: We assume that the backward channel is Type 1. For Type 2 the argument is similar.
\begin{itemize}
  \item Forward channel is of Case 3.1.2 Type 1: We use Relay Strategy 1 at $R_2$ and Relay Strategy 0 at $R_1$ or Relay Strategy 2 at $R_2$ and Relay Strategy 0 at $R_1$ or Relay Strategy 3 at $R_2$ and Relay Strategy 0 at $R_1$.

  \item Forward channel is of Case 3.1.2 Type 2: We use Relay Strategy 4 at $R_2$ and Relay Strategy 0 at $R_1$.
\end{itemize}

6. Backward channel is of Case 4.2: We assume that the backward channel is Type 1. For Type 2 the argument is similar.
\begin{itemize}
  \item Forward channel is of Case 3.1.2 Type 1: We use Relay Strategy 1 at $R_2$ and Relay Strategy 0 at $R_1$ or Relay Strategy 2 at $R_2$ and Relay Strategy 0 at $R_1$ or Relay Strategy 3 at $R_2$ and Relay Strategy 0 at $R_1$.

  \item Forward channel is of Case 3.1.2 Type 2: We use Relay Strategy 4 at $R_2$ and Relay Strategy 0 at $R_1$.
\end{itemize}




For the case that forward channel is of Case 4.1.2, the proof is given in Appendix \ref{apdx_sc2}. An essential difference compared to Case 3.1.2 includes the freedom in transmission strategy (there are more transmission streams at $A$ than the capacity) and no freedom at the receiver side (number of the reception streams at $B$ is equal to the capacity) for the forward channel.

\section{Both the forward and backward channels are either of Case 3.1.2 or 4.1.2}\label{bothare}

In Section \ref{one} and Appendix \ref{apdx_sc2}, we used Relay Strategy 2 or Relay Strategy 6 as one of the achievability strategies when the forward channel is of Case 3.1.2 or 4.1.2, respectively. In this section, we will show that using a modified combination of these strategies achieve the optimal  capacity region when both the forward and backward channels are either of Case 3.1.2 or 4.1.2.

We will define Relay Strategy $(m_i,n_i)$ at $R_i$ for $i\in\{1,2\}$, $m_i,n_i \in \{0,2,6\}$. If the forward channel is of Case 3.1.2, at $R_1$, we use $m_1=0$ when the forward channel is Type 1 and $m_1=2$ otherwise. At $R_2$, we use $m_2=2$ when the forward channel is Type 1 and $m_2=0$ otherwise. If the forward channel is of Case 4.1.2, at $R_1$, we use $m_1=6$ when the forward channel is Type 1 and $m_1=0$ otherwise. At $R_2$, we use $m_2=0$ when the forward channel is Type 1 and $m_2=6$ otherwise. The value of $n_i$ is determined the same way based on the backward channel parameters.

Relay Strategy $(m_i,0)$ at $R_i$ uses Relay Strategy $m_i$ at $R_i$ based on the forward channel parameters, and Relay Strategy $(0,n_i)$ at $R_i$ uses Relay Strategy $n_i$ based on the backward channel parameters. For the remaining strategies $(m_i,n_i)\in \{(2,2), (2,6), (6,2), (6,6)\}$ at $R_i$, we use the combination of the repetitions suggested by Relay Strategies $m_i$ based on the forward channel parameters, and $n_i$ based on the backward channel parameters. If these two repetitions happen at the same level, we sum these modulo 2. However, there are some modifications to account for repetitions adding to zero modulo 2, or multiple repetitions due to different strategies at the relays. The modifications are described as follows.
\begin{enumerate}
\item If the repetitions happen in the same relay, i.e., $m_1=n_1=0$ or $m_2=n_2=0$: In case the repetition of a particular signal by both the forward and backward strategies is suggested at the same level, we send the repeated signal. If different repeated signals are suggested at a particular level, we send the sum of these two signals modulo two.

\item If the repetitions happen in different relays, i.e., $m_1=n_2=0$ or $m_2=n_1=0$:
    \begin{enumerate}
    \item In case that the repetitions of some streams from two relays are from the same level and are repeated on the same level at node $B$ (ignoring the backward signal component) $R_i$ skips repetitions at the corresponding levels if the forward channel is of Case 4.1.2 Type $i$ and $R_{\bar i}$ skips repetitions at the corresponding levels if the forward channel is of Case 3.1.2 Type $i$.
    \item In case that the repetitions of some streams from two relays are from the same level and are repeated on the same level at node $A$ (ignoring the forward signal component) $R_i$ skips repetitions at the corresponding levels if the backward channel is of Case 4.1.2 Type $i$ and $R_{\bar i}$ skips repetitions at the corresponding levels if the backward channel is of Case 3.1.2 Type $i$.
    \end{enumerate}
\end{enumerate}

We use the same transmission strategy as in Section \ref{one} for channel of both Cases 3.1.2 and 4.1.2. When the forward channel is of Case 3.1.2, node $A$ transmits ${[a_{C_{AB}},...,a_{1}]}^T$ for the forward channel and when the forward channel is of Case 4.1.2, node $A$ transmits ${[\underset{n_{A1}-(n_{1B}-n_{2B})}{\underbrace{0,...,0}}
,a_{n_{1B}},...,a_{n_{2B}+1},\underset{n_{A2}-(n_{A1}+n_{2B})}{\underbrace{0,...,0}}
,a_{n_{2B}},...,a_{1}]}^T$ for the forward channel. Also, similarly, when the backward channel is of Case 3.1.2, node $B$ transmits ${[b_{C_{BA}},...,b_{1}]}^T$ for the backward channel and when the backward channel is of Case 4.1.2, node $B$ transmits ${[\underset{n_{B1}-(n_{1A}-n_{2A})}{\underbrace{0,...,0}}
,b_{n_{1A}},...,b_{n_{2A}+1},\underset{n_{B2}-(n_{B1}+n_{2A})}{\underbrace{0,...,0}}
,b_{n_{2A}},...,b_{1}]}^T$ for the backward channel.



For Case 3.1.2 Type 1, all the messages can be decoded with the same order of decoding similar to the one in Relay Strategy 2 in Section \ref{one} based on the partitioning of the parameter space into nine parts as shown in Figure \ref{fig:sup1}. Also for Case 4.1.2 Type 2, all the messages can be decoded with the same order of decoding similar to the one in Relay Strategy 6 in Appendix \ref{apdx_sc2} based on the partitioning of the parameter space into seven parts as shown in Figure \ref{fig:sup2}. Type 2 cases can be explained similarly. The complete proof can be seen in Appendix \ref{apdx_sc3}.

\begin{example}
Consider the case $(n_{A1},n_{A2},n_{1B},n_{2B},n_{B1},n_{B2},n_{1A},n_{2A})=(6,4,5,7,6,8,7,5)$. With these parameters, the forward channel is of Case 3.1.2 Type 1, and the backward channel is of Case 4.1.2 Type 1. $A$ transmits ${[a_{C_{AB}},...,a_{1}]}^T$ and $B$ transmits ${[\underset{\max\{n_{B1},n_{B2}\}-C_{BA}}{\underbrace{0,...,0}},b_{C_{BA}},...,b_{1}]}^T$. Relay $R_2$ uses Relay Strategy (2,0) and $R_1$ uses Relay Strategy (0,6). See in Figure \ref{fig:Example.S3} that the messages can be decoded by both nodes $A$ and $B$.
\end{example}


\begin{figure}[htbp]
\centering
\subfigure[Transmission to relays.]{
	\includegraphics[width=8.5cm]{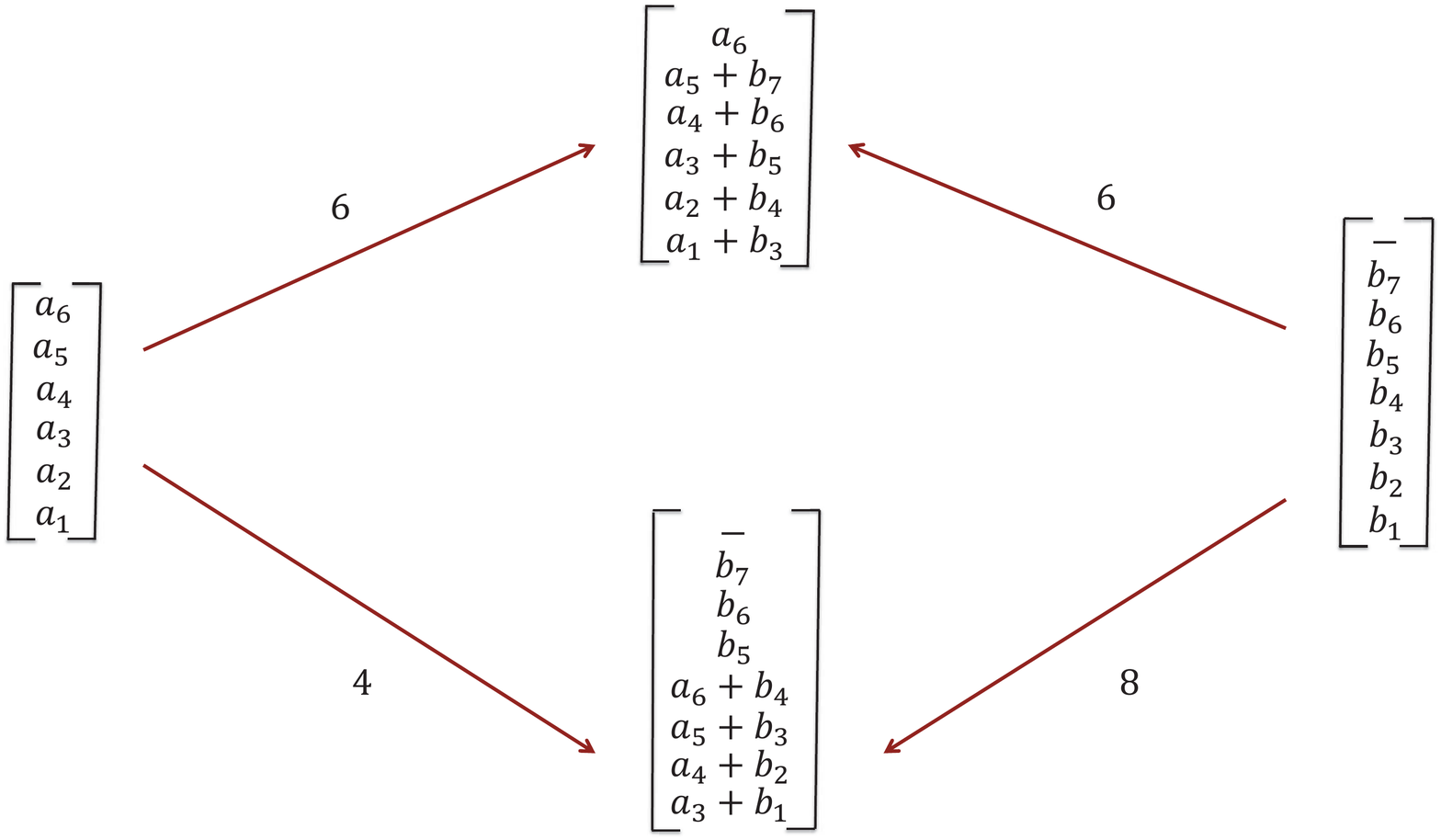}
\label{fig:subfigS3a}
}
\subfigure[Reception from relays.]{
	\includegraphics[width=8.5cm]{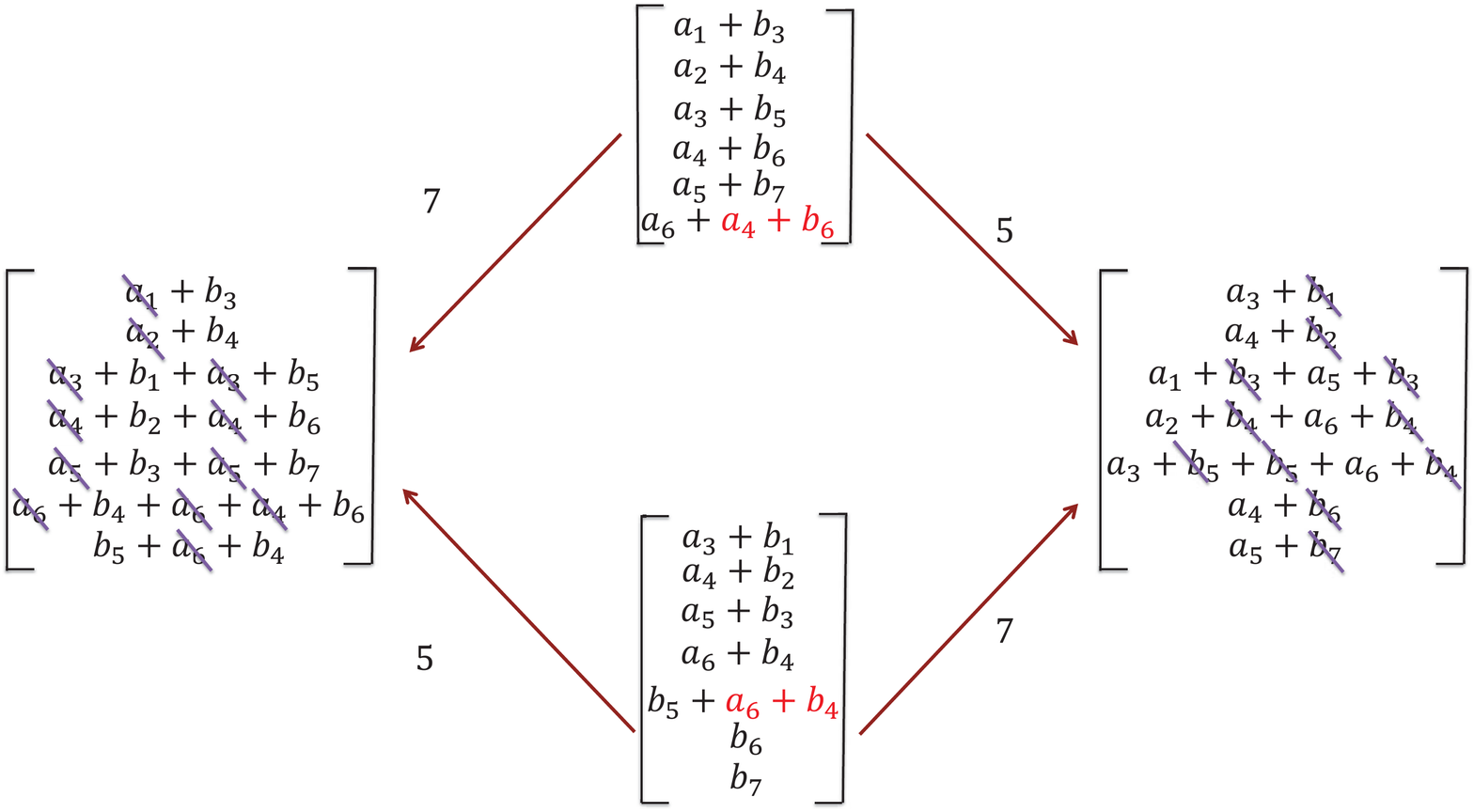}
\label{fig:subfigS3b}
}
\caption[Optional caption for list of figures]{Example for $(n_{A1},n_{A2},n_{1B},n_{2B},n_{B1},n_{B2},n_{1A},n_{2A})=(6,4,5,7,6,8,7,5)$. The red part in the transmission from $R_1$ is due to the repeat strategy for backward channel ($b_1$) and the red part in the transmission from $R_2$ is due to the repeat strategy for forward channel ($a_2$).}
\label{fig:Example.S3}
\end{figure}

\section{Gaussian Diamond Channels}

In this section we first present the Gaussian diamond channel model. Then, we present our results on the capacity of the two-way Gaussian relay channel and a special case of the two-way Gaussian diamond channel.


\subsection{System Model}

A two-way Gaussian diamond channel consists of two nodes $A$ and $B$ who wish to communicate to each other through two relays $R_1$ and $R_2$. We assume there is no direct link between $A$ and $B$ and between $R_1$ and $R_2$. The channels are assumed time-invariant and known to all nodes, the channel gains from node $i$ to the $R_j$ is denoted by $h_{ij}$ and the channel gain from the $R_j$ to node $i$ to is denoted by $h_{ji}$ for $i\in\{A,B\}$ and $j\in\{1,2\}$. The received signals at the relays are given by:
\begin{eqnarray}
Y_1(t)&=&h_{A1}X_{A}(t)+h_{B1}X_B(t)+Z_1(t),\\
Y_2(t)&=&h_{A2}X_{A}(t)+h_{B2}X_B(t)+Z_2(t),
\end{eqnarray}
where $X_A(t), X_B(t)\in\mathbb{C}$ are the transmitted signals from nodes $A$ and $B$, respectively. $Z_1(t), Z_2(t)\sim \mathbf{CN}(0, 1)$ are i.i.d. Gaussian noise at the relays. The received signals at the nodes are given by:
\begin{eqnarray}
Y_A(t)&=&h_{A1}X_1(t)+h_{A2}X_2(t)+Z_A(t),\\
Y_B(t)&=&h_{B2}X_2(t)+h_{B1}X_1(t)+Z_B(t).
\end{eqnarray}
where $X_j(t)\in \mathbb{C}$, $j\in \{1, 2\}$ is the transmitted signal from $R_j$. $Z_A(t), Z_B(t)\sim \mathbf{CN}(0, 1)$ are i.i.d. Gaussian noise at nodes $A$ and $B$, respectively. We have the following power constraints:
\begin{eqnarray}\label{pp}
E\left({|X_i(t)|}^2\right) \le 1,
\end{eqnarray}
for $i \in \{A, B, 1, 2\}$. Let $R_A$ and $R_B$ be the data rates of nodes $A$ and $B$, respectively. In a period consisting of $N$ channel symbols, node $A$ wants to send one of the $2^{NR_A}$ codewords to node $B$, and node $B$ wants to send one of the $2^{NR_B}$ codewords to node $A$. A ($2^{NR_A}$, $2^{NR_B}$, $N$) code for the two-way Gaussian diamond channel consists of two message sets $M_A = \{1, 2, ..., 2^{NR_A}\}$ and $M_B = \{1, 2, ..., 2^{NR_B}\}$, two encoding functions at each time $t$ as
\begin{eqnarray}
f_{it} : (M_i, Y_i^{t-1}) \rightarrow {\mathbb{C}}^N, i \in {A, B},
\end{eqnarray}
two relay functions at each time $t$ as
\begin{eqnarray}
\phi_{jt} :  \mathbb{C}^N \rightarrow  \mathbb{C}^N, j \in {1, 2},
\end{eqnarray}
and two decoding functions
\begin{eqnarray}
g_A : {\mathbb{C}}^N \times M_A \rightarrow M_B,
g_B : {\mathbb{C}}^N \times M_B \rightarrow M_A.
\end{eqnarray}
For $i = A, B$, node $i$ transmits the codeword $f_i(m_i)$, where $m_i$ is the message to be transmitted. For $j = 1, 2$, relay $j$ applies the function $\phi_j$ to its received signal and transmits the resulting signal. Let the received signals at the nodes $A$ and $B$ be $Y^N_A$ and $Y^N_B$, respectively, where the superscript $N$ denotes a sequence of length $N$. We note that the decoding function $g_i$ uses the message from node $i$ as input as well. We say that a decoding error occurs if $g_A(Y^N_A,m_A)\neq m_B$ or $g_B(Y^N_B
,m_B) \neq m_A$. The average probability of error is $P^N_e =\frac{1}{|M_A||M_B|}\times\sum_{(m_A,m_B)\\
\in M_A\times M_B} Pr\{g_A(Y^N_A,m_A)\neq m_B, \ \text{or} \ g_B(Y^N_B,m_B)\neq m_A|(m_A,m_B) \ \text{is sent}\}$.

A rate pair $(R_A, R_B)$ is said to be achievable if there exists a sequence of ($2^{NR_A}$, $2^{NR_B}$, $N$) codes, satisfying the power constraints in \eqref{pp} with $P^N_e \rightarrow 0$ as $N \rightarrow \infty$. The capacity region is the convex hull of all achievable rate pairs $(R_A,R_B)$. The two-way Gaussian diamond channel is characterized by the set of channel parameters $(h_{A1},h_{A2},h_{1B},h_{2B},h_{B1},h_{B2},h_{1A},h_{2A})$.


\subsection{Results for Two-Way Gaussian Relay Channel}

In a diamond channel, if channel gains to and from one relay are zero, we have a two-way relay channel. There are several works on two-way Gaussian relay channels. In \cite{s6,s7}, a deterministic approach was used to achieve the information theoretic cut-set bound \cite{Cover2} within 3 bits for each user. Later, in \cite{i6}, the achievable rate region is within 1 bit from the capacity region for each user for the special case that the channels from the relay to the nodes are the same. The achievability scheme in \cite{i6} is composed of nested lattice codes for the uplink and structured binning for the downlink. Their codes utilize two different shaping lattices for source nodes based on a three-stage lattice partition chain to satisfy their different transmit power constraints. Here we propose a simpler achievability scheme for a general two-way Gaussian relay channel compared to \cite{i6,s7}. Define $h_{AB}\triangleq\min\{|h_{A1}|,|h_{1B}|\}$ and $h_{BA}\triangleq\min\{|h_{B1}|,|h_{1A}|\}$. By symmetry we can assume $|h_{BA}|\le|h_{AB}|$. $W_A$ and $W_B$ are the messages of nodes $A$ and $B$, respectively, that they want to convey to the other node. We divide the message from $A$ to $B$ to two parts, as $W_A=(W_{A1}, W_{A2})$.

\begin{theorem}\label{thm_rel.g}
For the two-way Gaussian relay channel with the parameters of $({h_{A1},h_{B1},h_{1A},h_{1B}})$, the capacity region is outer-bounded by the following region
\begin{eqnarray}\label{r-GG}
{R}_{AB}&\le&\min\{\log(1+{|h_{A1}|}^2),\log(1+{|h_{1B}|}^2)\}=\log(1+{h_{AB}}^2),\nonumber\\
{R}_{BA}&\le&\min\{\log(1+{|h_{B1}|}^2),\log(1+{|h_{1A}|}^2)\}=\log(1+{h_{BA}}^2).
\end{eqnarray}
Furthermore, if $h_{AB}=h_{BA}$, this region is achievable within 1 bit for each user and otherwise assuming $h_{AB}<h_{BA}$, then this region is achievable within 1 bit for the $A\rightarrow B$ direction and within 2 bits for the $B\rightarrow A$ direction.
\end{theorem}
\begin{proof}
The above outer-bound results from the cut-set bound. In order to encode $W_{A1}$ and $W_{B}$, we use the common lattice code $\Lambda=\Lambda_f \cap \nu_c$\footnote{
We use a nested lattice code \cite{Zamirr} which is generated using a quantization lattice for shaping and a channel coding lattice. We have $T$-dimensional nested lattices $\Lambda_c\subseteq\Lambda_f$, where $\Lambda_c$ is a quantization lattice with $\sigma^2(\Lambda_c)=1$ and $G(\Lambda_c)\approx1/2\pi e$, and $\Lambda_f$ is a good channel coding lattice. We construct a codebook $\Lambda=\Lambda_f \cap \nu_c$, where $\nu_c$ is the Voronoi cell of the lattice $\Lambda_c$. We will use the following properties of lattice codes \cite{Mohajer}:
\begin{enumerate}
\item Codebook $\Lambda$ is a closed set with respect to summation under the ``$mod \ \Lambda_c$" operation, i.e., if $x_1, x_2 \in {\Lambda}$ are two codewords, then $(x_1+x_2) \ mod \ \Lambda_c \in\Lambda$ is also a codeword.
\item Lattice code $\Lambda$ can be used to reliably transmit up to rate $R=\log(\mathsf {SNR})$ over a Gaussian channel modeled by $Y=\sqrt{\mathsf {SNR}}X+Z$ with ${\mathbb{E}}[Z^2] = 1$, while a more sophisticated scheme can achieve rate $R=\log(1+\mathsf {SNR})$.
\end{enumerate}
For more detail refer to \cite{Lattice1,Lattice2,Lattice3}.
}. Let $s_A$ be the lattice codeword to which $W_{A1}$ is mapped and $s_B$ be the lattice codeword to which $W_{B}$ is mapped, and define $X_L=s_A+s_B$. We use the signal $X_A^{(2)}$ to encode $W_{A2}$ which is Gaussian with unit power.

Once the encoding process is performed, the signal transmitted by $A$ is formed as $X_A=\sqrt{\alpha_1}c_A+\sqrt{1-\alpha_1}X_A^{(2)}$ where $0\le \alpha_1\le 1$ and $c_A=[s_A-d_A]\mod\Lambda_c$ with unit power, and $d_A$ is a random dither uniformly distributed over $\nu_c$, and shared between both transceivers and both relays. Also, transceiver $B$ sends the signal $\frac{|h_{BA}|}{|h_{B1}|}c_B$ where $c_B=[s_B-d_B]\mod\Lambda_c$, and $d_B$ is a random dither uniformly distributed over $\nu_c$, and shared between both transceivers and both relays. We also choose $\alpha_1=\frac{{|h_{BA}|}^2}{{|h_{A1}|}^2}$ so that $c_A$ and $c_B$ arrive at the relay with the same power and add together as a lattice code. Messages from $s_A$ and $s_B$ are being sent with rate $R_u$ and messages from $X_A^{(2)}$ is being sent with rate $R_v$. So, $R_{AB}=R_{u}+R_{v}$ and $R_{BA}=R_{u}$.

The relay receives $Y_R=|h_{BA}|(c_A+c_B)+\sqrt{\left({{|h_{A1}|}^2}-{{|h_{BA}|}^2}\right)}X_A^{(2)}+Z_R$. The signal $X_A^{(2)}$ (which is Gaussian) can be decoded by treating the rest ($c_A$ and $c_B$) as noise provided that:
\begin{eqnarray}\label{r10}
R_{v} &\le& \log\left(1+\frac{\left(1-\frac{{|h_{BA}|}^2}{{|h_{A1}|}^2}\right)|h_{A1}|^2}{1+2|h_{BA}|^2}\right).
\end{eqnarray}

Recall that $X_L=[s_A+s_B \mod \Lambda_c]=[c_A+c_B+(d_A+d_B) \mod \Lambda_c] \in \Lambda$. So it can be decoded from the received signal after subtracting the signal $X_A^{(2)}$, provided that
\begin{eqnarray}\label{r11}
R_{u} &\le& \log\left(|h_{BA}|^2\right).
\end{eqnarray}

Then, we use a structured binning for the transmission from the relay to nodes $A$ and $B$. We generate $2^{nR_u}$ length-$n$ sequences with each element i.i.d. according to $\mathbf{CN}(0,1)$. These sequences form a codebook $\Lambda_R$. We assume one-to-one correspondence between each $t\in \Lambda_A$ and a codeword $X_R \in \Lambda_R$. To make this correspondence explicit, we use the notation $X_R(t)$. After the relay decodes $\hat{X_L}$, it transmits $X_R(\hat{X_L})$ at the next block to nodes $A$ and $B$. $\hat{X_L}$ is uniform over $\Lambda_A$, and, thus, $X_R(\hat{X_L})$ is also uniformly chosen from $\Lambda_R$. Then, relay sends out $\sqrt{\alpha_2}X_R+\sqrt{1-\alpha_2}X_A^{(2)}$, $0\le\alpha_2\le1$.

Node $B$ can decode $X_L$ taking $X_A^{(2)}$  as noise (while we send them both as Gaussian signals this time) if:
\begin{eqnarray}\label{r12}
R_{u} &\le& \log\left(1+\frac{\alpha_2 |h_{1B}|^2}{1+(1-\alpha_2)|h_{1B}|^2}\right).
\end{eqnarray}
Then, $B$ can decode $X_A^{(2)}$ after decoding $X_L$ if:
\begin{eqnarray}\label{r14}
R_{v} &\le& \log\left(1+(1-\alpha_2) |h_{1B}|^2\right).
\end{eqnarray}
Also, $A$ can decode $X_L$ taking $X_A^{(2)}$ as noise if:
\begin{eqnarray}\label{r13}
R_{u} &\le& \log\left(1+\frac{\alpha_2 |h_{1A}|^2}{1+(1-\alpha_2)|h_{1A}|^2}\right).
\end{eqnarray}
Now, we show that we can achieve the capacity within 1 bit for $B\rightarrow A$ direction and within 2 bits for $A\rightarrow B$ direction, by showing $R_{u}^{opt}\ge \log(1+{|h_{BA}|}^2)-1$ and $R_{u}^{opt}+R_{v}^{opt}\ge \log(1+{|h_{AB}|}^2)-2$. In other words, it is enough to show that $R_{u}^{opt}\ge \log(1+{|h_{BA}|}^2)-1$ and $R_{v}^{opt}\ge \log(1+{|h_{AB}|}^2)-\log(1+{|h_{BA}|}^2)-1$. We also assume that all the links have $|h_{ij}|\ge 1$ otherwise we take it as zero and that direction does not send. We only need to prove that the above equations satisfy the claimed gap. For all the equations in \eqref{r10}-\eqref{r13}, we need to show that the RHS for those ones with $R_u$, is $\ge \log(1+{|h_{BA}|}^2)-1$ and RHS for those ones with $R_v$, is $\ge \log(1+{|h_{AB}|}^2)-\log(1+{|h_{BA}|}^2)-1$. Thus, we need to show the following.
\begin{eqnarray}
&&{\text {RHS of \eqref{r10}: \ \ }}\log\left(1+\frac{\left(1-\frac{{|h_{BA}|}^2}{{|h_{A1}|}^2}\right)|h_{A1}|^2}{1+2|h_{BA}|^2}\right)\ge \log(1+{|h_{AB}|}^2)-\log(1+{|h_{BA}|}^2)-1,\label{r110}\\
&&{\text {RHS of \eqref{r11}: \ \ }}\log\left(|h_{BA}|^2\right)\ge \log(1+{|h_{BA}|}^2)-1,\label{r111}\\
&&{\text {RHS of \eqref{r12}: \ \ }}\log\left(1+\frac{\alpha_2 |h_{1B}|^2}{1+(1-\alpha_2)|h_{1B}|^2}\right)\ge \log(1+{|h_{BA}|}^2)-1,\label{r112}\\
&&{\text {RHS of \eqref{r14}: \ \ }}\log\left(1+(1-\alpha_2) |h_{1B}|^2\right)\ge \log(1+{|h_{AB}|}^2)-\log(1+{|h_{BA}|}^2)-1.\label{r114}\\
&&{\text {RHS of \eqref{r13}: \ \ }}\log\left(1+\frac{\alpha_2 |h_{1A}|^2}{1+(1-\alpha_2)|h_{1A}|^2}\right)\ge \log(1+{|h_{BA}|}^2)-1,\label{r113}
\end{eqnarray}
where Eqs. \eqref{r110} and \eqref{r111} trivially hold. \eqref{r112}, \eqref{r113} and \eqref{r114} are, respectively, equivalent to:
\begin{eqnarray}
\alpha_2 &\ge& \frac{(1+{|h_{1B}|}^2)(-1+{|h_{BA}|}^2)}{({|h_{1B}|}^2)(1+{|h_{BA}|}^2)},\label{ra}\\
\alpha_2 &\ge& \frac{(1+{|h_{1A}|}^2)(-1+{|h_{BA}|}^2)}{({|h_{1A}|}^2)(1+{|h_{BA}|}^2)},\label{rb}\\
\alpha_2 &\le& \frac{2(1+{|h_{1B}|}^2)(1+{|h_{BA}|}^2)-(1+{|h_{AB}|}^2)}{2({|h_{1B}|}^2)(1+{|h_{BA}|}^2)}.\label{rc}
\end{eqnarray}
If we name $h_{1}=\min\{|h_{1A}|,|h_{1B}|\}$, and $f(x)=\frac{(1+{x}^2)(-1+{|h_{BA}|}^2)}{({x}^2)(1+{|h_{BA}|}^2)}$, we can see that $1\ge f(h_1)\ge f(|h_{1A}|), f(|h_{1A}|)$. So, $\alpha_2=f(h_{1})$ satisfies \eqref{ra} and \eqref{rb}. Also, $\alpha_2=f(h_{1})$ satisfies \eqref{rc}. This completes the proof of the Theorem.
\end{proof}

\begin{corollary} \label{re1.xx}
For a Gaussian diamond relay channel the following sum-rate is achievable by using only the strongest path in each direction
\begin{equation}
{R}_{AB}+{R}_{BA}=\max_{i\in\{1,2\}}\{\min\{\log(1+{|h_{Ai}|}^2),\log(1+{|h_{iB}|}^2)\}\}+\max_{j\in\{1,2\}}\{\min\{\log(1+{|h_{Bj}|}^2),\log(1+{|h_{jA}|}^2)\}\}-3
\end{equation}
\end{corollary}
\begin{proof}
It follows from Theorem \ref{thm_rel.g}.
\end{proof}

\begin{remark}
We used a strategy similar to Relay Strategy 0 which was introduced for the deterministic channel in Section \ref{neitherr}. The transmitter of the direction with higher rate divides its power for two signals. It sends a signal that combined with the signal received from the other transmitter forms a lattice code at the relay, and the rest of the power is allocated to the other signal. Then the relay performs a reverse amplify and sends the lattice code on a higher power which will be decoded by both receivers and the other signal on a lower power which will be decoded only by one of the receivers.
\end{remark}

\subsection{Results for Two-Way Gaussian Diamond Model}

In this subsection, we will give an achievability scheme for a symmetric case of two-way Gaussian diamond channel with parameters $(h_{A1},h_{A2},h_{1B},h_{2B})=(h_{B1},h_{B2},h_{1A},h_{2A})=(a,b,c,d)$ (Figure \ref{fig:dosar}) that satisfy $\log(1+{|a|}^2) \ge \log(1+{|c|}^2)+\log(1+{|b|}^2)$ and $\log(1+{|d|}^2)\ge \log(1+{|c|}^2)+\log(1+{|b|}^2)$, which achieves the capacity region within a constant number of bits in each direction as in the following theorem.

\begin{figure}[htbp]
\centering
	\includegraphics[width=15cm]{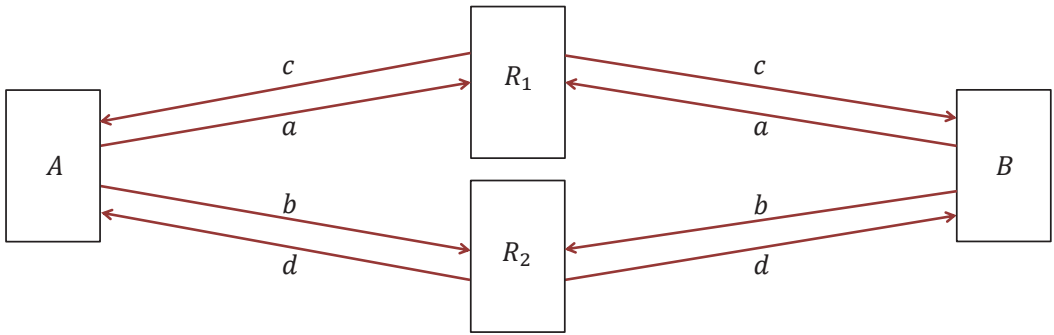}
\caption{The investigated reciprocal two-way Gaussian diamond channel.}
\label{fig:dosar}
\end{figure}

\begin{theorem}\label{thm_detgg}
For the  two-way symmetric Gaussian diamond channel in Figure \ref{fig:dosar} that satisfy $\log(1+{|a|}^2) \ge \log(1+{|c|}^2)+\log(1+{|b|}^2)$ and $\log(1+{|d|}^2)\ge \log(1+{|c|}^2)+\log(1+{|b|}^2)$, the capacity region is outer-bounded by the following region
\begin{eqnarray}
{R}_{AB}, R_{BA}&\le&\min\{\log(1+{|a|}^2+{|b|}^2),\log(1+{|c|}^2+{|d|}^2),\nonumber\\
&&\log(1+{|c|}^2)+\log(1+{|b|}^2),\log(1+{|d|}^2)+\log(1+{|a|}^2)\}\nonumber\\ &=&\log(1+{|c|}^2)+\log(1+{|b|}^2).\label{r-d1ggg}
\end{eqnarray}
Further, this region is achievable within 4 bits in each direction.
\end{theorem}

\begin{proof}
The above outer-bound results from the cut-set bound.
Take $W_A$ and $W_B$ as the messages of the nodes $A$ and $B$, respectively, that they want to convey to the other node. We divide the messages $W_A$ and $W_B$ into two parts, as $W_A=(W_{A1}, W_{A2})$ and $W_B=(W_{B1}, W_{B2})$.

In order to encode $W_{Ai}$ and $W_{Bi}$, where $i\in \{1, 2\}$, we use the common lattice code $\Lambda^i$ defined in the last subsection. Let $s_{A_i}$ be the lattice codeword to which $W_{Ai}$ is mapped and $s_{B_i}$ be the lattice codeword to which $W_{Bi}$ is mapped, and define $X_{L_i}=s_{A_i}+s_{B_i}$.

Once the encoding process is performed, the signal transmitted by $u$ is formed as $X_u=\sqrt{\alpha_2} c_{u_1}+\sqrt{1-\alpha_2} c_{u_2}$ where $0\le \alpha_2\le 1$ and $c_{u_i}=[s_{u_i}-d_{u_i}]\mod\Lambda_c$, and $d_{u_i}$ is a random dither uniformly distributed over $\nu_c$, and shared between all the terminals in the network for $u\in\{A,B\}$. We also set $\alpha_2=\frac{{|b|}}{\sqrt{1+{|h_{A2}|}^2}}$ which is equivalent to ${|b|}=\frac{\alpha_2}{\sqrt{1-\alpha_2}}$. The rate of $s_{A_1}$ and $s_{B_1}$ is being shown by $R_u$ and the rate of $s_{A_2}$ and $s_{B_2}$ is being shown by $R_v$. So, $R_{AB}=R_{BA}=R_{u}+R_{v}$. We can decode both $s_{A_1}$ and $s_{A_2}$ in $B$ and $s_{B_1}$ and $s_{B_2}$ in $A$ with the following strategy.

Recall that $X_{L_i}=[s_{A_i}+s_{B_i} \mod \Lambda_c]=[c_{A_i}+c_{B_i}+(d_{A_i}+d_{B_i}) \mod \Lambda_c] \in \Lambda^i$. $R_1$ can decode $X_{L_1}$ and $X_{L_2}$ by successive interference cancellation if \eqref{g2}-\eqref{g3} holds.

\begin{eqnarray}
R_u&\le&\log\left(\frac{\alpha^2 {|a|}^2}{2(1-\alpha^2) {|a|}^2+1}\right),\label{g2}\\
R_v&\le&\log\left({(1-\alpha^2) {|a|}^2}\right),\label{g3}
\end{eqnarray}
Also, $R_2$ decodes $X_{L1}$ considering $X_{L_2}$ as noise, as long as \eqref{g1} holds.
\begin{eqnarray}
R_u&\le&\log\left(\frac{\alpha^2 {|b|}^2}{{(1-\alpha^2) {|b|}}^2+1}\right),\label{g1}
\end{eqnarray}

Then, we use a structured binning for the transmission from the relays to the nodes $A$ and $B$. We generate $2^{nR_i}$ $n$-sequences with each element i.i.d. according to $\mathbf{CN}(0,1)$, for $i\in\{u,v\}$. These sequences form a codebook $\Lambda^i_{R}$. We assume one-to-one correspondence between each $t\in \Lambda_{A_i}$ and a codeword $X_{R_j} \in \Lambda^i_{R}$. To make this correspondence explicit, we use the notation $X_{R_j}(t)$ for the $R_j$.

After $R_1$ decodes $\hat X_{L_1}$ and $\hat X_{L_2}$, and $R_2$ decodes $\hat X_{L_1}$, for $i, j \in \{1, 2\}$, $i \neq j$, $R_j$ transmits $X_{R_j}(\hat X_{L_i})$ at the next block to nodes $A$ and $B$. $\hat X_{L_i}$ is uniform over $\Lambda_{A_i}$, and, thus, $X_{R_j}(\hat X_{L_i})$ is also uniformly chosen from $\Lambda^i_R$. Then, $R_j$ sends out $X_{R_j}$.

The node $B$ decodes $X_{L_1}$ with low probability of error as long as \eqref{g4} holds by considering $X_{L_2}$ as noise. Then it decodes $X_{L_2}$ with low probability of error as long as \eqref{g5} holds.
\begin{eqnarray}
R_u&\le&\log\left(1+\frac{{|d|}^2}{{|c|}^2}\right),\label{g4}\\
R_v&\le&\log\left(1+{|c|}^2\right).\label{g5}
\end{eqnarray}

Now, we show all the bounds \eqref{g2}-\eqref{g5} satisfy the four bit gap to the outer bounds as in the statement of the theorem.
\begin{eqnarray}
{\text {RHS of \eqref{g2}:}} && \log\left(\frac{\alpha^2 {|a|}^2}{2(1-\alpha^2) {|a|}^2+1}\right)\nonumber\\
&\ge&\log\left(\frac{\alpha^2 {|a|}^2}{3(1-\alpha^2) {|a|}^2}\right)\nonumber\\
&=&\log\left(\frac{\alpha^2 {|a|}^2}{(1-\alpha^2) {|a|}^2}\right)-\log 3 \nonumber\\
&=&\log\left({|b|}^2\right)-\log 3 \nonumber\\
&\ge&\log\left(1+{|b|}^2\right)-1-\log 3.\\
{\text {RHS of \eqref{g3}:}} &&\log\left({(1-\alpha^2) {|a|}^2}\right)\nonumber\\
&\ge&\log\left({\frac{1}{1+{|b|}^2} {|a|}^2}\right)\nonumber\\
&=&\log\left({|a|}^2\right)-\log\left(1+{|b|}^2\right)\nonumber\\
&\ge&\log\left(1+{|a|}^2\right)-\log\left(1+{|b|}^2\right)-1\nonumber\\
&\ge&\log\left(1+{|c|}^2\right)-1.\\
{\text {RHS of \eqref{g1}:}} &&\log\left(\frac{\alpha^2 {|b|}^2}{(1-\alpha^2) {|b|}^2+1}\right)\nonumber\\
&\ge&\log\left(\frac{\alpha^2 {|b|}^2}{2}\right)\nonumber\\
&=&\log\left(\alpha^2 {|b|}^2\right)-1\nonumber\\
&=&\log\left(\frac{{|b|}^4}{{|b|}^2+1}\right)-1\nonumber\\
&=&2\log\left({|b|}^2\right)-\log\left({{|b|}^2+1}\right)-1\nonumber\\
&\ge&\log\left({{|b|}^2+1}\right)-3.\\
{\text {RHS of \eqref{g4}:}} &&\log\left(1+\frac{{|d|}^2}{{|c|}^2}\right)\nonumber\\
&\ge&\log\left(1+{|d|}^2\right)-\log\left({|c|}^2\right)\nonumber\\
&\ge&\log\left(1+{|c|}^2\right)+\log\left(1+{|b|}^2\right)-\log\left({|c|}^2\right)\nonumber\\
&\ge&\log\left(1+{|b|}^2\right).
\end{eqnarray}

This shows that $R_{AB}$ can be achieved within 4 bits of the outer bound for the forward channel by symmetry. It can be also seen that we can achieve $R_{BA}$ within 4 bits of the outer bound for the backward channel.
\end{proof}

\begin{remark}
Due to the fact that this model corresponds to a special case of Case 1 (for both directions) we used a strategy similar to Relay Strategy 0, by using reverse amplify-and-forward.
\end{remark}

\begin{remark}
The authors of \cite{d3} considered a reciprocal two-way diamond channel, and gave an achievable rate region. For the case of $h_i={|h_{Ai}|}={|h_{iA}|}={|h_{Bi}|}={|h_{iB}|}$, for $i\in \{1, 2\}$, with two relay nodes, the achievable sum-rate in \cite{d3} is given by $R_{AB}+R_{BA}=\log(\frac{1}{2}+h_1^2)+\log(\frac{1}{2}+h_2^2)$ which is not within a finite gap from the upper bound $2\log(1+h_1^2+h_2^2)$. However, our achievable region in Corollary \ref{re1.xx} with one relay node can achieve within 5 bits of the cut-set upper bound
\begin{equation}
2\log(1+\max\{h_1^2,h_2^2\})-3=2\log(2+2\max\{h_1^2,h_2^2\})-5\ge 2\log(1+h_1^2+h_2^2)-5.
\end{equation}
\end{remark}

For the case of ${|h_{A2}|}={|h_{2A}|}={|h_{B2}|}={|h_{2B}|}=50$ and ${|h_{A1}|}={|h_{1A}|}={|h_{B1}|}={|h_{1B}|}=5000$, we can see a comparison of our achievable sum-rate in Corollary \ref{re1.xx} and the achievable sum-rate  $R_{AB}+R_{BA}=\log(\frac{1}{2}+h_1^2)+\log(\frac{1}{2}+h_2^2)$ given in \cite{d3}, in Figure \ref{fig:iit}.

\begin{figure}[htbp]
\centering
	\includegraphics[width=9cm]{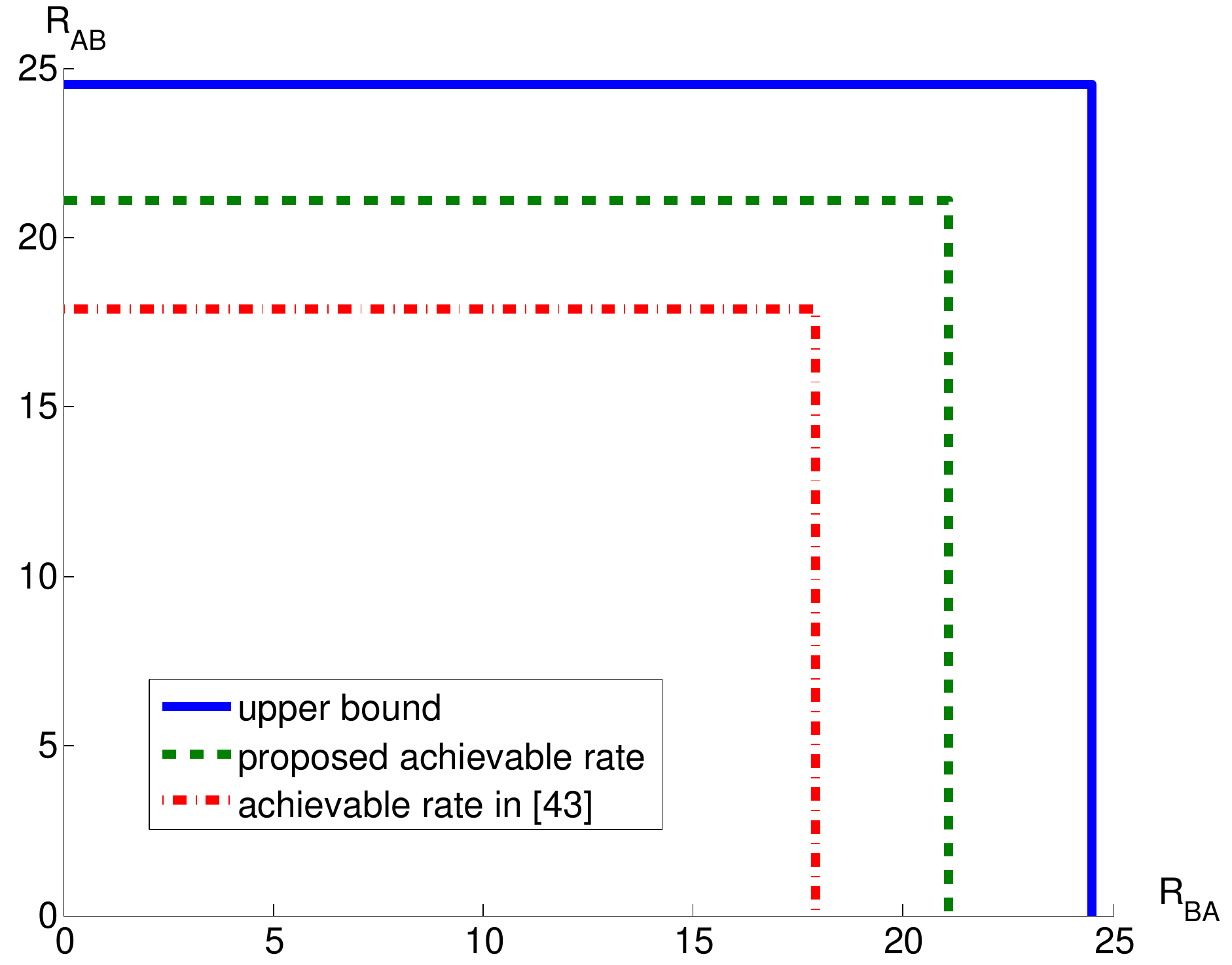}
\caption{An comparison of our results and \cite{d3} for ${|h_{A2}|}={|h_{2A}|}={|h_{B2}|}={|h_{2B}|}=50$ and ${|h_{A1}|}={|h_{1A}|}={|h_{B1}|}={|h_{1B}|}=5000$.}
\label{fig:iit}
\end{figure}

\begin{remark} For general two-way Gaussian diamond channels, an achievability scheme idea is inspired from the deterministic channel results. Each set of streams in deterministic scheme that is being transmitted together with similar interfering properties can be considered as a group. For example, for $(n_{A1},n_{A2},n_{1B},n_{2B},n_{B1},n_{B2},n_{1A},n_{2A})=(n,q,p,m,n,q,p,m)$ where $m, n \ge p, q$ we can make $A_2=[a_p,...,a_1]^T$, $A_1=[a_{p+q},...,a_{p+1}]^T$, $B_2=[b_p,...,b_1]^T$ and $B_1=[b_{p+q},...,b_{p+1}]^T$ each one as a group. $R_1$ receives $[A_1^T+B_1^T,A_2^T+B_2^T]^T$ and $R_2$ receives $[A_1^T+B_1^T]^T$. In deterministic model, we simply send them in the reverse direction but in Gaussian model it is necessary to decode the sums of the groups of received signals as lattice codes in the descending order. One difference of the deterministic model as compared to the Gaussian model is that we may need to set a distance between the groups of streams but there is not such a thing in the Gaussian model. There are some complexities due to which the general Gaussian model seems intractable. The main one is that in Gaussian case, each equivalent group of streams correspondent from deterministic model, is translated into a unique message and there is a power allocation for each one that should be optimized based on the the resulting bounds. For any set of channel parameters $(h_{A1},h_{A2},h_{1B},h_{2B},h_{B1},h_{B2},h_{1A},h_{2A})$, we can use the corresponding achievability scheme in deterministic channel with parameters $(n_{A1},n_{A2},n_{1B},n_{2B},n_{B1},n_{B2},n_{1A},n_{2A})$ by a corresponding relationship from $(\log(1+{|h_{A1}|}^2),\log(1+{|h_{A2}|}^2),\log(1+{|h_{1B}|}^2),\log(1+{|h_{2B}|}^2),\log(1+{|h_{B1}|}^2),\log(1+{|h_{B2}|}^2),\log(1+{|h_{1A}|}^2),\log(1+{|h_{2A}|}^2))$  and decode the sum of each $d$ received signals as a $d$-dimensional nested lattice code \cite{Zamirr} (instead of simply adding them), which at most causes one bit of decrease in the rate of each of the messages included in the lattice codes.
\end{remark}

Exploring the gap for the general two-way Gaussian diamond channel is a case by case analysis, and we leave that as an important next step.

\section{Conclusions}
In this paper, we studied the capacity of the  bidirectional (or two-way) diamond channel with two nodes and two relays. We used the deterministic approach to capture the essence of the problem and to determine capacity-achieving transmission and relay strategies. Depending on the forward and backward channel gains, we used either a reverse amplify-and-forward or a particular modified strategy involving repetitions, and reversing order of some streams at the relays. The proposed scheme is used to find the capacity region within a constant gap in two special cases of the Gaussian diamond channel. First, for the general two-way Gaussian relay channel a smaller gap is achieved compared to the prior works. Then, a special symmetric case of the Gaussian diamond model is considered and capacity region is achieved within 4 bits.

\section*{Acknowledgement}
The authors are very grateful to Srikrishna Bhashyam from the Indian Institute of Technology, Madras for numerous discussions and comments on the drafts that helped the paper significantly.

\begin{appendices}
\section{Transmission Strategies when diamond channel is neither Case 3.1.2 nor 4.1.2} \label{apdx_sc1}

We consider the transmission strategy for different cases as follows. 

Case 1: $C_{AB}=n_{A2}+n_{1B}$: Since we have $C_{AB}=n_{A2}+n_{1B}$, \eqref{r-d1} shows that $n_{A1}, n_{2B}\ge C_{AB}$. We send the data from $A$ as ${[a_{C_{AB}},...,a_{n_{1B}+1},\underset{n_{A1}-(n_{A2}+n_{1B})}{\underbrace{0,...,0}},a_{n_{1B}},...,a_{1}]}^T$. Node $B$ can decode all $C_{AB}$ streams as illustrated in Figure \ref{fig:Examplea}.

\begin{figure}[htbp]
\centering
	\includegraphics[width=15cm]{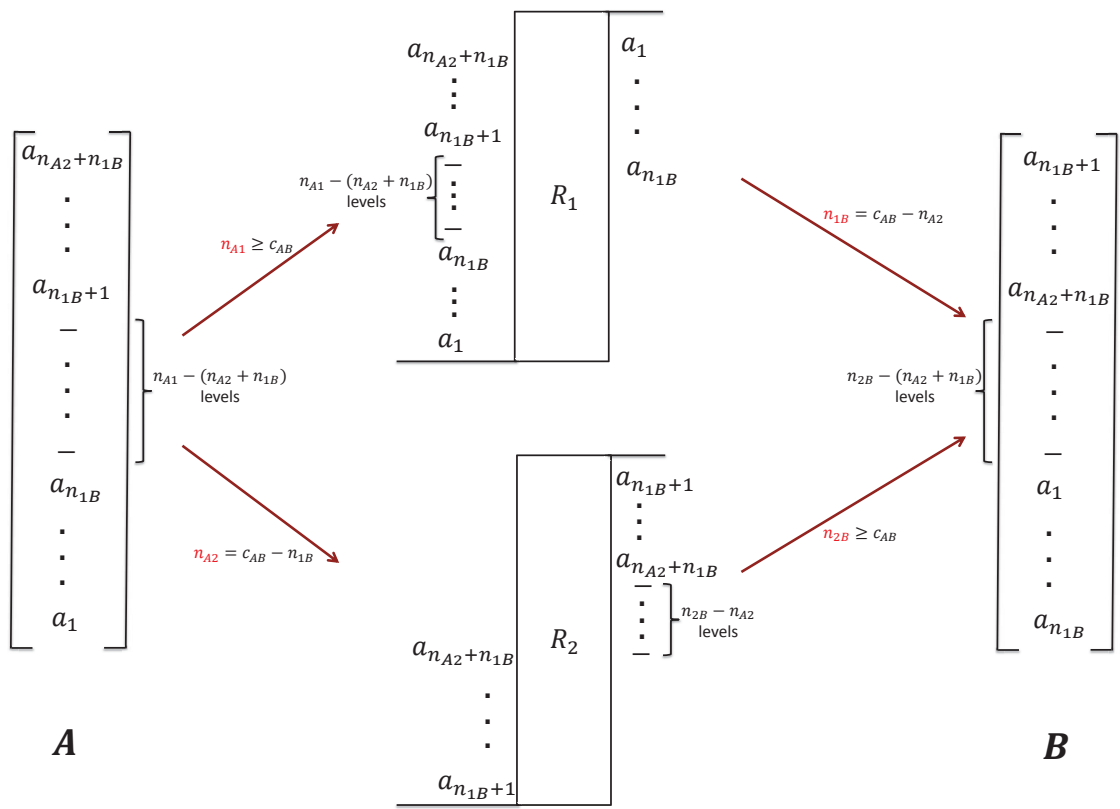}
\caption{Achievability scheme for Case 1 by using Relay Strategy 0.}
\label{fig:Examplea}
\end{figure}

Case 2: $C_{AB}=n_{A1}+n_{2B}$: We send the data from $A$ as ${[a_{C_{AB}},...,a_{n_{2B}+1},0,...,0,a_{n_{2B}},...,a_{1}]}^T$ with $n_{A2}-(n_{A1}+n_{2B})$ zeros in it.
The proof is similar to Case 1, obtained by interchanging $R_1$ and $R_2$ and is thus omitted.

Case 3: $C_{AB}=\max\{n_{A1},n_{A2}\}$: We assume that the channel is of Type 1. For Type 2 the proof is similar. We have $C_{AB}=\max\{n_{A1},n_{A2}\}=n_{A1}$. 

Case 3.1.1: $n_{1B}< C_{AB}$, $n_{2B}\ge n_{A2}+n_{1B}$: Since $C_{AB}=\max\{n_{A1},n_{A2}\}=n_{A1}$ and $n_{1B}< C_{AB}$, \eqref{r-d1} shows that $n_{2B}\ge C_{AB}$. Since $A$ can transmit $C_{AB}$ bits that can be heard by at least one relay, it transmits these bits, and Figure \ref{fig:Exampleb} illustrates that node $B$ can decode the data.

\begin{figure}[htbp]
\centering
	\includegraphics[width=15cm]{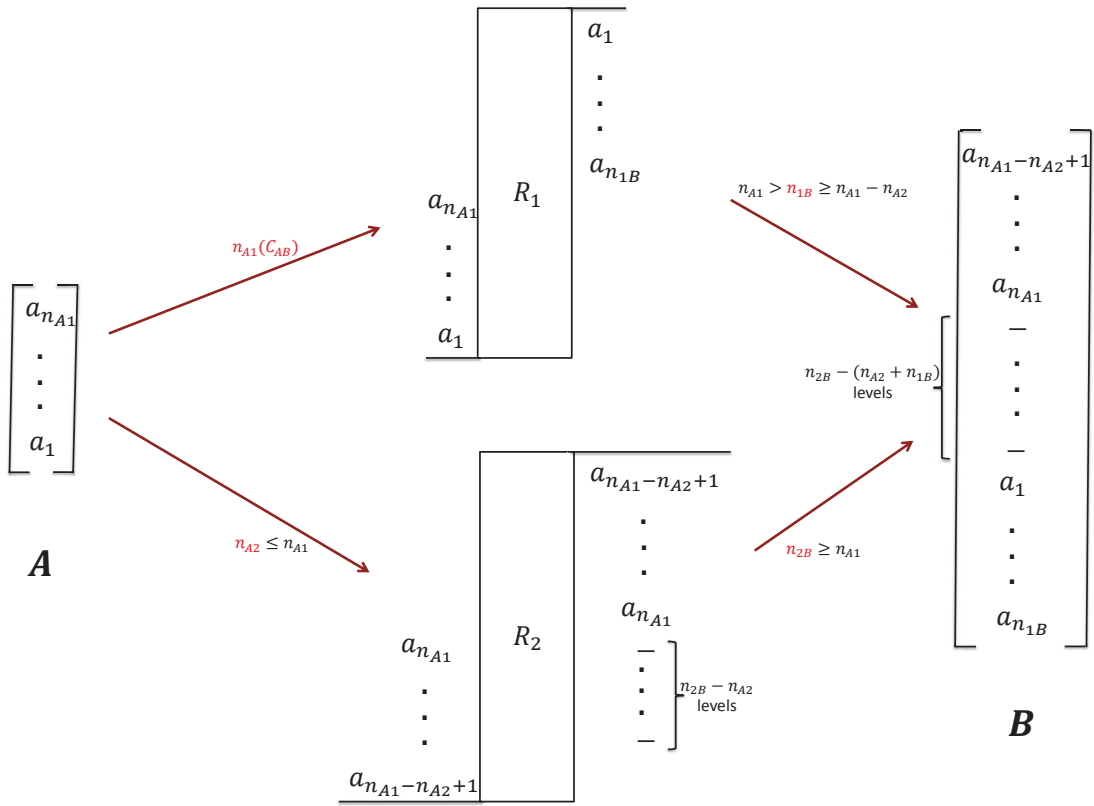}
\caption{Achievability scheme for Case 3.1.1 by using Relay Strategy 0.}
\label{fig:Exampleb}
\end{figure}

Case 3.2: $n_{1B}\ge C_{AB}$: As in Subcase 3.1.1, $A$ can only send $C_{AB}$ streams, and Figure \ref{fig:Exampled} illustrates that node $B$ is able to decode the data. For decoding, if $a_{n_{A1}-n_{A2}+1}$ received from $R_2$ is below all levels of the other relay received at node $B$, i.e., below $a_{n_{A1}}$ from $R_1$, we decode the streams from $R_1$ without interference. Otherwise, while a stream, $a_v$, is the one being added from $R_1$ to $a_{n_{A1}-n_{A2}+1}$ received from $R_2$, if $v>n_{A1}-n_{A2}+1$ we decode the streams starting from the highest level $a_1$ and then subtract them from the signal before decoding the next lower stream, and if $v<n_{A1}-n_{A2}+1$ we decode the streams starting from the lowest level, $a_{n_{A1}}$, and then subtract them from the signal before decoding the next upper stream.

\begin{figure}[htbp]
\centering
	\includegraphics[width=15cm]{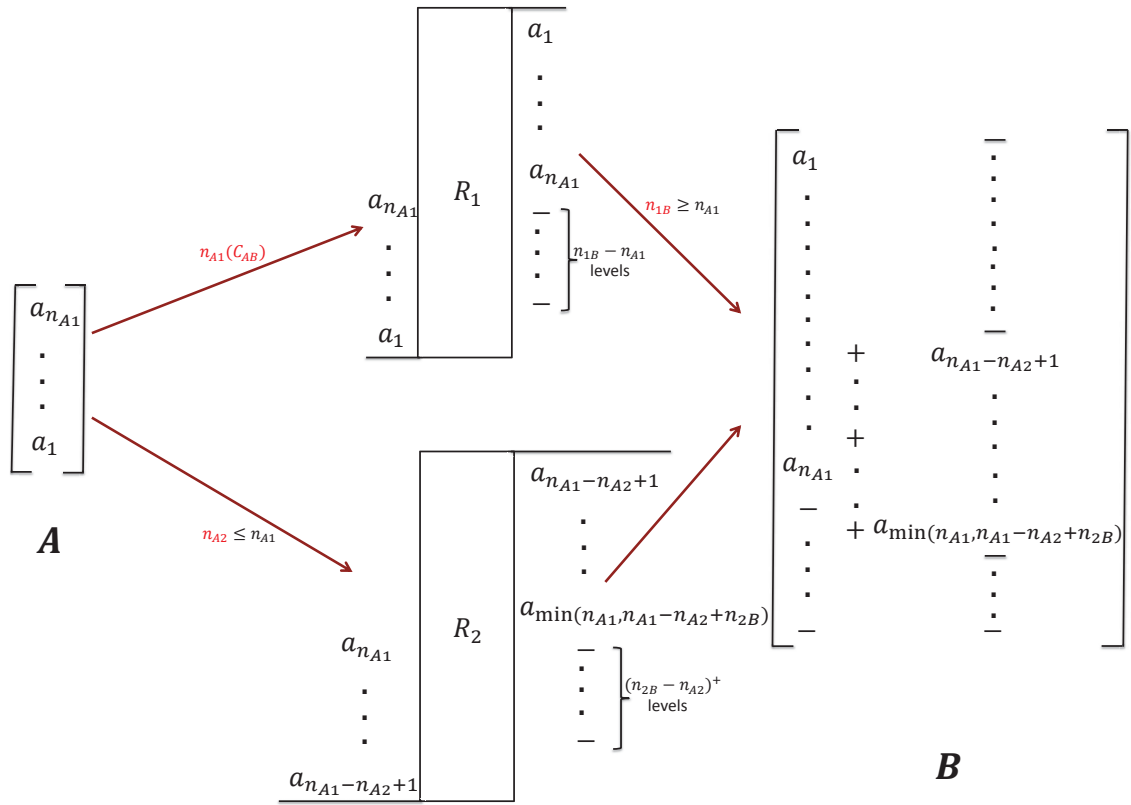}
\caption{Achievability scheme for Case 3.2 by using Relay Strategy 0.}
\label{fig:Exampled}
\end{figure}

Case 4: $C_{AB}=\max\{n_{1B},n_{2B}\}$; We assume that the channel is of Type 1. For Type 2 the proof is similar. We have $C_{AB}=\max\{n_{1B},n_{2B}\}=n_{1B}$. 

Case 4.1.1: $n_{A1}< C_{AB}$, $n_{A2}\ge n_{2B}+n_{A1}$: Since $C_{AB}=\max\{n_{1B},n_{2B}\}$ and $n_{A1}< C_{AB}$,  \eqref{r-d1} shows that $n_{A2}\ge C_{AB}$. Node $A$ transmits $C_{AB}$ streams as follows. It sends nothing on the highest $n_{A1}-(n_{1B}-n_{2B})$ levels, $a_{n_{1B}}$,...,$a_{n_{2B}+1}$ on the next levels, again nothing on the next $n_{A2}-(n_{A1}+n_{2B})$ levels and $a_{n_{2B}}$,...,$a_{1}$ on the next levels. $B$ can decode all $C_{AB}$ bits as illustrated in Figure \ref{fig:Examplee}. All the required streams reach $R_2$ since the number of levels is $n_{1B} + n_{A2}- n_{A1}-n_{2B}$ which is less than or equal to the number of the received stream levels, $n_{A2}$. Also ${[a_{n_{1B}},...,a_{n_{2B}+1}]}^T$ are the lowest levels at $R_1$ because there are $n_{A1}-(n_{1B}-n_{2B})$ zeros above them and they together are the highest $n_{A1}$ levels of the transmission from $A$.

\begin{figure}[htbp]
\centering
	\includegraphics[width=15cm]{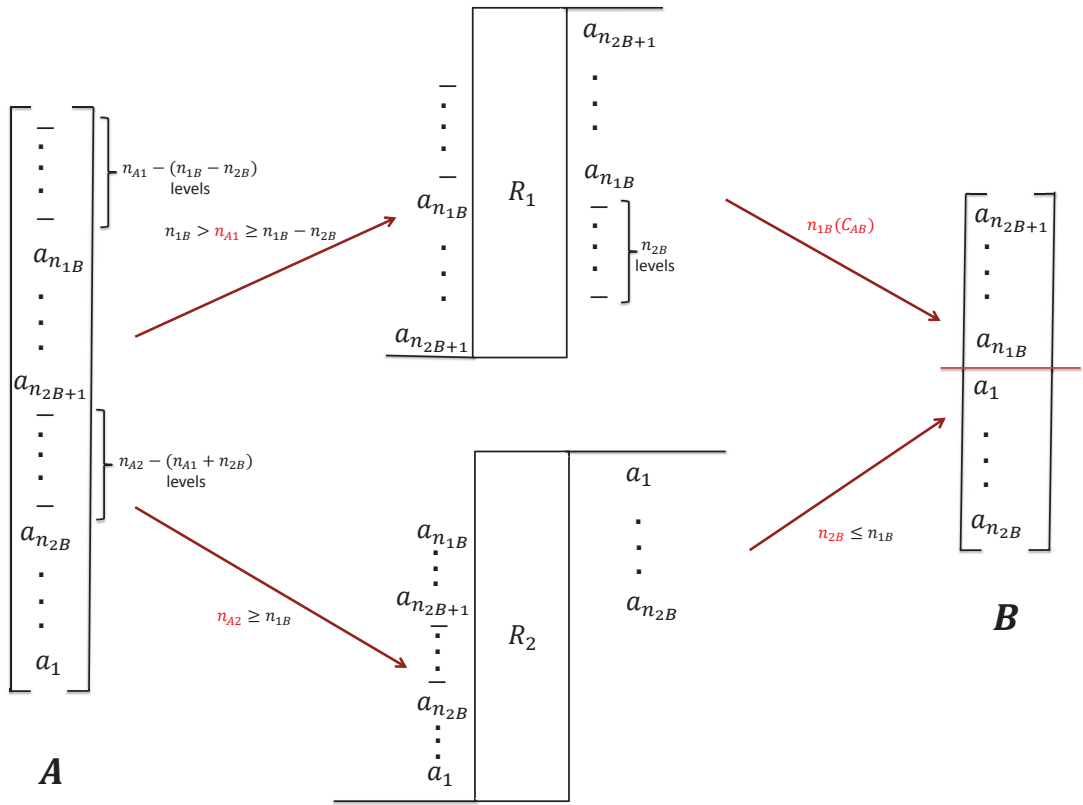}
\caption{Achievability scheme for Case 4.1.1 by using Relay Strategy 0.}
\label{fig:Examplee}
\end{figure}

Case 4.2: $n_{A1}\ge C_{AB}$: In this case, node $A$ transmits $C_{AB}$ bits on the lowest levels, and Figure \ref{fig:Exampleg} illustrates that node $B$ can decode the data. For decoding, while a stream, $a_v$, is the one being added from $R_1$ to $a_{\min\{n_{1B},n_{A1}-n_{A2}\}+1}$ received from $R_2$, if $v>\min\{n_{1B},n_{A1}-n_{A2}\}+1$ we decode the streams starting from the highest level $a_1$ and then subtract them from the signal before decoding the next lower stream, and if $v<\min\{n_{1B},n_{A1}-n_{A2}\}+1$ we decode the streams starting from the lowest level $a_{n_{1B}}$ and then subtract them from the signal before decoding the next upper stream.

\begin{figure}[htbp]
\centering
	\includegraphics[width=15cm]{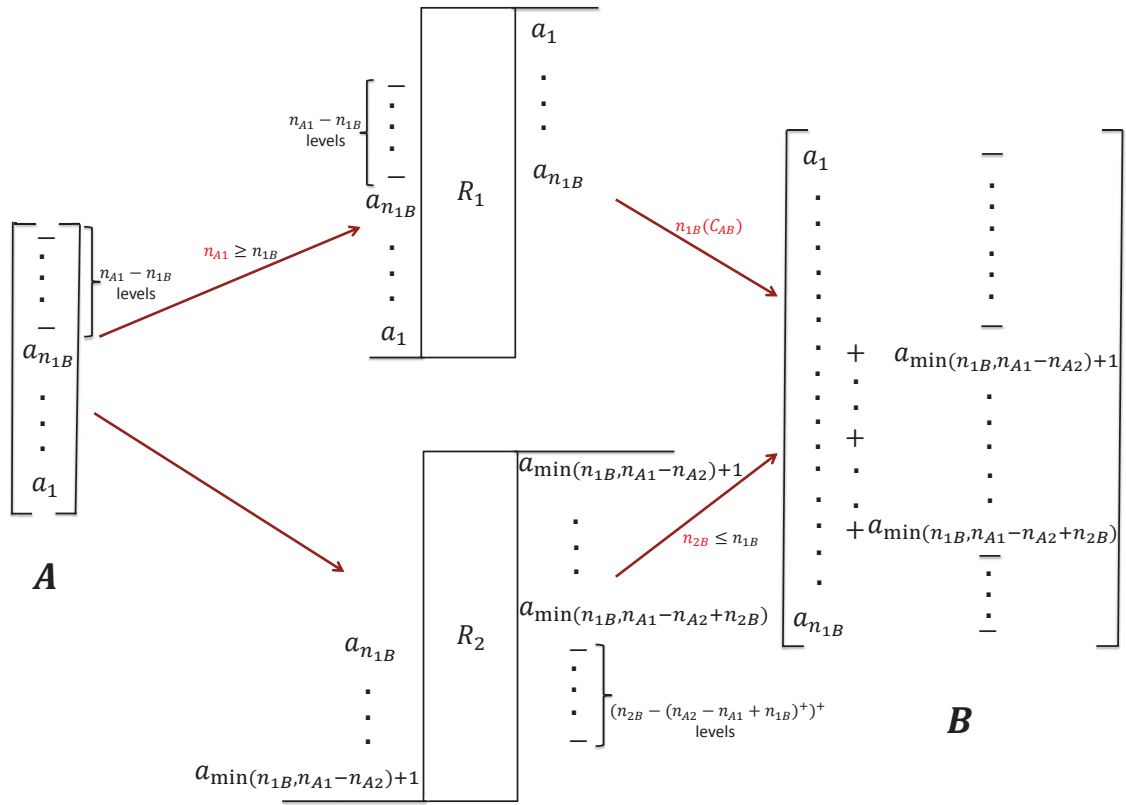}
\caption{Achievability scheme for Case 4.2 by using Relay Strategy 0.}
\label{fig:Exampleg}
\end{figure}

\section{Forward Channel is of Case 4.1.2} \label{apdx_sc2}





In this scenario, node $A$ uses the same transmission strategy as Case 4.1.1 Section \ref{neitherr} and Appendix \ref{apdx_sc1}, i.e., it transmits ${[\underset{n_{A1}-(n_{1B}-n_{2B})}{\underbrace{0,...,0}},a_{n_{1B}},...,a_{n_{2B}+1},\underset{n_{A2}-(n_{A1}+n_{2B})}{\underbrace{0,...,0}},a_{n_{2B}},...,a_{1}]}^T$ for the forward channel. Also, node $B$ uses the same strategy as in the corresponding case in Appendix \ref{apdx_sc1}. For the relay strategy, we will choose one of the four strategies below depending on the channel parameters and prove that all of these strategies are optimal for each set of parameters for the forward channel and then we will explain that at least one of these strategies is optimal for each set of parameters for the backward channel.

\subsection{\bf Relay Strategy 5:}

If the forward channel is of Case 4.1.2 Type $i$, then Relay Strategy 0 is used at $R_{\bar{i}}$, where $i,\bar{i}\in\{1,2\}, i\neq \bar{i}$, and Relay Strategy 5 is used at $R_{i}$. Here, we define Relay Strategy 5 at $R_1$ (forward channel of Case 4.1.2 Type $1$), while that for $R_2$ can be obtained by interchanging roles of $R_1$ and $R_2$ (interchanging 1 and 2 and forward channel of Case 4.1.2 Type $2$). As shown in Figure \ref{fig:rs5}, if $R_1$ receives a block of $n_{1B}$ bits, first it will reverse them as in Relay Strategy 0 and then change the order of the first $n_{2B}+n_{A1}-n_{A2}$ streams with the next $n_{1B}-n_{2B}$ streams.

\begin{figure}[htbp]
\centering
	\includegraphics[width=10cm]{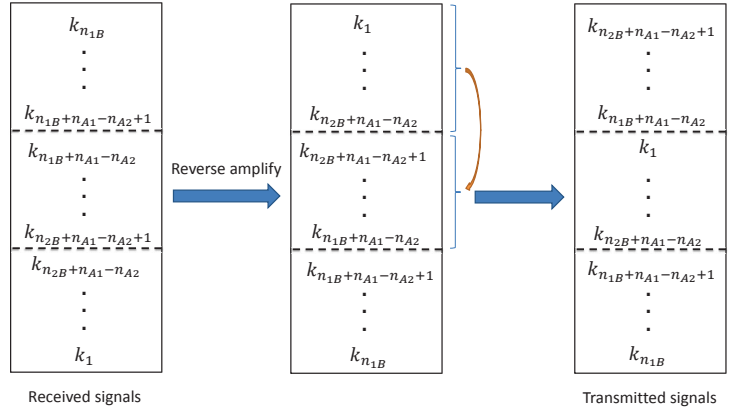}
\caption{Relay Strategy 5 at $R_1$.}
\label{fig:rs5}
\end{figure}

Node $A$ transmits ${[\underset{n_{A2}-n_{1B}}{\underbrace{0,...,0}},a_{n_{1B}},...,a_{1}]}^T$. The received signals can be seen in Figure \ref{fig:s5s}. The bits that were not delivered to node $B$ from $R_2$ ($a_{n_{2B}+1},...,a_{n_{1B}}$), are all sent in block $(R_1,B_1)$ to $B$ and thus are decoded with no interference. The remaining bits can be decoded by starting from the lowest level ($a_{n_{2B}}$ in block $(R_2,B_4)$) and removing the effect of the decoded bits.

\begin{figure}[htbp]
\centering
	\includegraphics[width=15cm]{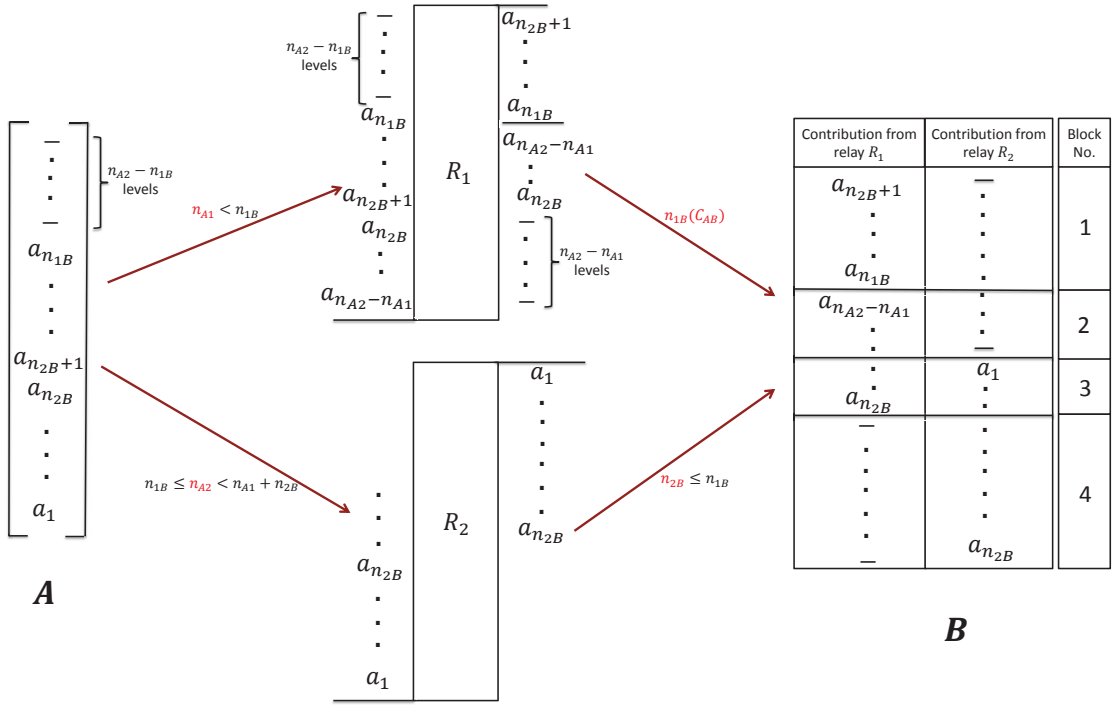}
\caption{Received signals by using Relay Strategy 5 when the forward channel is of Case 4.1.2 Type 1.}
\label{fig:s5s}
\end{figure}

\subsection{\bf Relay Strategy 6:}

If the forward channel is of Case 4.1.2 Type $i$, then Relay Strategy 0 is used at $R_{\bar{i}}$, where $i,\bar{i}\in\{1,2\}, i\neq \bar{i}$, and Relay Strategy 6 is used at $R_{i}$. Here, we define Relay Strategy 6 at $R_1$ (forward channel of Case 4.1.2 Type $1$), while that for $R_2$ can be obtained by interchanging roles of relays $R_1$ and $R_2$ (interchanging 1 and 2 and forward channel of Case 4.1.2 Type $2$). Relays work similar to Relay Strategy 0 with the only difference that $R_1$ repeats a part of the top $n_{1B}-n_{A2}+n_{A1}$ streams later too, as explained below in seven scenarios, based on the parameters of the forward channel.


\begin{figure}[htbp]
\centering
	\includegraphics[width=10cm]{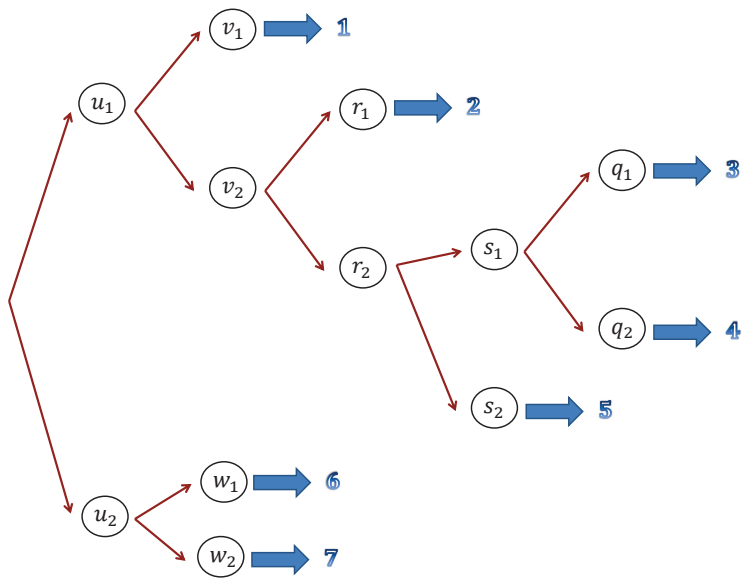}
\caption{Dividing the 4-dimensional space consisting of $(n_{A1},n_{A2},n_{1B},n_{2B})$ into seven subspaces.}
\label{fig:sup2}
\end{figure}

As shown in Figure \ref{fig:sup2}, we define the partition of the four-dimensional space $(n_{A1},n_{A2},n_{1B},n_{2B})$ into seven parts that lead to different received signal structures in node $B$, shown in Figures \ref{fig:11.1}-\ref{fig:22.2}, respectively. Specifically, $\{u_1,u_2\}=\{n_{A2}-n_{A1} \le 2(n_{A1}+n_{2B}-n_{A2}),n_{A2}-n_{A1} > 2(n_{A1}+n_{2B}-n_{A2})\}$, $\{v_1,v_2\}=\{n_{1B} > n_{A1}-n_{A2}+2n_{2B},n_{1B} \le n_{A1}-n_{A2}+2n_{2B}\}$, $\{w_1,w_2\}=\{n_{2B}+n_{A1}-n_{A2} \le n_{1B}-n_{2B},n_{2B}+n_{A1}-n_{A2} > n_{1B}-n_{2B}\}$, $\{r_1,r_2\}=\{n_{1B} \ge 2(n_{A1}+n_{2B}-n_{A2}),n_{1B} < 2(n_{A1}+n_{2B}-n_{A2})\}$, $\{s_1,s_2\}=\{n_{A2}-n_{A1}+n_{1B}-n_{2B} \le n_{A1}-n_{A2}+2n_{2B}-n_{1B},n_{A2}-n_{A1}+n_{1B}-n_{2B} > n_{A1}-n_{A2}+2n_{2B}-n_{1B}\}$ and $\{q_1,q_2\}=\{2(n_{2B}+n_{A1}-n_{A2})-n_{1B} \ge n_{A2}-n_{A1},2(n_{2B}+n_{A1}-n_{A2})-n_{1B} < n_{A2}-n_{A1}\}$.

\begin{enumerate}
  \item $(u_1,v_1)$: Figure \ref{fig:11.1} depicts the received signal at node $B$ (ignoring the effect of transmitted signal from $B$) assuming that both relays use Relay Strategy 0.
\begin{figure}[htbp]
\centering
	\includegraphics[width=8cm]{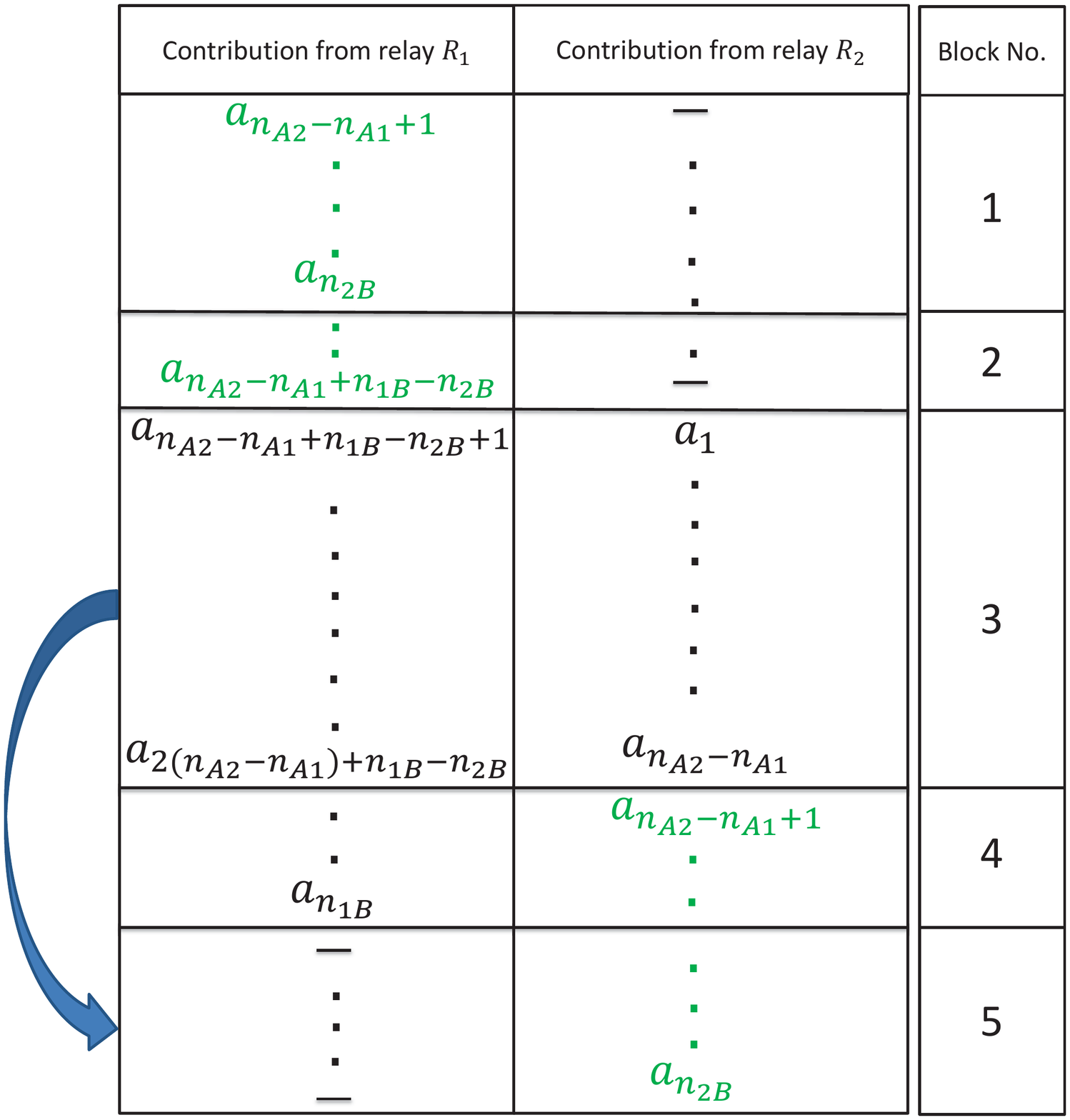}
\caption{The received signals at node $B$ (ignoring the effect of transmitted signal from $B$) assuming that both relays use Relay Strategy 0 for channel parameters of case $(u_1,v_1)$.}
\label{fig:11.1}
\end{figure}
      $R_1$ repeats block $(R_1, B_3)$ on block $(R_1, B_5)$. The decoding order is
      \begin{eqnarray}
&& {\text {decode \& subtract\ }} (R_1,B_1) \& (R_1,B_1) \rightarrow {\text {subtract\ }} (R_2,B_4) \& (R_2,B_5) \rightarrow \nonumber\\
&& {\text {decode \& subtract\ }} (R_1,B_5) \rightarrow  {\text {subtract\ }} (R_1,B_3) \rightarrow  {\text {decode \& subtract\ }} (R_2,B_3) \rightarrow \nonumber\\
&& {\text {decode \& subtract\ }} (R_1,B_4).\nonumber
\end{eqnarray}

  \item $(u_1,v_2,r_1)$: As shown in Figure \ref{fig:11.21}, 
\begin{figure}[htbp]
\centering
	\includegraphics[width=8cm]{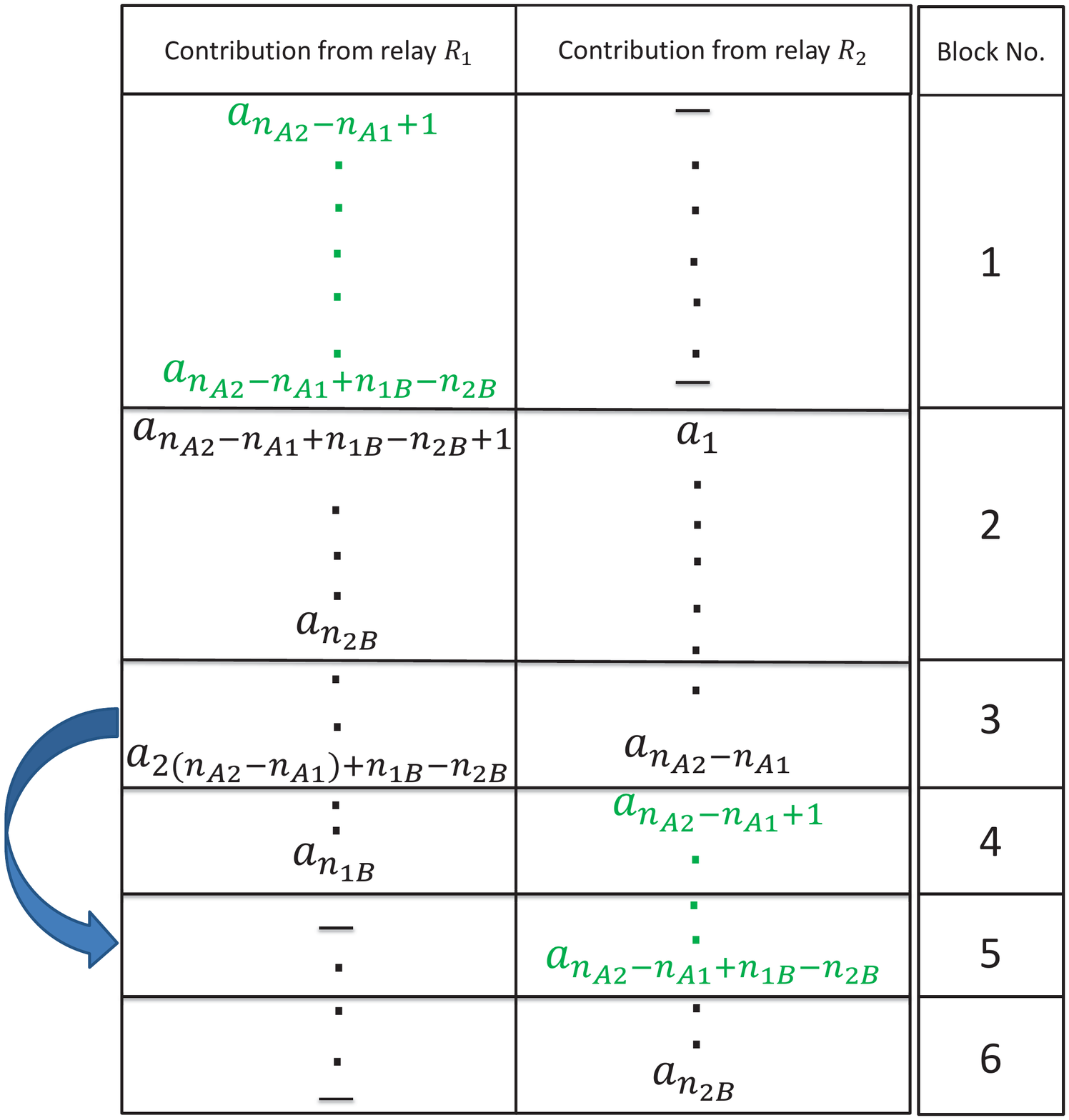}
\caption{The received signals at node $B$ (ignoring the effect of transmitted signal from $B$) assuming that both relays use Relay Strategy 0 for channel parameters of case $(u_1,v_2,r_1)$.}
\label{fig:11.21}
\end{figure}
      $R_1$ repeats block $(R_1, B_3)$ on block $(R_1, B_5)$. The decoding order is
      \begin{eqnarray}
&& {\text {decode \& subtract\ }} (R_1,B_1) \rightarrow {\text {subtract\ }} (R_2,B_4) \& (R_2,B_5) \rightarrow {\text {decode \& subtract\ }} (R_1,B_5) \rightarrow \nonumber\\
&& {\text {subtract\ }} (R_1,B_3) \rightarrow  {\text {decode \& subtract\ }} (R_2,B_6) \rightarrow
{\text {subtract\ }} (R_1,B_2) \rightarrow \nonumber\\
&& {\text {decode \& subtract\ }} (R_2,B_2)\& (R_2,B_3) \rightarrow {\text {decode \& subtract\ }} (R_1,B_4).\nonumber
\end{eqnarray}

  \item $(u_1,v_2,r_2,s_1,q_1)$: As shown in Figure \ref{fig:11.22.1.a}, 
\begin{figure}[htbp]
\centering
	\includegraphics[width=8cm]{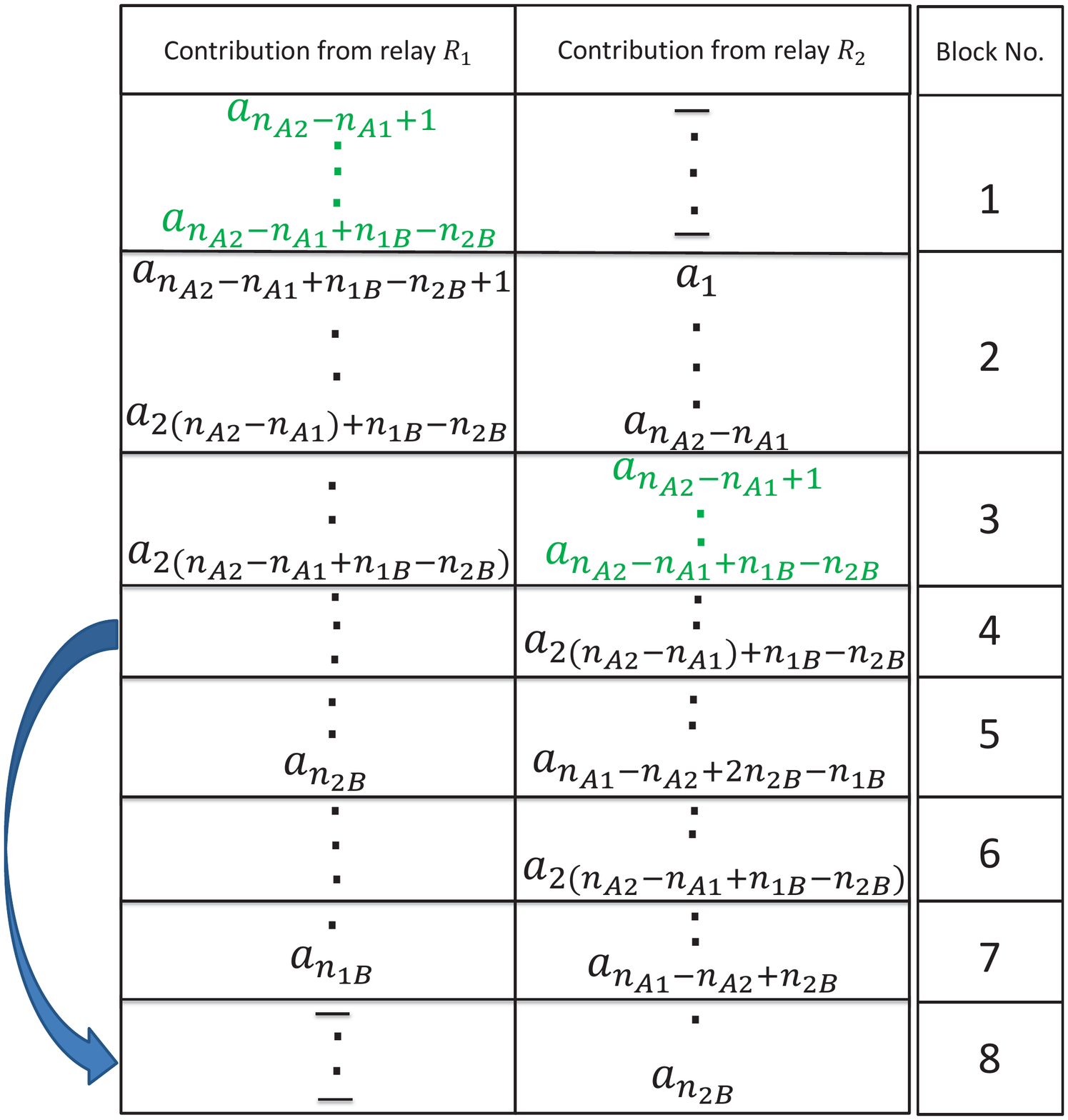}
\caption{The received signals at node $B$ (ignoring the effect of transmitted signal from $B$) assuming that both relays use Relay Strategy 0 for channel parameters of case $(u_1,v_2,r_2,s_1,q_1)$.}
\label{fig:11.22.1.a}
\end{figure}
      $R_1$ repeats block $(R_1, B_4)$ on block $(R_1, B_8)$. The decoding order is
      \begin{eqnarray}
&& {\text {decode \& subtract\ }} (R_1,B_1) \rightarrow {\text {subtract\ }} (R_2,B_3) \rightarrow {\text {decode \& subtract\ }} (R_1,B_3) \rightarrow \nonumber\\
&& {\text {subtract\ }} (R_2,B_5)\& (R_2,B_6) \rightarrow {\text {decode \& subtract\ }} (R_1,B_5) \& (R_1,B_6) \rightarrow \nonumber\\
&& {\text {decode \& subtract\ }} (R_1,B_8) \rightarrow
{\text {subtract\ }} (R_2,B_7)\& (R_2,B_8) \rightarrow {\text {decode \& subtract\ }} (R_1,B_7) \rightarrow \nonumber\\
&&{\text {subtract\ }} (R_1,B_4) \rightarrow  {\text {decode \& subtract\ }} (R_2,B_4) \rightarrow {\text {subtract\ }} (R_1,B_2) \rightarrow \nonumber\\
&& {\text {decode \& subtract\ }} (R_2,B_2).\nonumber
\end{eqnarray}

  \item $(u_1,v_2,r_2,s_1,q_2)$: As shown in Figure \ref{fig:11.22.1.b},
\begin{figure}[htbp]
\centering
	\includegraphics[width=8cm]{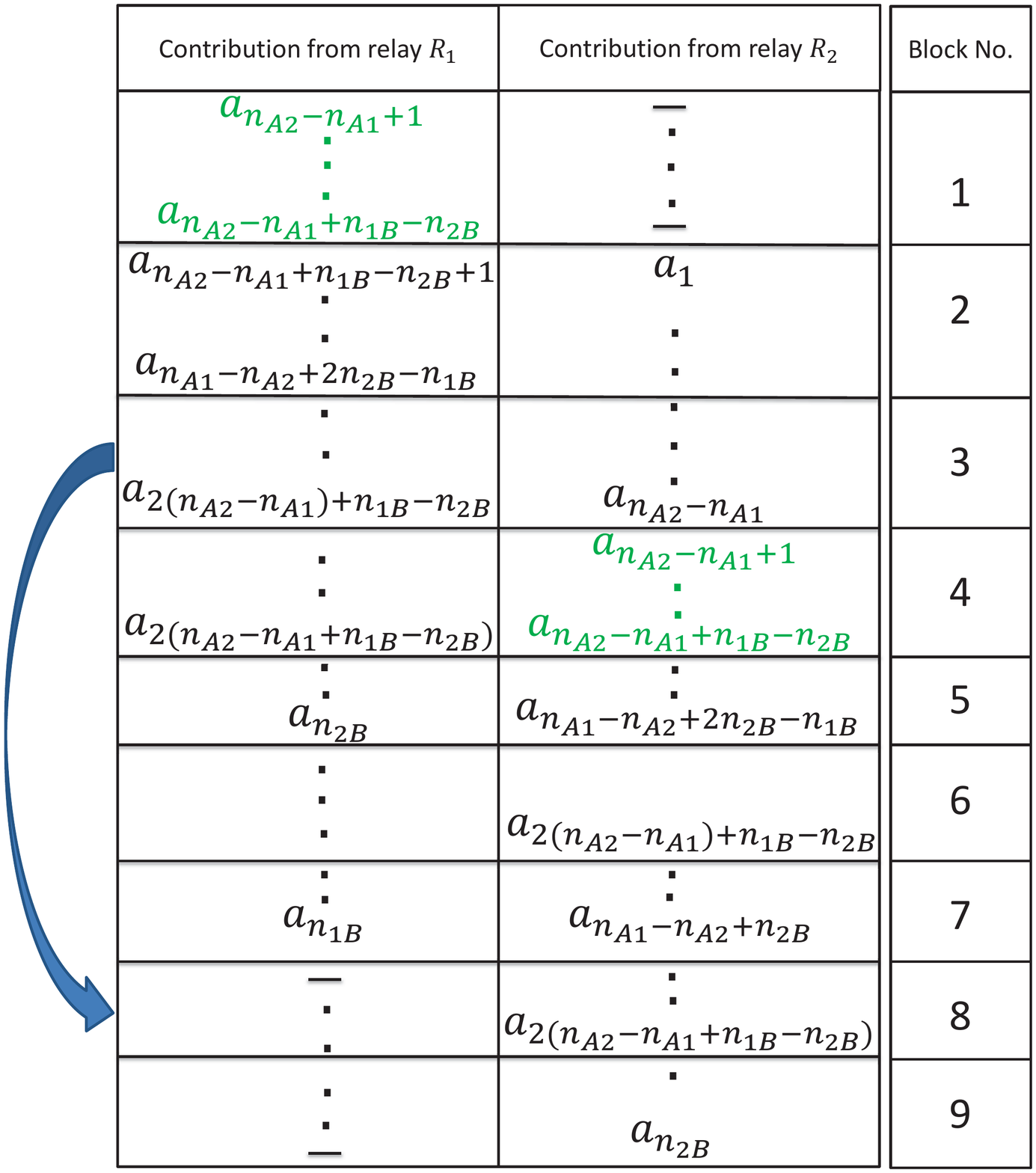}
\caption{The received signals at node $B$ (ignoring the effect of transmitted signal from $B$) assuming that both relays use Relay Strategy 0 for channel parameters of case $(u_1,v_2,r_2,s_1,q_2)$.}
\label{fig:11.22.1.b}
\end{figure}
      $R_1$ repeats block $(R_1, B_3)$ on block $(R_1, B_8)$. The decoding order is
      \begin{eqnarray}
&& {\text {decode \& subtract\ }} (R_1,B_1) \rightarrow {\text {subtract\ }} (R_2,B_4) \rightarrow {\text {decode \& subtract\ }} (R_1,B_4) \rightarrow \nonumber\\
&& {\text {subtract\ }} (R_2,B_7)\& (R_2,B_8) \rightarrow {\text {decode \& subtract\ }} (R_1,B_7) \& (R_1,B_8) \rightarrow \nonumber\\
&& {\text {subtract\ }} (R_1,B_3)\& (R_2,B_6)\rightarrow {\text {decode \& subtract\ }} (R_2,B_9) \rightarrow {\text {subtract\ }} (R_1,B_5) \rightarrow \nonumber\\
&&{\text {decode \& subtract\ }} (R_2,B_5)\rightarrow {\text {subtract\ }} (R_1,B_2) \rightarrow {\text {decode \& subtract\ }} (R_2,B_2) \& (R_2,B_3) \rightarrow \nonumber\\
&& {\text {decode \& subtract\ }} (R_1,B_6).\nonumber
\end{eqnarray}

  \item $(u_1,v_2,r_2,s_2)$: As shown in Figure \ref{fig:11.22.2},
\begin{figure}[htbp]
\centering
	\includegraphics[width=8cm]{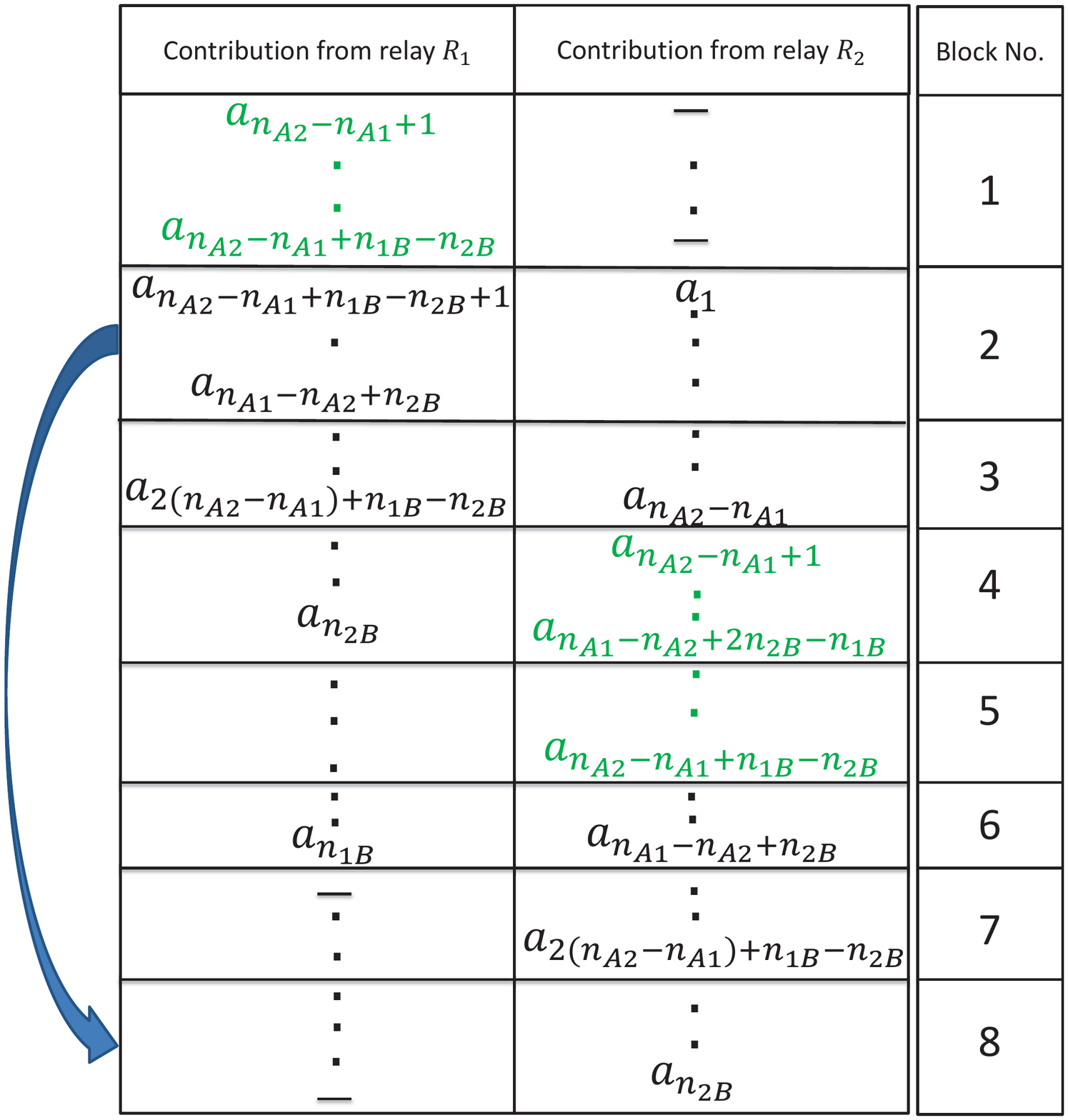}
\caption{The received signals at node $B$ (ignoring the effect of transmitted signal from $B$) assuming that both relays use Relay Strategy 0 for channel parameters of case $(u_1,v_2,r_2,s_2)$.}
\label{fig:11.22.2}
\end{figure}
      $R_1$ repeats block $(R_1, B_2)$ on block $(R_1, B_8)$. The decoding order is
      \begin{eqnarray}
&& {\text {decode \& subtract\ }} (R_1,B_1) \rightarrow {\text {subtract\ }} (R_2,B_4) \& (R_2,B_5) \rightarrow \nonumber\\
&&{\text {decode \& subtract\ }} (R_1,B_4)\&(R_1,B_5) \rightarrow  {\text {subtract\ }} (R_2,B_8) \rightarrow {\text {decode \& subtract\ }} (R_1,B_8) \rightarrow \nonumber\\
&& {\text {subtract\ }} (R_1,B_2)\& (R_2,B_6)\rightarrow {\text {decode \& subtract\ }} (R_2,B_7) \rightarrow {\text {subtract\ }} (R_1,B_3) \rightarrow \nonumber\\
&&{\text {decode \& subtract\ }} (R_2,B_2) \& (R_2,B_3) \& (R_1,B_6).\nonumber
\end{eqnarray}

  \item $(u_2,w_1)$: As shown in Figure \ref{fig:22.1},
\begin{figure}[htbp]
\centering
	\includegraphics[width=8cm]{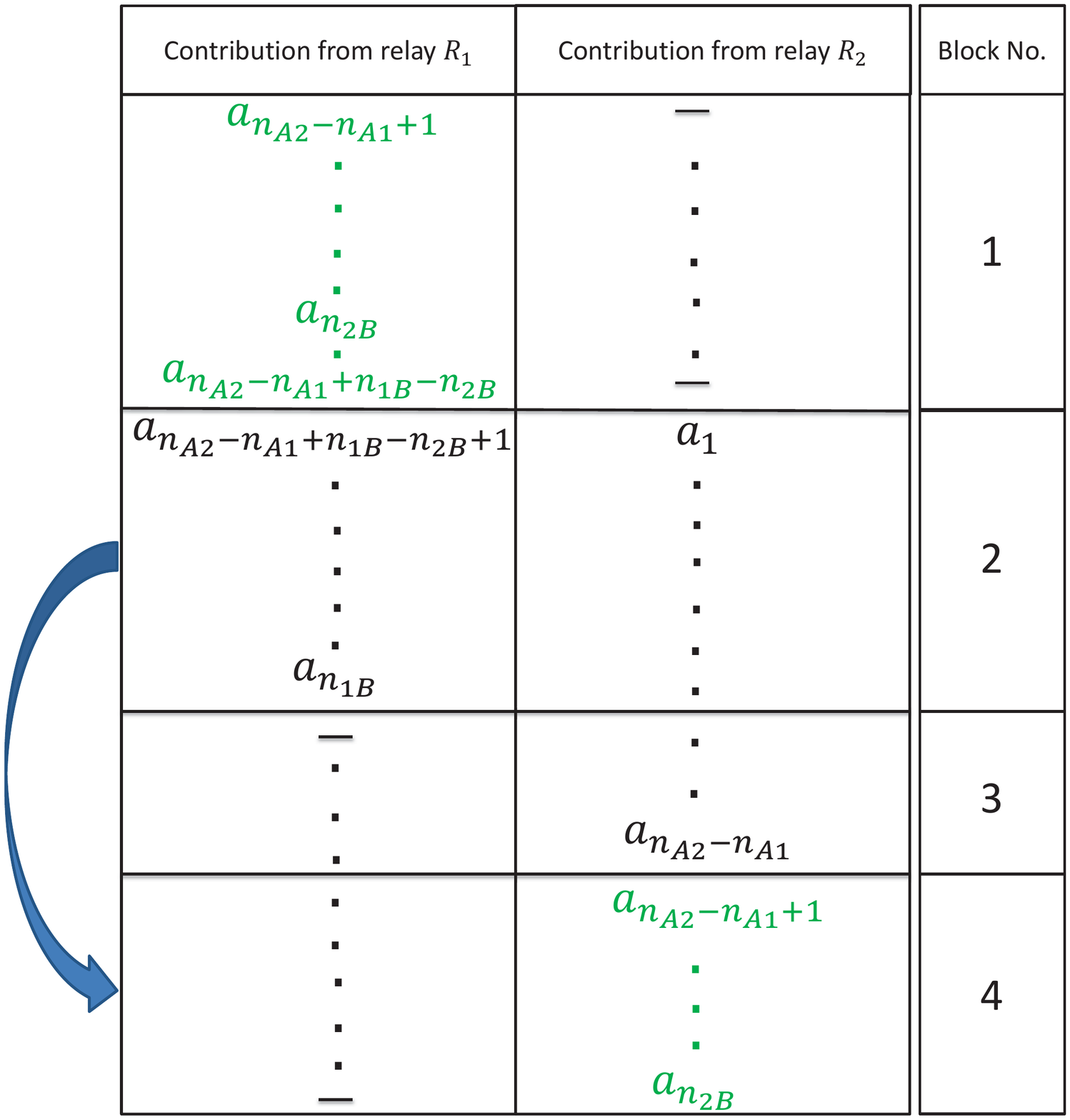}
\caption{The received signals at node $B$ (ignoring the effect of transmitted signal from $B$) assuming that both relays use Relay Strategy 0 for channel parameters of case $(u_2,w_1)$.}
\label{fig:22.1}
\end{figure}
      $R_1$ repeats block $(R_1, B_2)$ on block $(R_1, B_4)$. The decoding order is
      \begin{eqnarray}
&& {\text {decode \& subtract\ }} (R_1,B_1) \rightarrow {\text {subtract\ }} (R_2,B_4) \rightarrow {\text {decode \& subtract\ }} (R_1,B_4) \rightarrow \nonumber\\
&& {\text {subtract\ }} (R_1,B_2) \rightarrow {\text {decode \& subtract\ }} (R_2,B_2) \rightarrow {\text {decode \& subtract\ }} (R_2,B_3).\nonumber
\end{eqnarray}

  \item $(u_2,w_2)$: As shown in Figure \ref{fig:22.2},
\begin{figure}[htbp]
\centering
	\includegraphics[width=8cm]{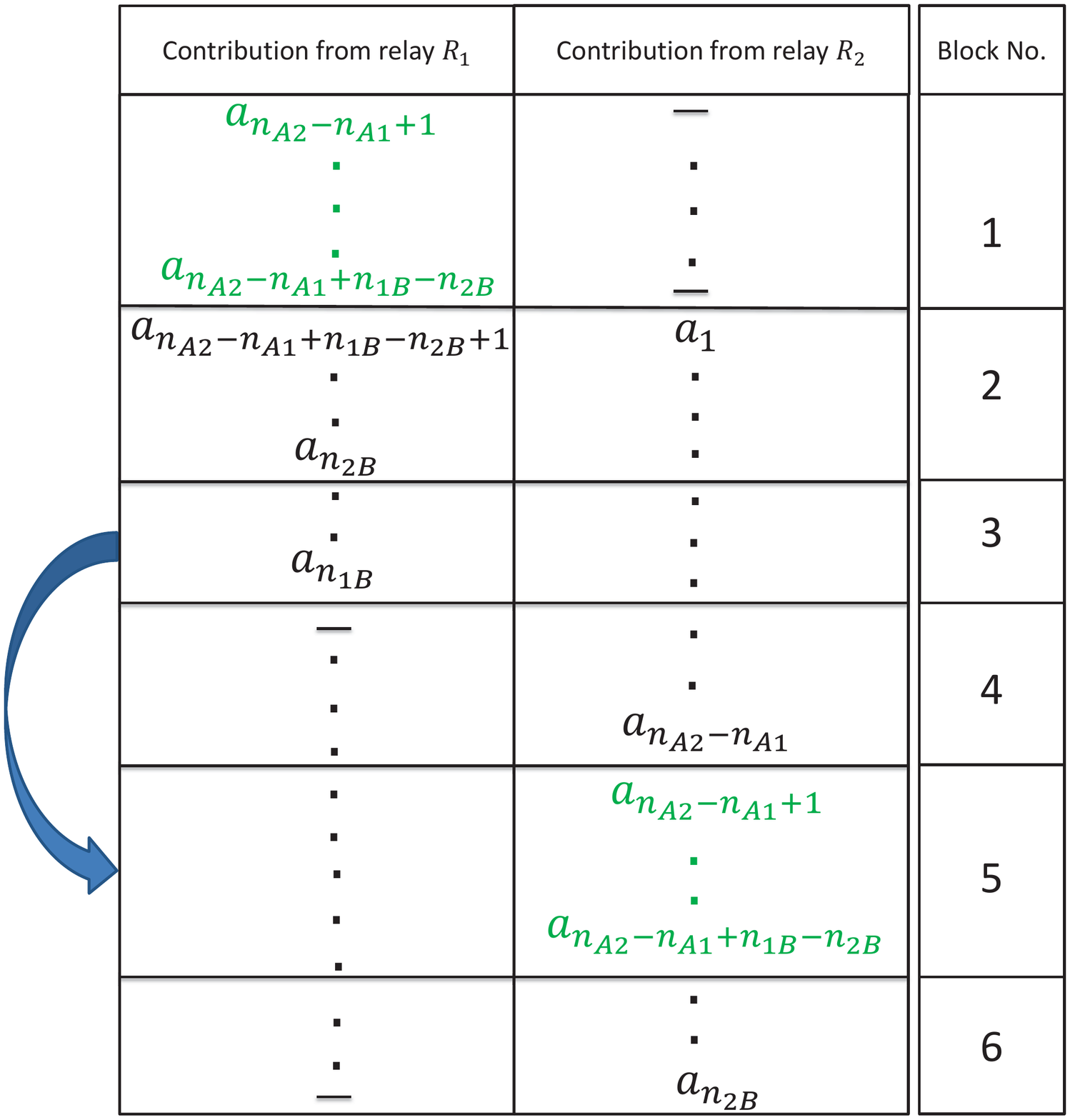}
\caption{The received signals at node $B$ (ignoring the effect of transmitted signal from $B$) assuming that both relays use Relay Strategy 0 for channel parameters of case $(u_2,w_2)$.}
\label{fig:22.2}
\end{figure}
      $R_1$ repeats block $(R_1, B_3)$ on block $(R_1, B_5)$. The decoding order is
      \begin{eqnarray}
&& {\text {decode \& subtract\ }} (R_1,B_1) \rightarrow {\text {subtract\ }} (R_2,B_5) \rightarrow {\text {decode \& subtract\ }} (R_1,B_5) \rightarrow \nonumber\\
&& {\text {subtract\ }} (R_1,B_3) \rightarrow {\text {decode \& subtract\ }} (R_2,B_6) \rightarrow
{\text {subtract\ }} (R_1,B_2) \rightarrow \nonumber\\
&& {\text {decode \& subtract\ }} (R_2,B_2) \& (R_2,B_3) \& (R_2,B_4).\nonumber
\end{eqnarray}

\end{enumerate}

\subsection{\bf Relay Strategy 7:}

If the forward channel is of Case 4.1.2 Type $i$, then Relay Strategy 0 is used at $R_{\bar{i}}$, where $i,\bar{i}\in\{1,2\}, i\neq \bar{i}$, and Relay Strategy 7 is used at $R_{i}$. Here, we define Relay Strategy 7 at $R_1$ (forward channel of Case 4.1.2 Type $1$), while that for Relay $R_2$ can be obtained by interchanging roles of $R_1$ and the $R_1$ (interchanging 1 and 2 and forward channel of Case 4.1.2 Type $2$). As shown in Figure \ref{fig:rs7}, if $R_1$ receives a block of $n_{1B}$, first it will reverse them as in Relay Strategy 0 and then changes the order of the $n_{2B}-n_{A2}+n_{A1}$ streams right after the first $n_{1B}-n_{2B}$ streams, with the following $n_{A2}-n_{A1}$ streams.

\begin{figure}[htbp]
\centering
	\includegraphics[width=10cm]{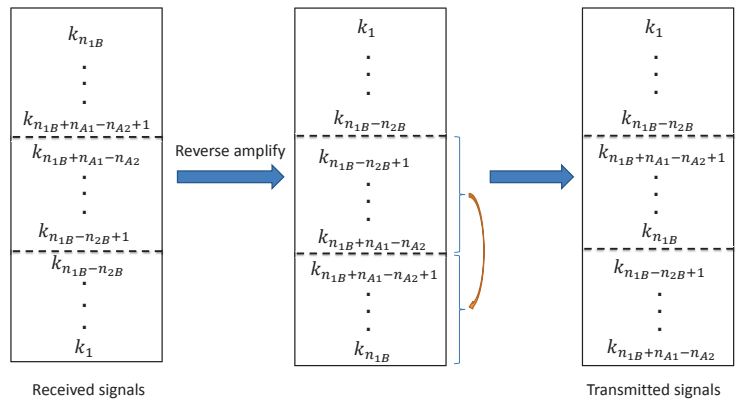}
\caption{Relay Strategy 7 at $R_1$.}
\label{fig:rs7}
\end{figure}

Node $A$ transmits ${[\underset{n_{A2}-n_{1B}}{\underbrace{0,...,0}},a_{n_{1B}},...,a_{1}]}^T$. The received signals can be seen in Figure \ref{fig:s7s}. The bits that are not delivered to node $B$ from $R_1$ ($a_{1},...,a_{n_{A2}-n_{A1}}$), are delivered from $R_2$ to node $B$ (block $(R_2,B_2)$) with no interference. The remaining bits can be decoded by starting from the highest level ($a_{n_{A2}-n_{A1}+1}$ in block $(R_1,B_1)$) and removing the effect of the decoded bits.

\begin{figure}[htbp]
\centering
	\includegraphics[width=15cm]{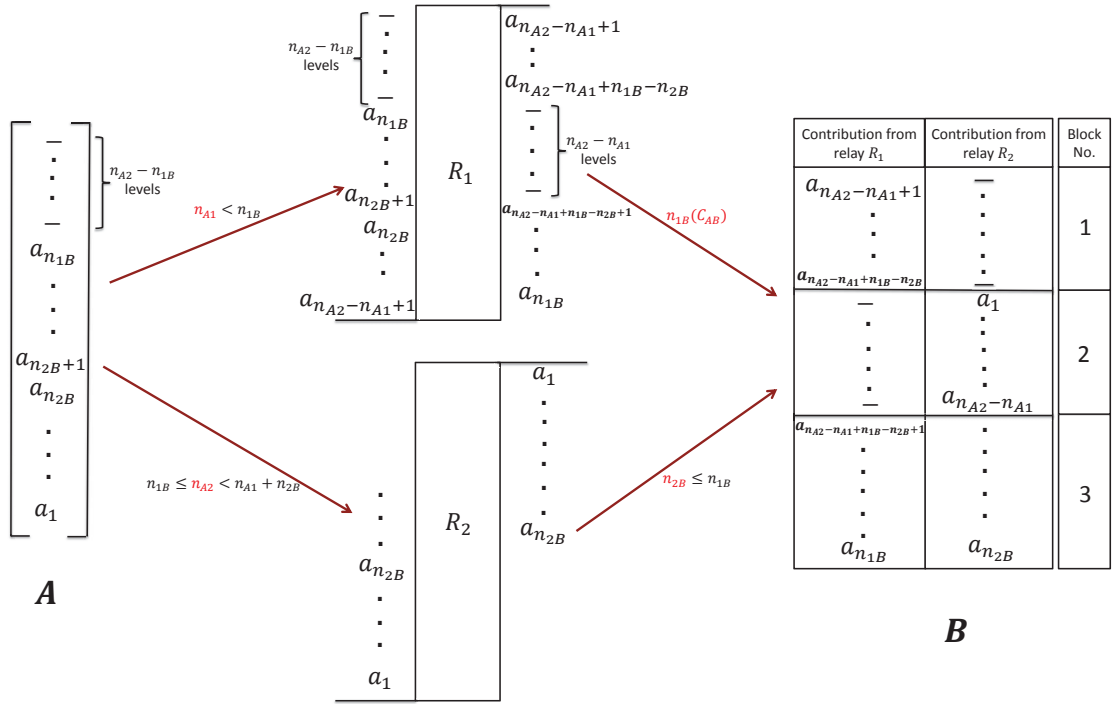}
\caption{Received signals by using Relay Strategy 7 when the forward channel is of Case 4.1.2 Type 1.}
\label{fig:s7s}
\end{figure}

\subsection{\bf Relay Strategy 8:}

If the forward channel is of Case 4.1.2 Type $i$, then Relay Strategy 0 is used at $R_i$, and Relay Strategy 8 is used at $R_{\bar{i}}$, where $i,\bar{i}\in\{1,2\}, i\neq \bar{i}$. Here, we define Relay Strategy 8 at $R_2$ (forward channel of Case 4.1.2 Type $1$), while that for $R_1$ can be obtained by interchanging roles of relays $R_1$ and $R_2$ (interchanging 1 and 2 and forward channel of Case 4.1.2 Type $2$). As shown in Figure \ref{fig:rs8}, if $R_2$ receives a block of $n_{2B}$ streams, first of all it will reverse them as in Relay Strategy 0 and then changes the order of the first $n_{A2}-n_{A1}$ streams with the next $n_{2B}-(n_{A2}-n_{A1})$ streams.

\begin{figure}[htbp]
\centering
	\includegraphics[width=10cm]{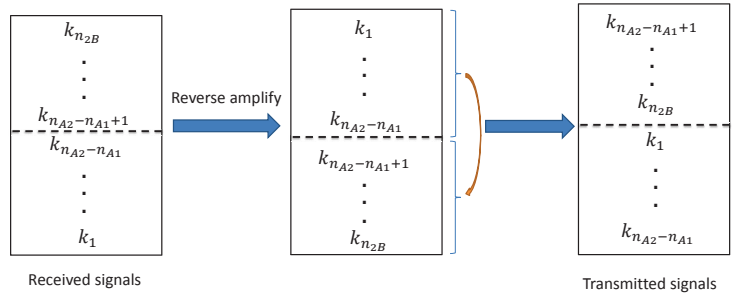}
\caption{Relay Strategy 8 at $R_2$.}
\label{fig:rs8}
\end{figure}

Node $A$ transmits ${[\underset{n_{A2}-n_{1B}}{\underbrace{0,...,0}},a_{n_{1B}},...,a_{1}]}^T$. The received signals can be seen in Figure \ref{fig:s8s}. The bits that are not delivered to node $B$ from $R_1$ ($a_{1},...,a_{n_{A2}-n_{A1}}$), are all sent on the lowest levels from $R_2$ to node $B$ (in block $(R_1,B_1)$) and thus are decoded with no interference. The remaining bits can be decoded by starting from the highest level ($a_{n_{A2}-n_{A1}+1}$) and removing the effect of the decoded bits.

\begin{figure}[htbp]
\centering
	\includegraphics[width=15cm]{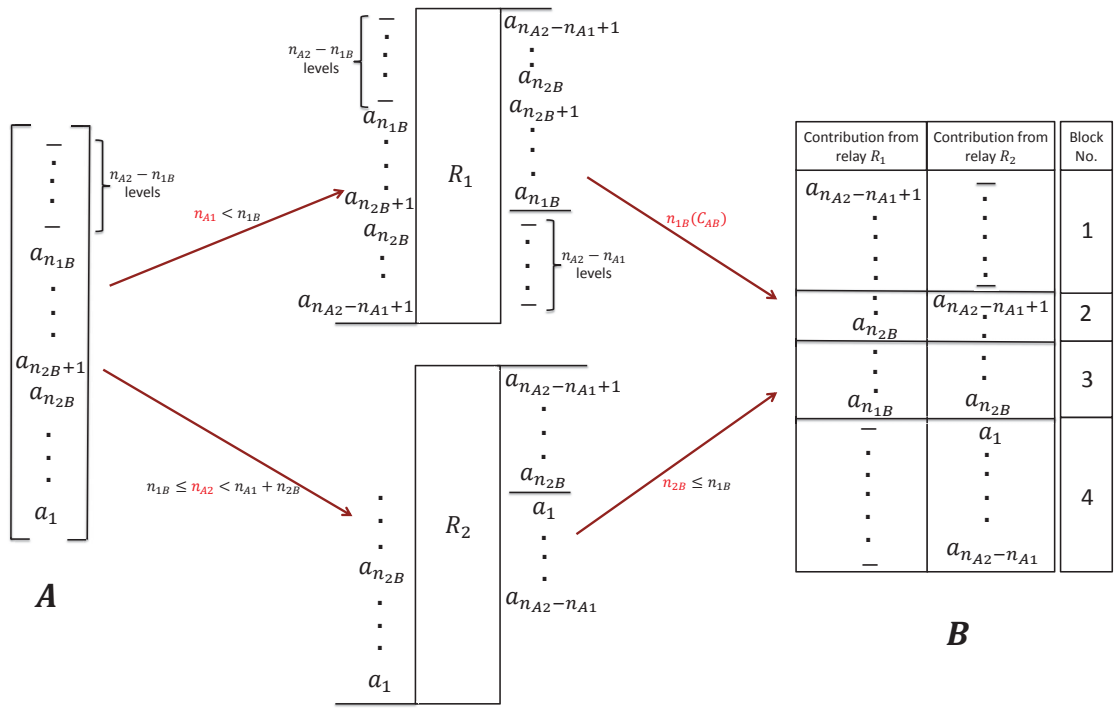}
\caption{Received signals by using Relay Strategy 8 when the forward channel is of Case 4.1.2 Type 1.}
\label{fig:s8s}
\end{figure}

\subsection{Achieving the Optimum Rate}

Now we explain how the above mentioned strategies achieve the optimum rate. Take the case that the forward channel is of Case 4.1.2 and backward channel is neither of Case 3.1.2 nor of Case 4.1.2. If the forward channel is neither of Case 3.1.2 nor of Case 4.1.2 and the backward channel is of Case 4.1.2 everything is similar except exchanging all $A$ and $B$'s together. The first one (which we consider here) includes the following situations:

1. Backward channel is of Case 1:
\begin{itemize}
  \item Forward channel is of Case 4.1.2 Type 1: If $n_{1B}+n_{A1}-n_{A2}>n_{1A}$, we use Relay Strategy 6 at $R_1$ and Relay Strategy 0 at $R_2$. Otherwise use Relay Strategy 5 at $R_1$ and Relay Strategy 0 at $R_2$. If $n_{1B}+n_{A1}-n_{A2}>n_{1A}$, $R_1$ repeats from the streams ($b_{1},...,b_{n_{1A}}$) received below the noise level in $A$, and otherwise some of the equations ($b_{1},...,b_{n_{1A}}$) are relocated.

  \item Forward channel is of Case 4.1.2 Type 2: If $n_{2B}+n_{A2}-n_{A1}>n_{B2}$, we use Relay Strategy 6 at $R_2$ and Relay Strategy 0 at $R_1$. Otherwise use Relay Strategy 5 at $R_2$ and Relay Strategy 0 at $R_1$. If $n_{2B}+n_{A2}-n_{A1}>n_{B2}$, $R_2$ repeats from the streams that are already decoded from the highest levels received in $A$, ($b_{n_{1A}+1},...,b_{n_{1A}+n_{B2}}$), on the lower levels, and otherwise some of the equations at the highest levels received in $A$, ($b_{n_{1A}+1},...,b_{n_{1A}+n_{B2}}$) are relocated.
\end{itemize}

2. Backward channel is of Case 2:
\begin{itemize}
  \item Forward channel is of Case 4.1.2 Type 1: If $n_{1B}+n_{A1}-n_{A2}>n_{B1}$, we use Relay Strategy 6 at $R_1$ and Relay Strategy 0 at $R_2$. Otherwise use Relay Strategy 5 at $R_1$ and Relay Strategy 0 at $R_2$. If $n_{1B}+n_{A1}-n_{A2}>n_{B1}$, $R_1$ repeats from the streams that are already decoded from the highest levels received in $A$, ($b_{n_{2A}+1},...,b_{n_{2A}+n_{B1}}$), on the lower levels, and otherwise we only exchange the place of some of the equations at the highest levels received in $A$, ($b_{n_{2A}+1},...,b_{n_{2A}+n_{B1}}$).

  \item Forward channel is of Case 4.1.2 Type 2: If $n_{2B}+n_{A2}-n_{A1}>n_{2A}$, we use Relay Strategy 6 at $R_2$ and Relay Strategy 0 at $R_1$. Otherwise use Relay Strategy 5 at $R_2$ and Relay Strategy 0 at $R_1$. If $n_{2B}+n_{A2}-n_{A1}>n_{2A}$, $R_2$ repeats from the streams ($b_{1},...,b_{n_{2A}}$) received below the noise levels in $A$, and otherwise some of the equations ($b_{1},...,b_{n_{2A}}$) are relocated.
\end{itemize}

3. Backward channel is of Case 3.1.1: We assume that the backward channel is Type 1. For Type 2 the proof is similar.
\begin{itemize}
  \item Forward channel is of Case 4.1.2 Type 1: If $n_{1B}+n_{A1}-n_{A2}>n_{B1}$, we use Relay Strategy 6 at $R_1$ and Relay Strategy 0 at $R_2$. Otherwise use Relay Strategy 5 at $R_1$ and Relay Strategy 0 at $R_2$. If $n_{1B}+n_{A1}-n_{A2}>n_{B1}$, $R_1$ repeats from the streams ($b_{1},...,b_{n_{1A}}$) received below the noise level in $A$, and otherwise some of the equations ($b_{1},...,b_{n_{1A}}$) are relocated.

  \item Forward channel is of Case 4.1.2 Type 2: If $n_{2B}+n_{A2}-n_{A1}>n_{B2}$, we use Relay Strategy 6 at $R_2$ and Relay Strategy 0 at $R_1$. Otherwise use Relay Strategy 5 at $R_2$ and Relay Strategy 0 at $R_1$. If $n_{2B}+n_{A2}-n_{A1}>n_{B2}$, $R_2$ repeats from the streams that are already decoded from the highest levels received in $A$, ($b_{n_{B1}-n_{B2}+1},...,b_{n_{B1}}$), on the lower levels, and otherwise some of the equations at the highest levels received in $A$, ($b_{n_{B1}-n_{B2}+1},...,b_{n_{B1}}$) are relocated.
\end{itemize}

4. Backward channel is Case 4.1.1: We assume that the backward channel is Type 1. For Type 2 the proof is similar.
\begin{itemize}
  \item Forward channel is of Case 4.1.2 Type 1: If $n_{1B}+n_{A1}-n_{A2}>n_{1B}-n_{2B}$, we use Relay Strategy 6 at $R_1$ and Relay Strategy 0 at $R_2$. Otherwise use Relay Strategy 5 at $R_1$ and Relay Strategy 0 at $R_2$. If $n_{1B}+n_{A1}-n_{A2}>n_{1B}-n_{2B}$, $R_1$ repeats from the streams that are already decoded from the highest levels received in $A$, ($b_{n_{2A}+1},...,b_{n_{1A}}$), on the lower levels, and otherwise some of the equations at the highest levels received in $A$, ($b_{n_{2A}+1},...,b_{n_{1A}}$) are relocated.

  \item Forward channel is of Case 4.1.2 Type 2: If $n_{1B}+n_{A1}-n_{A2}>n_{1B}$, we use Relay Strategy 6 at $R_2$ and Relay Strategy 0 at $R_1$. Otherwise use Relay Strategy 5 at $R_2$ and Relay Strategy 0 at $R_1$. If $n_{1B}+n_{A1}-n_{A2}>n_{1B}$, $R_2$ repeats from the streams ($b_{1},...,b_{n_{2A}}$) received below the noise level in $A$, and otherwise some of the equations ($b_{1},...,b_{n_{2A}}$) are relocated.
\end{itemize}

5. Backward channel is of Case 3.2: We assume that the backward channel is Type 1. For Type 2 the proof is similar.
\begin{itemize}
  \item Forward channel is of Case 4.1.2 Type 1: We use Relay Strategy 5 at $R_1$ and Relay Strategy 0 at $R_2$ or Relay Strategy 6 at $R_1$ and Relay Strategy 0 at $R_2$ or Relay Strategy 7 at $R_1$ and Relay Strategy 0 at $R_2$.

  \item Forward channel is of Case 4.1.2 Type 2: We use Relay Strategy 8 at $R_1$ and Relay Strategy 0 at $R_2$.
\end{itemize}

6. Backward channel is of Case 4.2: We assume that the backward channel is Type 1. For Type 2 the proof is similar.
\begin{itemize}
  \item Forward channel is of Case 4.1.2 Type 1: We use Relay Strategy 5 at $R_1$ and Relay Strategy 0 at $R_2$ or Relay Strategy 6 at $R_1$ and Relay Strategy 0 at $R_2$ or Relay Strategy 7 at $R_1$ and Relay Strategy 0 at $R_2$.

  \item Forward channel is of Case 4.1.2 Type 2: We use Relay Strategy 8 at $R_1$ and Relay Strategy 0 at $R_2$.
\end{itemize}

\section{Both  the forward and backward channels are either of Case 3.1.2 or 4.1.2}\label{apdx_sc3}

\subsection{Forward channel is of Case 3.1.2}

Having described the relay and transmission strategies which are symmetric in both the forward and backward channels, we will show that the message can be decoded in the forward direction. The other side holds by symmetry. We first assume that the forward channel is of Case 3.1.2 Type 1. For Type 2, the proof is similar and is thus omitted. Then we have $C_{AB}=\max\{n_{A1},n_{A2}\}=n_{A1}$. We consider the partitioning shown in Figure \ref{fig:sup1} and show that the messages can be decoded when the  backward channel is either of Cases 3.1.2 or 4.1.2. It can be seen that all of the streams can be decoded with the same order as in Relay Strategy 2 in Section \ref{one}.

1. $(u_1,v_1)$: Figure \ref{fig:1} depicts the received signal at node $B$ (ignoring the effect of transmitted signal from $B$) assuming that both relays use Relay Strategy 0.

First, consider the case when the repetitions in favor of the forward and backward channels happen in different relays, i.e., $m_1= n_2=0$. In this case, $R_2$ repeats in favor of the forward channel (uses Relay Strategy $(2,0)$) and $R_1$ repeats in favor of the backward channel (uses Relay Strategy $(0,2)$ or $(0,6)$). Repeating $a_{n_{A1}-n_{A2}+1},...,a_{n_{1B}}$ by $R_1$  in favor of the backward channel does not affect the achievability of the forward channel because these signals are decoded from block $(R_2,B_1)$. Also, repeating the $a_{1},...,a_{n_{A1}-n_{A2}}$ by $R_1$ within the top $n_{A1}-n_{A2}$ streams (i.e., on block $(R_1,B_2)$) does not affect the decoding since the signals can be decoded from top and their effect can be cancelled. Repeating any of these signals on the next $2 n_{A2}-n_{A1}+n_{1B}-n_{2B}$ lower levels by $R_1$ (i.e., on block $(R_1,B_3)$) does not affect the decoding because we can decode those upper $n_{A1}-n_{A2}$ levels (i.e., on block $(R_1,B_2)$) first and cancel the effect of the repeated signals. If $a_{1},...,a_{n_{A1}-n_{A2}}$ is repeated by $R_1$ on the next lower $n_{2B}-n_{A2}$ levels (i.e., on block $(R_1,B_4)$), in case that the repetitions of some streams from two relays are from the same level and are repeated on the same level, $R_2$ does not repeat for those levels as was explained in the definition of the strategies, so there is no problem in decoding the forward channel. If any of them is repeated in any lower level by $R_1$ (i.e. below block 4 in Figure \ref{fig:1}), these are out of range and does not have effect on decoding the forward channel.

Now, take the case that the repetitions in favor of the forward and backward channels both happen in $R_2$, i.e., $(m_1,n_1)=(0,0)$. In this case, $R_2$ repeats in favor of both directions ($(m_2,n_2)\in \{(2,2), (2,6)\}$) and $R_1$ uses Relay Strategy $(0,0)$, i.e., Relay Strategy 0. Repeating $a_{n_{A1}-n_{A2}+1},...,a_{n_{A1}-n_{A2}+n_{2B}-n_{1B}}$ in favor of the backward channel does not affect the achievability of the forward channel because these signals are decoded from the top levels of the received signal from the same relay (i.e., on block $(R_2,B_1)$). Also, repeating the signals $a_{n_{A1}-n_{A2}+n_{2B}-n_{1B}+1},...,a_{2n_{A1}-2n_{A2}+n_{2B}-n_{1B}}$ within the next top $n_{A1}-n_{A2}$ streams (i.e., on block $(R_2,B_2)$) does not affect the decoding since the signals can be decoded from top and their effect can be cancelled. Repeating any of them on the next $2 n_{A2}-n_{A1}+n_{1B}-n_{2B}$ lower levels (i.e., on block $(R_2,B_3)$) does not affect the decoding because node $B$ can decode the upper $n_{A1}-n_{A2}$ levels (i.e., on block $(R_2,B_2)$) first and cancel the effect of the repeated signals. If $a_{n_{A1}-n_{A2}+n_{2B}-n_{1B}+1},...,a_{2n_{A1}-2n_{A2}+n_{2B}-n_{1B}}$ is repeated on the next lower $n_{2B}-n_{A2}$ levels (i.e., on block $(R_2,B_4)$), it does not affect the decoding because as it was explained in the definition of the strategies, in case that the repetitions of some streams from two relays are from the same levels and are supposed to be repeated on the same level, $R_2$ does not repeat for those levels and repeats those streams only one time. If any of $a_{2n_{A1}-2n_{A2}+n_{2B}-n_{1B}+1},...,a_{n_{A1}}$ is repeated on the next lower $n_{2B}-n_{A2}$ levels (i.e., on block $(R_2,B_4)$), it does not affect the decoding since the signals can be decoded from top and their effect can be cancelled. If any of them is repeated in any lower level by $R_2$ (i.e. below block 4 in Figure \ref{fig:1}), it is out of range and does not have effect on decoding of forward channel.

2. $(u_1,v_2)$: Figure \ref{fig:2} depicts the received signal at node $B$ (ignoring the effect of transmitted signal from $B$) assuming that both relays use Relay Strategy 0.

First, consider the case that the repetitions in favor of the forward and backward channels happen in different relays, i.e., $m_1, n_2=0$. In this case, $R_2$ repeats in favor of the forward channel (uses Relay Strategy $(2,0)$) and $R_1$ repeats in favor of the backward channel (uses Relay Strategy $(0,2)$ or $(0,6)$). Repeating $a_{n_{A1}-n_{A2}+1},...,a_{n_{A1}-n_{A2}+n_{2B}-n_{1B}}$ by $R_1$ in favor of the backward direction communication does not affect the achievability of the forward channel because these signals have been decoded from the top levels of the received signal from the other relay ($R_2$) where there is no interference from $R_1$. Also, repeating the $a_{1},...,a_{n_{A1}-n_{A2}}$ by $R_1$ within the top $n_{A1}-n_{A2}$ streams (i.e., on blocks 2 and 3 in Figure \ref{fig:2}) does not affect the decoding since the signals can be decoded from top and their effect can be cancelled. Repeating any of them on the next $2 n_{A2}-n_{A1}+n_{1B}-n_{2B}$ lower levels by $R_1$ (i.e., on block 3 in Figure \ref{fig:2}) does not affect the decoding because we can decode those upper $n_{A1}-n_{A2}$ levels ($a_{1},...,a_{n_{A1}-n_{A2}}$) first and cancel the effect of the repeated signals. If $a_{1},...,a_{n_{A1}-n_{A2}}$ is repeated by $R_1$ on the next lower $2(n_{2B}-n_{1B}-n_{A2})+n_{A1}$ levels (i.e., on block 5 in Figure \ref{fig:2}), in case that the repetitions of some streams from two relays are from the same level and are repeated on the same level, i.e., they create an equation as higher levels, $R_2$ does not repeat for those levels as was explained in the definition of the strategies, so there is no problem in decoding of forward direction channel. Also if any of $a_{1},...,a_{n_{A1}-n_{A2}}$ is repeated by $R_1$ on the next $2 n_{1B}-n_{2B}+n_{A2}-n_{A1}$ levels (i.e., on block 6 in Figure \ref{fig:2}), it does not affect the decoding since the signals can be decoded from top and their effect can be cancelled. Also, if $a_{n_{A1}-n_{A2}+n_{2B}-n_{1B}+1},...,a_{n_{1B}}$ is repeated by $R_1$ in the same level range that they are located (i.e., on block 6 in Figure \ref{fig:2}), it does not affect the decoding since the signals can be decoded from top and their effect can be cancelled. If any of them is repeated in any lower level by $R_1$, it is out of range and does not have effect on decoding of forward channel.

Now, take the case that the repetitions in favor of the forward and backward channels happen in the same relay ($R_2$), i.e., $(m_1,n_1)=(0,0)$. In this case, $R_2$ repeats in favor of both directions ($(m_2,m_2)\in \{(2,2), (2,6)\}$) and $R_1$ uses Relay Strategy $(0,0)$. Repeating $a_{n_{A1}-n_{A2}+1},...,a_{n_{A1}-n_{A2}+n_{2B}-n_{1B}}$ in favor of the backward channel does not affect the achievability of the forward channel because these signals have been decoded from the top levels of the received signal from the same relay (i.e., on block 1 in Figure \ref{fig:2}). Also, repeating the signals $a_{n_{A1}-n_{A2}+n_{2B}-n_{1B}+1},...,a_{n_{1B}}$ does not affect the decoding because these signals are decoded from the last $2 n_{1B}-n_{2B}+n_{A2}-n_{A1}$ levels from the other relay, $R_1$ (i.e., on block 6 in Figure \ref{fig:2}). Also, repeating the $a_{n_{1B}+1},...,a_{2n_{A1}-2n_{A2}+n_{2B}-n_{1B}}$ in the same level range that they are located (i.e., on block 3 in Figure \ref{fig:2}) does not affect the decoding since the signals can be decoded from top and their effect can be cancelled. Repeating any of $a_{n_{1B}+1},...,a_{2n_{A1}-2n_{A2}+n_{2B}-n_{1B}}$ on the next $2 n_{A2}-n_{A1}+n_{1B}-n_{2B}$ lower levels (i.e., on block 4 in Figure \ref{fig:2}) does not affect the decoding because node $B$ decodes the upper $n_{A1}-n_{A2}$ levels ($a_{n_{1B}+1},...,a_{2n_{A1}-2n_{A2}+n_{2B}-n_{1B}}$) first and cancel the effect of the repeated signals. Repeating any of $a_{n_{1B}+1},...,a_{2n_{A1}-2n_{A2}+n_{2B}-n_{1B}}$ on the next lower $2(n_{2B}-n_{1B}-n_{A2})+n_{A1}$ levels (i.e., on block 5 in Figure \ref{fig:2}), it does not affect the decoding because as it was explained in the definition of the strategies, in case that the repetitions of some streams from two relays are from the same levels and are supposed to be repeated on the same level, the repeating relay ($R_2$) does not repeat for those levels and repeats those streams only one time. Also, repeating any of the $a_{2(n_{A1}-n_{A2})+n_{2B}-n_{1B}+1},...,a_{n_{A1}}$ in the same level range that they are located (i.e., on block 4 in Figure \ref{fig:2}), does not affect the decoding since the signals can be decoded from top and their effect can be cancelled. Repeating any of $a_{2(n_{A1}-n_{A2})+n_{2B}-n_{1B}+1},...,a_{n_{A1}}$ on the next lower $2(n_{2B}-n_{1B}-n_{A2})+n_{A1}$ levels (i.e., on block 5 in Figure \ref{fig:2}), does not affect the decoding because they are already decoded. If any of them is repeated in any lower level by $R_2$, it is out of range and does not have effect on decoding of forward channel.

3. $(u_1,v_3)$: Figure \ref{fig:3} depicts the received signal at node $B$ (ignoring the effect of transmitted signal from $B$) assuming that both relays use Relay Strategy 0.

First, consider the case that the repetitions in favor of the forward and backward channels happen in different relays, i.e., $m_1, n_2=0$. In this case, $R_2$ repeats in favor of the forward channel (uses Relay Strategy $(2,0)$) and $R_1$ repeats in favor of the backward channel (uses Relay Strategy $(0,2)$ or $(0,6)$). Repeating $a_{n_{A1}-n_{A2}+1},...,a_{n_{A1}-n_{A2}+n_{2B}-n_{1B}}$ by $R_1$ in favor of the backward direction communication does not affect the achievability of the forward channel because these signals have been decoded from the top levels of the received signal from the other relay ($R_2$) where there is no interference from $R_1$. Also, repeating the $a_{1},...,a_{n_{A1}-n_{A2}}$ by $R_1$ within the top $n_{A1}-n_{A2}$ streams (i.e., on block 2 in Figure \ref{fig:3}) does not affect the decoding since the signals can be decoded from top and their effect can be cancelled. Repeating any of them on the next $2 n_{A2}-n_{A1}+n_{1B}-n_{2B}$ lower levels by $R_1$ (i.e., on blocks 3 and 4 in Figure \ref{fig:3}) does not affect the decoding because we can decode those upper $n_{A1}-n_{A2}$ levels ($a_{1},...,a_{n_{A1}-n_{A2}}$) first and cancel the effect of the repeated signals. If $a_{1},...,a_{n_{A1}-n_{A2}}$ is repeated by $R_1$ on the next lower $2(n_{2B}-n_{1B}-n_{A2})+n_{A1}$ levels (i.e., on block 5 in Figure \ref{fig:3}), $R_2$ does not repeat for those levels as was explained in the definition of the strategies, so there is no problem in decoding of forward direction channel. Also if any of $a_{1},...,a_{n_{A1}-n_{A2}}$ is repeated by $R_1$ on the next lower $2 n_{1B}-n_{2B}+n_{A2}-n_{A1}$ levels (i.e., on block 6 in Figure \ref{fig:3}), it does not affect the decoding since the signals can be decoded from top and their effect can be cancelled. Also if $a_{n_{A1}-n_{A2}+n_{2B}-n_{1B}+1},...,a_{n_{1B}}$ is repeated by $R_1$ in the same level range that they are located (i.e., on block 5 in Figure \ref{fig:3}), it does not affect the decoding since the signals can be decoded from top and their effect can be cancelled. If any of them is repeated in any lower level by $R_1$, it is out of range and does not have effect on decoding of forward channel.

Now, take the case that the repetitions in favor of the forward and backward channels happen in the same relay ($R_2$), i.e., $(m_1,n_1)=(0,0)$. In this case, $R_2$ repeats in favor of both directions ($(m_2,m_2)\in \{(2,2), (2,6)\}$) and $R_1$ uses Relay Strategy $(0,0)$. Repeating $a_{n_{A1}-n_{A2}+1},...,a_{n_{A1}-n_{A2}+n_{2B}-n_{1B}}$ in favor of the backward channel does not affect the achievability of the forward channel because these signals have been decoded from the top levels of the received signal from the same relay (i.e., on block 1 in Figure \ref{fig:3}). Also, repeating the signals $a_{n_{A1}-n_{A2}+n_{2B}-n_{1B}+1},...,a_{n_{1B}}$ does not affect the decoding because these signals are decoded from the last $2 n_{1B}-n_{2B}+n_{A2}-n_{A1}$ levels from the other relay, $R_1$ (i.e., on block 6 in Figure \ref{fig:3}). Also, repeating the $a_{n_{1B}+1},...,a_{n_{A1}}$ in the same level range that they are located (i.e., on block 4 in Figure \ref{fig:3}), does not affect the decoding since the signals can be decoded from top and their effect can be cancelled. Repeating any of $a_{n_{1B}+1},...,a_{n_{A1}}$ on the next $n_{A1}-2(n_{A2}+n_{1B}-n_{2B})$ lower levels (i.e., on block 5 and 6 in Figure \ref{fig:3}) does not affect the decoding, since $R_2$ does not repeat in favor of the forward direction for those levels with these parameters as was explained in the definition of the strategies, so there is no problem in decoding of forward direction channel. If any of them is repeated in any lower level by $R_2$, it is out of range and does not have effect on decoding of forward channel.

4. $(u_1,v_4,r_1)$: Figure \ref{fig:4.1} depicts the received signal at node $B$ (ignoring the effect of transmitted signal from $B$) assuming that both relays use Relay Strategy 0.

First, consider the case that the repetitions in favor of the forward and backward channels happen in different relays, i.e., $m_1, n_2=0$. In this case, $R_2$ repeats in favor of the forward channel (uses Relay Strategy $(2,0)$) and $R_1$ repeats in favor of the backward channel (uses Relay Strategy $(0,2)$ or $(0,6)$). Repeating $a_{n_{A1}-n_{A2}+1},...,a_{n_{A1}-n_{A2}+n_{2B}-n_{1B}}$ by $R_1$ in favor of the backward direction communication does not affect the achievability of the forward channel because these signals have been decoded from the top levels of the received signal from the other relay ($R_2$) where there is no interference from $R_1$. Also, repeating the $a_{1},...,a_{n_{A1}-n_{A2}}$ by $R_1$ within the top $n_{A1}-n_{A2}$ streams (i.e., on blocks 2 and 3 in Figure \ref{fig:4.1}) does not affect the decoding since the signals can be decoded from top and their effect can be cancelled. Repeating any of them on the next $n_{2B}-n_{1B}$ lower levels by $R_1$ (i.e., on blocks 4 and 5 in Figure \ref{fig:4.1}) does not affect the decoding because we can decode those upper $n_{A1}-n_{A2}$ levels ($a_{1},...,a_{n_{A1}-n_{A2}}$) first and cancel the effect of the repeated signals. If $a_{1},...,a_{n_{A1}-n_{A2}}$ is repeated by $R_1$ on the next lower $2(n_{1B}-n_{2B}+n_{A2})-n_{A1}$ levels (i.e., on block 6 in Figure \ref{fig:4.1}), does not affect the decoding because node $B$ decodes these levels at the end. If $a_{1},...,a_{n_{A1}-n_{A2}}$ is repeated on the next $n_{2B}-n_{A2}$ levels (i.e., on blocks 7 and 8 in Figure \ref{fig:4.1}), in case that the repetitions of some streams from two relays are from the same level and are repeated on the same level, i.e., they create an equation as higher levels, $R_2$ does not repeat for those levels as was explained in the definition of the strategies, so there is no problem in decoding of forward direction channel. Also if $a_{n_{A1}-n_{A2}+n_{2B}-n_{1B}+1},...,a_{n_{A2}-n_{2B}+n_{1B}}$ is repeated by $R_1$ in the same level range that they are located (i.e., on block 6 in Figure \ref{fig:4.1}), it does not affect the decoding since the signals can be decoded from top and their effect can be cancelled. If $a_{n_{A1}-n_{A2}+n_{2B}-n_{1B}+1},...,a_{n_{A2}-n_{2B}+n_{1B}}$ is repeated on the next $n_{2B}-n_{A2}$ levels (i.e., on blocks 7 and 8 in Figure \ref{fig:4.1}), in case that the repetitions of some streams from two relays are from the same level and are repeated on the same level, i.e., they create an equation as higher levels, $R_2$ does not repeat for those levels as was explained in the definition of the strategies, so there is no problem in decoding of forward direction channel. Also if $a_{n_{A1}-n_{A2}+n_{2B}-n_{1B}+1},...,a_{n_{1B}}$ is repeated by $R_1$ in the same level range that they are located, it does not affect the decoding since the signals can be decoded from top and their effect can be cancelled. If any of them is repeated in any lower level by $R_1$, it is out of range and does not have effect on decoding of forward channel.

Now, take the case that the repetitions in favor of the forward and backward channels happen in the same relay ($R_2$), i.e., $(m_1,n_1)=(0,0)$. In this case, $R_2$ repeats in favor of both directions ($(m_2,m_2)\in \{(2,2), (2,6)\}$) and $R_1$ uses Relay Strategy $(0,0)$. Repeating $a_{n_{A1}-n_{A2}+1},...,a_{n_{A1}-n_{A2}+n_{2B}-n_{1B}}$ in favor of the backward channel does not affect the achievability of the forward channel because these signals have been decoded from the top levels of the received signal from the same relay (i.e., on block 1 in Figure \ref{fig:4.1}). Also, repeating the $a_{n_{A1}-n_{A2}+n_{2B}-n_{1B}+1},...,a_{n_{A2}-n_{2B}+n_{1B}}$ does not affect the decoding because these signals are decoded from the repetitions by $R_2$ (i.e., on block 8 in Figure \ref{fig:4.1}). Also, repeating the $a_{n_{A2}-n_{2B}+n_{1B}+1},...,a_{2n_{A1}-2n_{A2}+n_{2B}-n_{1B}}$ does not affect the decoding because these signals are decoded from the low levels from the other relay without interference, ($R_1$)  (i.e., on block 3 in Figure \ref{fig:4.1}). Also, repeating the $a_{2n_{A1}-2n_{A2}+n_{2B}-n_{1B}+1},...,a_{2(n_{A1}-n_{A2}+n_{2B}-n_{1B})}$ in the same level range that they are located  (i.e., on blocks 4 and 5 in Figure \ref{fig:4.1}), does not affect the decoding since their interference is already decoded from highest block and the signals can be decoded from top and their effect can be cancelled. Repeating any of $a_{2n_{A1}-2n_{A2}+n_{2B}-n_{1B}+1},...,a_{2(n_{A1}-n_{A2}+n_{2B}-n_{1B})}$ on the next $2(n_{A2}-n_{2B}+n_{1B})-n_{A1}$ lower levels (i.e., on block 6 in Figure \ref{fig:4.1}) does not affect the decoding because $B$ decodes the upper $n_{2B}-n_{1B}$ levels ($a_{2n_{A1}-2n_{A2}+n_{2B}-n_{1B}+1},...,a_{2(n_{A1}-n_{A2}+n_{2B}-n_{1B})}$) first and cancel the effect of the repeated signals. Repeating any of $a_{2n_{A1}-2n_{A2}+n_{2B}-n_{1B}+1},...,a_{2(n_{A1}-n_{A2}+n_{2B}-n_{1B})}$ on the next $n_{2B}-n_{A2}$ levels (i.e., on blocks 7 and 8 in Figure \ref{fig:4.1}), does not affect the decoding because as it was explained in the definition of the strategies, in case that the repetitions of some streams from two relays are from the same levels and are supposed to be repeated on the same level, the repeating relay ($R_2$) does not repeat for those levels and repeats those streams only one time. Also, repeating any of the $a_{2(n_{A1}-n_{A2}+n_{2B}-n_{1B})+1},...,a_{n_{A1}}$ does not affect the decoding since the signals can be decoded from top and their effect can be cancelled. If any of them is repeated in any lower level by $R_2$, it is out of range and does not have effect on decoding of forward channel.

5. $(u_1,v_4,r_2,s_1)$: Figure \ref{fig:4.21} depicts the received signal at node $B$ (ignoring the effect of transmitted signal from $B$) assuming that both relays use Relay Strategy 0.

First, consider the case that the repetitions in favor of the forward and backward channels happen in different relays, i.e., $m_1, n_2=0$. In this case, $R_2$ repeats in favor of the forward channel (uses Relay Strategy $(2,0)$) and $R_1$ repeats in favor of the backward channel (uses Relay Strategy $(0,2)$ or $(0,6)$). Repeating $a_{n_{A1}-n_{A2}+1},...,a_{n_{A1}-n_{A2}+n_{2B}-n_{1B}}$ by $R_1$ in favor of the backward direction communication does not affect the achievability of the forward channel because they have been decoded from the top levels of the received signal from the other relay ($R_2$) where there is no interference from $R_1$. Also, repeating the $a_{1},...,a_{n_{A1}-n_{A2}}$ by $R_1$ within the top $n_{A1}-n_{A2}$ streams (i.e., on blocks 2 and 3 in Figure \ref{fig:4.21}) does not affect the decoding since the signals can be decoded from top and their effect can be cancelled. Repeating any of them on the next $n_{2B}-n_{1B}$ lower levels by $R_1$ (i.e., on block 4 in Figure \ref{fig:4.21}) does not affect the decoding because we can decode those upper $n_{A1}-n_{A2}$ levels ($a_{1},...,a_{n_{A1}-n_{A2}}$) first and cancel the effect of the repeated signals. If $a_{1},...,a_{n_{A1}-n_{A2}}$ is repeated by $R_1$ on the next lower $2(n_{A2}-n_{A1}-n_{2B})+3n_{1B}$ levels (i.e., on block 5 in Figure \ref{fig:4.21}) does not affect the decoding because node $B$ decodes them from the last $2(n_{A2}-n_{A1}-n_{2B})+3n_{1B}$ levels from $R_1$ (i.e., on block 9 in Figure \ref{fig:4.21}). If $a_{1},...,a_{n_{A1}-n_{A2}}$ is repeated on the next lower $n_{A1}-n_{1B}$ levels (i.e., on blocks 6 and 7 in Figure \ref{fig:4.21}), it does not affect the decoding because node $B$ decodes these levels at the end. Also if $a_{1},...,a_{n_{A1}-n_{A2}}$ is repeated on the next lower $n_{2B}-n_{A2}$ levels, in case that the repetitions of some streams from two relays are from the same level and are repeated on the same level, i.e., they create an equation as higher levels, $R_2$ does not repeat for those levels as was explained in the definition of the strategies, so there is no problem in decoding of forward direction channel. Also if $a_{n_{A1}-n_{A2}+n_{2B}-n_{1B}+1},...,a_{n_{1B}}$ is repeated by $R_1$ in the same level range that they are located, it does not affect the decoding since the signals can be decoded from top and their effect can be cancelled.  If any of them is repeated in any lower level by $R_1$, it is out of range and does not have effect on decoding of forward channel.

Now, take the case that the repetitions in favor of the forward and backward channels happen in the same relay ($R_2$), i.e., $(m_1,n_1)=(0,0)$. In this case, $R_2$ repeats in favor of both directions ($(m_2,m_2)\in \{(2,2), (2,6)\}$) and $R_1$ the Relay Strategy $(0,0)$. Repeating $a_{n_{A1}-n_{A2}+1},...,a_{n_{A1}-n_{A2}+n_{2B}-n_{1B}}$ in favor of the backward channel does not affect the achievability of the forward channel because they have been decoded from the top levels of the received signal from the same relay (i.e., on block 1 in Figure \ref{fig:4.21}). Also, repeating the $a_{2(n_{A1}-n_{A2}+n_{2B}-n_{1B})+1},...,a_{n_{1B}}$ does not affect the decoding because they are decoded from the last $3 n_{1B}+2(n_{A2}-n_{A1}-n_{2B})$ levels from the other relay ($R_1$). Also, repeating the $a_{n_{A1}-n_{A2}+n_{2B}-n_{1B}+1},...,a_{n_{A2}-n_{A1}+2n_{1B}-n_{2B}}$ does not affect the decoding because these signals are decoded from lower levels (block 5 in Figure \ref{fig:4.21}). Also, repeating the $a_{n_{A2}-n_{A1}+2n_{1B}-n_{2B}+1},...,a_{2(n_{A1}-n_{A2})+n_{2B}-n_{1B}}$ does not affect the decoding because these signals are decoded from the repetitions by $R_2$ (i.e., on block 8 in Figure \ref{fig:4.21}). Repeating any of $a_{2n_{1B}-n_{2B}+n_{A2}-n_{A1}+1},...,a_{2n_{A1}-2n_{A2}+n_{2B}-n_{1B}}$ on the low levels (i.e., on block 8 in Figure \ref{fig:4.21}), does not affect the decoding because as it was explained in the definition of the strategies, in case that the repetitions of some streams from two relays are from the same levels and are supposed to be repeated on the same level, the repeating relay ($R_2$) does not repeat for those levels and repeats those streams only one time. If $a_{2(n_{A1}-n_{A2})+n_{2B}-n_{1B}+1},...,a_{2(n_{A1}-n_{A2}+n_{2B}-n_{1B})}$ is repeated within the range that they are located (i.e., on block 4 in Figure \ref{fig:4.21}), it does not affect the decoding since the signals can be decoded from highest level of it and their effect can be cancelled. Also if $a_{2(n_{A1}-n_{A2})+n_{2B}-n_{1B}+1},...,a_{2(n_{A1}-n_{A2}+n_{2B}-n_{1B})}$ is repeated within the lower levels, it does not affect the decoding since they are decoded first from upper block. If $a_{n_{1B}+1},...,a_{n_{A1}}$ is repeated within the range that they are located (i.e., on blocks 6 and 7 in Figure \ref{fig:4.21}), it does not affect the decoding since the signals can be decoded from highest level of it and their effect can be cancelled. Also if If $a_{n_{1B}+1},...,a_{n_{A1}}$ is repeated within the lower levels, it does not affect the decoding since they are decoded first from upper block. If any of them is repeated in any lower level by $R_2$, it is out of range and does not have effect on decoding of forward channel.

6. $(u_1,v_4,r_2,s_2)$: Figure \ref{fig:4.22} depicts the received signal at node $B$ (ignoring the effect of transmitted signal from $B$) assuming that both relays use Relay Strategy 0.

First, consider the case that the repetitions in favor of the forward and backward channels happen in different relays, i.e., $m_1, n_2=0$. In this case, $R_2$ repeats in favor of the forward channel (uses Relay Strategy $(2,0)$) and $R_1$ repeats in favor of the backward channel (uses Relay Strategy $(0,2)$ or $(0,6)$). Repeating $a_{n_{A1}-n_{A2}+1},...,a_{n_{A1}-n_{A2}+n_{2B}-n_{1B}}$ by $R_1$ in favor of the backward direction communication does not affect the achievability of the forward channel because they have been decoded from the top levels of the received signal from the other relay ($R_2$) where there is no interference from $R_1$ (i.e., on block 1 in Figure \ref{fig:4.22}). Also, repeating the $a_{1},...,a_{n_{A1}-n_{A2}}$ by $R_1$ within the top $n_{A1}-n_{A2}$ streams (i.e., on block 2 in Figure \ref{fig:4.22}) does not affect the decoding since the signals can be decoded from top and their effect can be cancelled. Repeating any of them on the next $n_{2B}-n_{1B}$ lower levels by $R_1$ (i.e., on block 3 in Figure \ref{fig:4.22}) does not affect the decoding because we can decode those upper $n_{A1}-n_{A2}$ levels ($a_{1},...,a_{n_{A1}-n_{A2}}$) first and cancel the effect of the repeated signals. If $a_{1},...,a_{n_{A1}-n_{A2}}$ is repeated by $R_1$ on the next lower $2(n_{A2}-n_{A1}-n_{2B})+3n_{1B}$ levels (i.e., on blocks 4 and 5 in Figure \ref{fig:4.22}) does not affect the decoding because node $B$ decodes them from the last $2(n_{A2}-n_{A1}-n_{2B})+3n_{1B}$ levels from $R_1$ (i.e., on block 8 in Figure \ref{fig:4.22}). If $a_{1},...,a_{n_{A1}-n_{A2}}$ is repeated on the next lower $n_{A1}-n_{1B}$ levels (i.e., on block 6 in Figure \ref{fig:4.22}), it does not affect the decoding because node $B$ decodes these levels at the end. Also if $a_{1},...,a_{n_{A1}-n_{A2}}$ is repeated on the next $n_{2B}-n_{A2}$ lower levels, $R_2$ does not repeat for those levels as was explained in the definition of the strategies, so there is no problem in decoding of forward direction channel. Also if $a_{n_{A1}-n_{A2}+n_{2B}-n_{1B}+1},...,a_{n_{1B}}$ is repeated by $R_1$ in the same level range that they are located, it does not affect the decoding since the signals can be decoded from top and their effect can be cancelled.  If any of them is repeated in any lower level by $R_1$, it is out of range and does not have effect on decoding of forward channel.

Now, take the case that the repetitions in favor of the forward and backward channels happen in the same relay ($R_2$), i.e., $(m_1,n_1)=(0,0)$. In this case, $R_2$ repeats in favor of both directions ($(m_2,m_2)\in \{(2,2), (2,6)\}$) and $R_1$ uses Relay Strategy $(0,0)$. Repeating $a_{n_{A1}-n_{A2}+1},...,a_{n_{A1}-n_{A2}+n_{2B}-n_{1B}}$ in favor of the backward channel does not affect the achievability of the forward channel because they have been decoded from the top levels of the received signal from the same relay (i.e., on block 1 in Figure \ref{fig:4.22}). Also, repeating the $a_{2(n_{A1}-n_{A2}+n_{2B}-n_{1B})+1},...,a_{n_{1B}}$ does not affect the decoding because they are decoded from the last $3 n_{1B}+2(n_{A2}-n_{A1}-n_{2B})$ levels from the other relay ($R_1$) (i.e., on block 8 in Figure \ref{fig:4.22}). Also, repeating the $a_{n_{A1}-n_{A2}+n_{2B}-n_{1B}+1},...,a_{2(n_{A1}-n_{A2})+n_{2B}-n_{1B}}$ does not affect the decoding because these signals are decoded from lower levels (block 4 in Figure \ref{fig:4.22}). Repeating any of $a_{n_{2B}-n_{1B}+n_{A1}-n_{A2}+1},...,a_{2n_{A1}-2n_{A2}+n_{2B}-n_{1B}}$ on the low levels (i.e., on blocks 7 and 8 in Figure \ref{fig:4.21}), does not affect the decoding because as it was explained in the definition of the strategies, the repeating relay ($R_2$) does not repeat in favor of the forward channel for these set of parameters. If $a_{2(n_{A1}-n_{A2})+n_{2B}-n_{1B}+1},...,a_{2(n_{A1}-n_{A2}+n_{2B}-n_{1B})}$ is repeated within the range that they are located (i.e., on block 3 in Figure \ref{fig:4.22}), it does not affect the decoding since the signals can be decoded from highest level of it and their effect can be cancelled. Also if $a_{2(n_{A1}-n_{A2})+n_{2B}-n_{1B}+1},...,a_{2(n_{A1}-n_{A2}+n_{2B}-n_{1B})}$ is repeated within the lower levels, it does not affect the decoding since they are decoded first from upper block. If $a_{n_{1B}+1},...,a_{n_{A1}}$ is repeated within the range that they are located (i.e., on block 6 in Figure \ref{fig:4.22}), it does not affect the decoding since the signals can be decoded from highest level of it and their effect can be cancelled. Also if If $a_{n_{1B}+1},...,a_{n_{A1}}$ is repeated within the lower levels, it does not affect the decoding since they are decoded first from upper block. If any of them is repeated in any lower level by $R_2$, it is out of range and does not have effect on decoding of forward channel.

7. $(u_1,v_4,r_2,s_3)$: Figure \ref{fig:4.23} depicts the received signal at node $B$ (ignoring the effect of transmitted signal from $B$) assuming that both relays use Relay Strategy 0.

First, consider the case that the repetitions in favor of the forward and backward channels happen in different relays, i.e., $m_1, n_2=0$. In this case, $R_2$ repeats in favor of the forward channel (uses Relay Strategy $(2,0)$) and $R_1$ repeats in favor of the backward channel (uses Relay Strategy $(0,2)$ or $(0,6)$). Repeating $a_{n_{A1}-n_{A2}+1},...,a_{n_{A1}-n_{A2}+n_{2B}-n_{1B}}$ by $R_1$ in favor of the backward direction communication does not affect the achievability of the forward channel because they have been decoded from the top levels of the received signal from the other relay ($R_2$) where there is no interference from $R_1$ (i.e., on block 1 in Figure \ref{fig:4.23}). Also, repeating the $a_{1},...,a_{n_{A1}-n_{A2}}$ by $R_1$ within the top $n_{A1}-n_{A2}$ streams (i.e., on block 2 in Figure \ref{fig:4.23}) does not affect the decoding since the signals can be decoded from top and their effect can be cancelled. Repeating any of them on the next $n_{2B}-n_{1B}$ lower levels by $R_1$ (i.e., on block 3 in Figure \ref{fig:4.23}) does not affect the decoding because we can decode those upper $n_{A1}-n_{A2}$ levels ($a_{1},...,a_{n_{A1}-n_{A2}}$) first and cancel the effect of the repeated signals. If $a_{1},...,a_{n_{A1}-n_{A2}}$ is repeated by $R_1$ on the next lower $2(n_{A2}-n_{A1}-n_{2B})+3n_{1B}$ levels (i.e., on blocks 4 and 5 in Figure \ref{fig:4.23}) does not affect the decoding because node $B$ decodes them from the last $2(n_{A2}-n_{A1}-n_{2B})+3n_{1B}$ levels from $R_1$. If $a_{1},...,a_{n_{A1}-n_{A2}}$ is repeated on the next lower $n_{A1}-n_{1B}$ levels (i.e., on blocks 6 and 7 in Figure \ref{fig:4.23}), it does not affect the decoding because node $B$ decodes these levels at the end. Also if $a_{1},...,a_{n_{A1}-n_{A2}}$ is repeated on the next lower $n_{2B}-n_{A2}$ levels (i.e., on block 8 in Figure \ref{fig:4.23}), in case that the repetitions of some streams from two relays are from the same level and are repeated on the same level, i.e., they create an equation as higher levels, $R_2$ does not do the repeating for those levels as was explained in the definition of the strategies, so there is no problem in decoding of forward direction channel. Also if $a_{n_{A1}-n_{A2}+n_{2B}-n_{1B}+1},...,a_{n_{1B}}$ is repeated by $R_1$ in the same level range that they are located, it does not affect the decoding since the signals can be decoded from top and their effect can be cancelled.  If any of them is repeated in any lower level by $R_1$, it is out of range and does not have effect on decoding of forward channel.

Now, take the case that the repetitions in favor of the forward and backward channels happen in the same relay ($R_2$), i.e., $(m_1,n_1)=(0,0)$. In this case, $R_2$ repeats in favor of both directions ($(m_2,m_2)\in \{(2,2), (2,6)\}$) and $R_1$ uses Relay Strategy $(0,0)$. Repeating $a_{n_{A1}-n_{A2}+1},...,a_{n_{A1}-n_{A2}+n_{2B}-n_{1B}}$ in favor of the backward channel does not affect the achievability of the forward channel because they have been decoded from the top levels of the received signal from the same relay (i.e., on block 1 in Figure \ref{fig:4.23}). Also, repeating the $a_{2(n_{A1}-n_{A2}+n_{2B}-n_{1B})+1},...,a_{n_{1B}}$ does not affect the decoding because they are decoded from the last $3 n_{1B}+2(n_{A2}-n_{A1}-n_{2B})$ levels from the other relay ($R_1$). Repeating any of $a_{2(n_{2B}-n_{1B}+n_{A1}-n_{A2})+1},...,a_{3(n_{A1}-n_{A2})+2(n_{2B}-n_{1B})}$ on the lowest levels (i.e., on block 8 in Figure \ref{fig:4.23}), does not affect the decoding because as it was explained in the definition of the strategies, in case that the repetitions of some streams from two relays are from the same levels and are supposed to be repeated on the same level, the repeating relay ($R_2$) does not do the repeating for those levels and repeats those streams only one time. Also, repeating the $a_{n_{A1}-n_{A2}+n_{2B}-n_{1B}+1},...,a_{2(n_{A1}-n_{A2})+n_{2B}-n_{1B}}$ does not affect the decoding because these signals are decoded from lower levels (i.e. the block 4 in Figure \ref{fig:4.23}). If $a_{2(n_{A1}-n_{A2})+n_{2B}-n_{1B}+1},...,a_{2(n_{A1}-n_{A2}+n_{2B}-n_{1B})}$ is repeated within the range that they are located (i.e., on block 3 in Figure \ref{fig:4.23}), it does not affect the decoding since the signals can be decoded from highest level of it and their effect can be cancelled. Also if $a_{2(n_{A1}-n_{A2})+n_{2B}-n_{1B}+1},...,a_{2(n_{A1}-n_{A2}+n_{2B}-n_{1B})}$ is repeated within the lower levels, it does not affect the decoding since they are decoded first from upper block. If $a_{n_{1B}+1},...,a_{n_{A1}}$ is repeated within the range that they are located (i.e., on blocks 6 and 7 in Figure \ref{fig:4.23}), it does not affect the decoding since the signals can be decoded from highest level of it and their effect can be cancelled. Also if If $a_{n_{1B}+1},...,a_{n_{A1}}$ is repeated within the lower levels, it does not affect the decoding since they are decoded first from upper block. If any of them is repeated in any lower level by $R_2$, it is out of range and does not have effect on decoding of forward channel.

8. $(u_2,w_1)$: Figure \ref{fig:5} depicts the received signal at node $B$ (ignoring the effect of transmitted signal from $B$) assuming that both relays use Relay Strategy 0.

First, consider the case that the repetitions in favor of the forward and backward channels happen in different relays, i.e., $m_1, n_2=0$. In this case, $R_2$ repeats in favor of the forward channel (uses Relay Strategy $(2,0)$) and $R_1$ repeats in favor of the backward channel (uses Relay Strategy $(0,2)$ or $(0,6)$). Repeating $a_{n_{A1}-n_{A2}+1},...,a_{n_{1B}}$ by $R_1$ in favor of the backward direction communication does not affect the achievability of the forward channel because they have been decoded from the top levels of the received signal from the other relay ($R_2$) where there is no interference from $R_1$ (i.e., on block 1 in Figure \ref{fig:5}). Also, repeating the $a_{1},...,a_{n_{A2}-(n_{2B}-n_{1B})}$ by $R_1$ within the top $n_{A2}-n_{2B}+n_{1B}$ streams  (i.e., on block 2 in Figure \ref{fig:5}) does not affect the decoding since the signals can be decoded from top and their effect can be cancelled. Repeating any of them on the next $n_{A1}-2n_{A2}+n_{2B}-n_{1B}$ lower levels by $R_1$ (i.e., on block 3 in Figure \ref{fig:5}) does not affect the decoding because we can decode those upper $n_{A2}-(n_{2B}-n_{1B})$ levels ($a_{1},...,a_{n_{A2}-(n_{2B}-n_{1B})}$) first and cancel the effect of the repeated signals. If $a_{1},...,a_{n_{A2}-(n_{2B}-n_{1B})}$ is repeated by $R_1$ on the next lower $n_{1B}-n_{A1}+n_{A2}$ levels (i.e., on block 3 in Figure \ref{fig:5}), in case that the repetitions of some streams from two relays are from the same level and are repeated on the same level, i.e., they create an equation as higher levels, $R_2$ does not do the repeating for those levels as was explained in the definition of the strategies, so there is no problem in decoding of forward direction channel. If $a_{n_{A2}-n_{2B}+n_{1B}+1},...,a_{n_{A1}-n_{A2}}$ is repeated within the range that they are located (i.e., on block 3 in Figure \ref{fig:5}), it does not affect the decoding since the signals can be decoded from highest level of it and their effect can be cancelled. Also if $a_{n_{A2}-n_{2B}+n_{1B}+1},...,a_{n_{A1}-n_{A2}}$ is repeated within the next lower $n_{1B}-n_{A1}+n_{A2}$ levels, it does not affect the decoding since they are decoded first from upper block. If any of them is repeated in any lower level by $R_1$, it is out of range and does not have effect on decoding of forward channel.

Now, take the case that the repetitions in favor of the forward and backward channels happen in the same relay ($R_2$), i.e., $(m_1,n_1)=(0,0)$. In this case, $R_2$ repeats in favor of both directions ($(m_2,m_2)\in \{(2,2), (2,6)\}$) and $R_1$ uses Relay Strategy $(0,0)$. Repeating $a_{n_{A1}-n_{A2}+1},...,a_{n_{A1}-n_{A2}+n_{2B}-n_{1B}}$ in favor of the backward channel does not affect the achievability of the forward channel because these signals have been decoded from the top levels of the received signal from the same relay ($R_2$). Also, repeating the $a_{n_{A1}-n_{A2}+n_{2B}-n_{1B}+1},...,a_{n_{A1}}$ in the same level range that they are located (i.e., on block 2 in Figure \ref{fig:5}) does not affect the decoding since the signals can be decoded from top and their effect can be cancelled. Repeating any of them on the next lower $n_{A1}-2n_{A2}+n_{2B}-n_{1B}$ levels (i.e., on block 3 in Figure \ref{fig:5}) does not affect the decoding because node $B$ decodes those upper $n_{A2}-(n_{2B}-n_{1B})$ levels ($a_{n_{A1}-n_{A2}+n_{2B}-n_{1B}+1},...,a_{n_{A1}}$) first and cancel the effect of the repeated signals. If $a_{n_{A1}-n_{A2}+n_{2B}-n_{1B}+1},...,a_{n_{A1}}$ is repeated on the next $n_{1B}+n_{A2}-n_{A1}$ levels (i.e., on block 4 in Figure \ref{fig:5}), it does not affect the decoding because as it was explained in the definition of the strategies, in case that the repetitions of some streams from two relays are from the same levels and are supposed to be repeated on the same level, the repeating relay ($R_2$) does not do the repeating for those levels and repeats those streams only one time. If any of them is repeated in any lower level by $R_2$, it is out of range and does not have effect on decoding of forward channel.

9. $(u_2,w_2)$: Figure \ref{fig:6} depicts the received signal at node $B$ (ignoring the effect of transmitted signal from $B$) assuming that both relays use Relay Strategy 0.

First, consider the case that the repetitions in favor of the forward and backward channels happen in different relays, i.e., $m_1, n_2=0$. In this case, $R_2$ repeats in favor of the forward channel (uses Relay Strategy $(2,0)$) and $R_1$ repeats in favor of the backward channel (uses Relay Strategy $(0,2)$ or $(0,6)$). Repeating $a_{n_{A1}-n_{A2}+1},...,a_{n_{A1}-n_{A2}+n_{2B}-n_{1B}}$ by $R_1$ in favor of the backward direction communication does not affect the achievability of the forward channel because they have been decoded from the top levels of the received signal from the other relay ($R_2$) where there is no interference from $R_1$. Also, repeating the $a_{1},...,a_{n_{A2}-(n_{2B}-n_{1B})}$ by $R_1$ within the top $n_{A2}-n_{2B}+n_{1B}$ streams  (i.e., on block 2 in Figure \ref{fig:6}) does not affect the decoding since the signals can be decoded from top and their effect can be cancelled. Repeating any of them on the next lower $n_{A1}-2n_{A2}+n_{2B}-n_{1B}$ levels by $R_1$ (i.e., on blocks 3 and 4 in Figure \ref{fig:6}) does not affect the decoding because node $B$ decodes those upper $n_{A2}-(n_{2B}-n_{1B})$ levels ($a_{1},...,a_{n_{A2}-(n_{2B}-n_{1B})}$) first and cancel the effect of the repeated signals. If $a_{1},...,a_{n_{A2}-(n_{2B}-n_{1B})}$ is repeated by $R_1$ on the next lower $n_{2B}-n_{1B}$ levels (i.e., on block 5 in Figure \ref{fig:6}), in case that the repetitions of some streams from two relays are from the same level and are repeated on the same level, i.e., they create an equation as higher levels, $R_2$ does not do the repeating for those levels as was explained in the definition of the strategies, so there is no problem in decoding of forward direction channel. If $a_{n_{A2}-n_{2B}+n_{1B}+1},...,a_{n_{A1}-n_{A2}}$ is repeated within the range that they are located (i.e., on block 4 in Figure \ref{fig:6}), it does not affect the decoding since the signals can be decoded from highest level of it and their effect can be cancelled. Also if $a_{n_{A2}-n_{2B}+n_{1B}+1},...,a_{n_{A1}-n_{A2}}$ is repeated within the next lower $n_{1B}-n_{A1}+n_{A2}$ levels, it does not affect the decoding since they are decoded first from upper block. If any of them is repeated in any lower level by $R_1$, it is out of range and does not have effect on decoding of forward channel.

Now, take the case that the repetitions in favor of the forward and backward channels happen in the same relay ($R_2$), i.e., $(m_1,n_1)=(0,0)$. In this case, $R_2$ repeats in favor of both directions ($(m_2,m_2)\in \{(2,2), (2,6)\}$) and $R_1$ uses Relay Strategy $(0,0)$. Repeating $a_{n_{A1}-n_{A2}+1},...,a_{n_{A1}-n_{A2}+n_{2B}-n_{1B}}$ in favor of the backward channel does not affect the achievability of the forward channel because these signals have been decoded from the top levels of the received signal from the same relay. Also, repeating the $a_{n_{A1}-n_{A2}+n_{2B}-n_{1B}+1},...,a_{n_{A1}}$ in the same level range that they are located does not affect the decoding since the signals can be decoded from top and their effect can be cancelled. Repeating any of them on the next $n_{A1}-2n_{A2}+n_{2B}-n_{1B}$ lower levels (i.e., on block 4 in Figure \ref{fig:6}), does not affect the decoding because $B$ decodes those upper $n_{A2}-(n_{2B}-n_{1B})$ levels ($a_{n_{A1}-n_{A2}+n_{2B}-n_{1B}+1},...,a_{n_{A1}}$) first and cancel the effect of the repeated signals. If $a_{n_{A1}-n_{A2}+n_{2B}-n_{1B}+1},...,a_{n_{A1}}$ is repeated on the next $n_{1B}+n_{A2}-n_{A1}$ lower levels (i.e., on block 5 in Figure \ref{fig:6}), it does not affect the decoding because as it was explained in the definition of the strategies, in case that the repetitions of some streams from two relays are from the same levels and are supposed to be repeated on the same level, the repeating relay ($R_2$) does not do the repeating for those levels and repeats those streams only one time. If any of them is repeated in any lower level by $R_2$, it is out of range and does not have effect on decoding of forward channel.

\subsection{Forward channel is of Case 4.1.2}

Now let assume that the forward channel is of Case 4.1.2 Type 1. For Type 2 proof is similar and thus is omitted. Then we have $C_{AB}=\max\{n_{1B},n_{2B}\}=n_{B1}$. We consider the partition shown in Figure \ref{fig:sup2} and show that the message can be decoded when the backward channel is in any of Cases 3.1.2 or 4.1.2. It can be seen that all the streams can be decoded with the same order as in Relay Strategy 6 in Appendix \ref{apdx_sc2}.

1. $(u_1,v_1)$: Figure \ref{fig:11.1} depicts the received signal at node $B$ (ignoring the effect of transmitted signal from $B$) assuming that both relays use Relay Strategy 0.

First, consider the case that the repetitions in favor of the forward and backward channels happen in different relays, i.e., $m_2, n_1=0$. In this case, $R_1$ repeats in favor of the forward channel (uses Relay Strategy $(6,0)$) and $R_2$ repeats in favor of the backward channel (uses Relay Strategy $(0,2)$ or $(0,6)$). Repeating $a_{n_{A2}-n_{A1}+1},...,a_{n_{1B}-n_{2B}+n_{A2}-n_{A1}}$ by $R_2$ in favor of the backward direction communication does not affect the achievability of the forward channel because they have been decoded from the top levels of the received signal from the other relay ($R_1$) where there is no interference from $R_2$. Also, repeating the $a_{1},...,a_{n_{A2}-n_{A1}}$ by $R_2$ within the top $n_{A2}-n_{A1}$ streams  (i.e., on block 3 in Figure \ref{fig:11.1}) does not affect the decoding since the signals can be decoded from top and their effect can be cancelled. Repeating any of them on the next $n_{2B}-2(n_{A2}-n_{A1})$ lower levels by $R_2$ (i.e., on block 4 in Figure \ref{fig:11.1}) does not affect the decoding because we can decode those upper $n_{A2}-n_{A1}$ levels ($a_{1},...,a_{n_{A2}-n_{A1}}$) first and cancel the effect of the repeated signals. If $a_{1},...,a_{n_{A2}-n_{A1}}$ is repeated by $R_2$ on the next lower levels (i.e., on block 5 in Figure \ref{fig:11.1}), in case that the repetitions of some streams from two relays are from the same level and are repeated on the same level, i.e., they create an equation as higher levels, $R_1$ does not do the repeating for those levels as was explained in the definition of the strategies, so there is no problem in decoding of forward direction channel. If any of them is repeated in any lower level by $R_2$, it is out of range and does not have effect on decoding of forward channel.

Now, take the case that the repetitions in favor of the forward and backward channels happen in the same relay ($R_1$), i.e., $(m_2,n_2)=(0,0)$. In this case, $R_1$ repeats in favor of both directions ($(m_1,m_1)\in \{(6,2), (6,6)\}$) and $R_2$ uses Relay Strategy $(0,0)$. Repeating $a_{n_{A2}-n_{A1}+1},...,a_{n_{1B}-n_{2B}+n_{A2}-n_{A1}}$ in favor of the backward channel does not affect the achievability of the forward channel because these signals have been decoded from the top levels of the received signal from the same relay. Also, repeating the $a_{n_{1B}-n_{2B}+n_{A2}-n_{A1}+1},...,a_{n_{1B}-n_{2B}+2(n_{A2}-n_{A1})}$ in the same level range that they are located (i.e., on block 3 in Figure \ref{fig:11.1}) does not affect the decoding since the signals can be decoded from top and their effect can be cancelled. Repeating any of them on the next $2(n_{A1}-n_{A2})+n_{2B}$ lower levels (i.e., on block 4 in Figure \ref{fig:11.1}) does not affect the decoding because node $B$ decodes the upper $n_{A2}-n_{A1}$ levels ($a_{n_{1B}-n_{2B}+n_{A2}-n_{A1}+1},...,a_{n_{1B}-n_{2B}+2(n_{A2}-n_{A1})}$) first and cancel the effect of the repeated signals. If $a_{n_{1B}-n_{2B}+n_{A2}-n_{A1}+1},...,a_{n_{1B}-n_{2B}+2(n_{A2}-n_{A1})}$ is repeated on the next $n_{A2}-n_{A1}$ lower levels (i.e., on block 5 in Figure \ref{fig:11.1}), it does not affect the decoding because as it was explained in the definition of the strategies, in case that the repetitions of some streams from two relays are from the same levels and are supposed to be repeated on the same level, the repeating relay ($R_1$) does not do the repeating for those levels and repeats those streams only one time. If any of them is repeated in any lower level by $R_1$, it is out of range and does not have effect on decoding of forward channel.

2. $(u_1,v_2,r_1)$: Figure \ref{fig:11.21} depicts the received signal at node $B$ (ignoring the effect of transmitted signal from $B$) assuming that both relays use Relay Strategy 0. The actual repetitions will be described below to show that messages can be decoded with the proposed strategies.

First, consider the case that the repetitions in favor of the forward and backward channels happen in different relays, i.e., $m_2, n_1=0$. In this case, $R_1$ repeats in favor of the forward channel (uses Relay Strategy $(6,0)$) and $R_2$ repeats in favor of the backward channel (uses Relay Strategy $(0,2)$ or $(0,6)$). Repeating $a_{n_{A2}-n_{A1}+1},...,a_{n_{1B}-n_{2B}+n_{A2}-n_{A1}}$ by $R_2$ in favor of the backward direction communication does not affect the achievability of the forward channel because they have been decoded from the top levels of the received signal from the other relay ($R_1$) (i.e., on block 1 in Figure \ref{fig:11.21}) where there is no interference from $R_2$. Also, repeating the $a_{1},...,a_{n_{A2}-n_{A1}}$ by $R_2$ within the top $n_{A2}-n_{A1}$ streams  (i.e., on blocks 2 and 3 in Figure \ref{fig:11.21}) does not affect the decoding since the signals can be decoded from top and their effect can be cancelled. Repeating any of them on the next $n_{2B}-2(n_{A2}-n_{A1})$ lower levels (i.e., on block 4 in Figure \ref{fig:11.21}) by $R_2$ does not affect the decoding because we can decode those upper $n_{A2}-n_{A1}$ levels ($a_{1},...,a_{n_{A2}-n_{A1}}$) first and cancel the effect of the repeated signals. If $a_{1},...,a_{n_{A2}-n_{A1}}$ is repeated by $R_2$ on the next lower levels (i.e., on blocks 5 and 6 in Figure \ref{fig:11.21}), in case that the repetitions of some streams from two relays are from the same level and are repeated on the same level, i.e., they create an equation as higher levels, $R_1$ does not do the repeating for those levels as was explained in the definition of the strategies, so there is no problem in decoding of forward direction channel. If any of them is repeated in any lower level by $R_2$, it is out of range and does not have effect on decoding of forward channel.

Now, take the case that the repetitions in favor of the forward and backward channels happen in the same relay ($R_1$), i.e., $(m_2,n_2)=(0,0)$. In this case, $R_1$ repeats in favor of both directions ($(m_1,m_1)\in \{(6,2), (6,6)\}$) and $R_2$ uses Relay Strategy $(0,0)$. Repeating $a_{n_{A2}-n_{A1}+1},...,a_{n_{1B}-n_{2B}+n_{A2}-n_{A1}}$ in favor of the backward channel does not affect the achievability of the forward channel because these signals have been decoded from the top levels of the received signal from the same relay. Also, repeating the $a_{n_{1B}-n_{2B}+n_{A2}-n_{A1}+1},...,a_{n_{1B}-n_{2B}+2(n_{A2}-n_{A1})}$ in the same level range that they are located (i.e., on blocks 2 and 3 in Figure \ref{fig:11.21}) does not affect the decoding since the signals can be decoded from top and their effect can be cancelled. Repeating any of them on the next $2(n_{A1}-n_{A2})+n_{2B}$ lower levels (i.e., on block 4 in Figure \ref{fig:11.21}) does not affect the decoding because node $B$ decodes the upper $n_{A2}-n_{A1}$ levels ($a_{n_{1B}-n_{2B}+n_{A2}-n_{A1}+1},...,a_{n_{1B}-n_{2B}+2(n_{A2}-n_{A1})}$) first and cancel the effect of the repeated signals. If $a_{n_{1B}-n_{2B}+n_{A2}-n_{A1}+1},...,a_{n_{1B}-n_{2B}+2(n_{A2}-n_{A1})}$ is repeated on the next $n_{A2}-n_{A1}$ lower levels (i.e., on blocks 5 and 6 in Figure \ref{fig:11.21}), it does not affect the decoding because as it was explained in the definition of the strategies, in case that the repetitions of some streams from two relays are from the same levels and are supposed to be repeated on the same level, the repeating relay ($R_1$) does not do the repeating for those levels and repeats those streams only one time. If $a_{2(n_{A2}-n_{A1})+n_{1B}-n_{2B}+1},...,a_{n_{1B}}$ is repeated within the range that they are located (i.e., on block 4 in Figure \ref{fig:11.21}), it does not affect the decoding since the signals can be decoded from highest level of it and their effect can be cancelled. Also if $a_{2(n_{A2}-n_{A1})+n_{1B}-n_{2B}+1},...,a_{n_{1B}}$ is repeated on the lower levels, it does not affect the decoding since they are decoded first from upper block. If any of them is repeated in any lower level by $R_1$, it is out of range and does not have effect on decoding of forward channel.

3. $(u_1,v_2,r_2,s_1,q_1)$: Figure \ref{fig:11.22.1.a} depicts the received signal at node $B$ (ignoring the effect of transmitted signal from $B$) assuming that both relays use Relay Strategy 0.

First, consider the case that the repetitions in favor of the forward and backward channels happen in different relays, i.e., $m_2, n_1=0$. In this case, $R_1$ repeats in favor of the forward channel (uses Relay Strategy $(6,0)$) and $R_2$ repeats in favor of the backward channel (uses Relay Strategy $(0,2)$ or $(0,6)$). Repeating $a_{n_{A2}-n_{A1}+1},...,a_{n_{1B}-n_{2B}+n_{A2}-n_{A1}}$ by $R_2$ in favor of the backward direction communication does not affect the achievability of the forward channel because they have been decoded from the top levels of the received signal from the other relay ($R_1$) where there is no interference from $R_2$. Also, repeating the $a_{1},...,a_{n_{A2}-n_{A1}}$ by $R_2$ within the top $n_{A2}-n_{A1}$ streams  (i.e., on block 2 in Figure \ref{fig:11.22.1.a}) does not affect the decoding since the signals can be decoded from top and their effect can be cancelled. Repeating any of them on the next $n_{1B}-n_{2B}$ lower levels by $R_2$ (i.e., on block 3 in Figure \ref{fig:11.22.1.a}) does not affect the decoding because we can decode those upper $n_{A2}-n_{A1}$ levels ($a_{1},...,a_{n_{A2}-n_{A1}}$) first and cancel the effect of the repeated signals. If $a_{1},...,a_{n_{A2}-n_{A1}}$ is repeated by $R_2$ on the next lower levels, in case that the repetitions of some streams from two relays are from the same level and are repeated on the same level, i.e., they create an equation as higher levels, $R_1$ does not do the repeating for those levels as was explained in the definition of the strategies, so there is no problem in decoding of forward direction channel. If any of them is repeated in any lower level by $R_2$, it is out of range and does not have effect on decoding of forward channel.

Now, take the case that the repetitions in favor of the forward and backward channels happen in the same relay ($R_1$), i.e., $(m_2,n_2)=(0,0)$. In this case, $R_1$ repeats in favor of both directions ($(m_1,m_1)\in \{(6,2), (6,6)\}$) and $R_2$ uses Relay Strategy $(0,0)$. Repeating $a_{n_{A2}-n_{A1}+1},...,a_{n_{1B}-n_{2B}+n_{A2}-n_{A1}}$ in favor of the backward channel does not affect the achievability of the forward channel because these signals have been decoded from the top levels of the received signal from the same relay. Also, repeating the $a_{n_{1B}-n_{2B}+n_{A2}-n_{A1}+1},...,a_{n_{1B}-n_{2B}+2(n_{A2}-n_{A1})}$ in the same level range that they are located (i.e., on block 2 in Figure \ref{fig:11.22.1.a}) does not affect the decoding since the signals can be decoded from top and their effect can be cancelled. Repeating any of them on the next $2(n_{A1}-n_{A2})+n_{2B}$ lower levels does not affect the decoding because node $B$ decodes the upper $n_{A2}-n_{A1}$ levels ($a_{n_{1B}-n_{2B}+n_{A2}-n_{A1}+1},...,a_{n_{1B}-n_{2B}+2(n_{A2}-n_{A1})}$) first and cancel the effect of the repeated signals. If $a_{n_{1B}-n_{2B}+n_{A2}-n_{A1}+1},...,a_{n_{1B}-n_{2B}+2(n_{A2}-n_{A1})}$ is repeated on the lowest $n_{A2}-n_{A1}$ levels (i.e., on block 8 in Figure \ref{fig:11.22.1.a}), it does not affect the decoding because as it was explained in the definition of the strategies, in case that the repetitions of some streams from two relays are from the same levels and are supposed to be repeated on the same level, the repeating relay ($R_1$) does not do the repeating for those levels and repeats those streams only one time. If any of them is repeated in any lower level by $R_1$, it is out of range and does not have effect on decoding of forward channel.

4. $(u_1,v_2,r_2,s_1,q_2)$: Figure \ref{fig:11.22.1.b} depicts the received signal at node $B$ (ignoring the effect of transmitted signal from $B$) assuming that both relays use Relay Strategy 0.

First, consider the case that the repetitions in favor of the forward and backward channels happen in different relays, i.e., $m_2, n_1=0$. In this case, $R_1$ repeats in favor of the forward channel (uses Relay Strategy $(6,0)$) and $R_2$ repeats in favor of the backward channel (uses Relay Strategy $(0,2)$ or $(0,6)$). Repeating $a_{n_{A2}-n_{A1}+1},...,a_{n_{1B}-n_{2B}+n_{A2}-n_{A1}}$ by $R_2$ in favor of the backward direction communication does not affect the achievability of the forward channel because they have been decoded from the top levels of the received signal from the other relay ($R_1$) where there is no interference from $R_2$. Also, repeating the $a_{1},...,a_{n_{A2}-n_{A1}}$ by $R_2$ within the top $n_{A2}-n_{A1}$ streams (i.e., on blocks 2 and 3 in Figure \ref{fig:11.22.1.b}) does not affect the decoding since the signals can be decoded from top and their effect can be cancelled. Repeating any of them on the next $n_{1B}-n_{2B}$ lower levels by $R_2$ (i.e., on block 4 in Figure \ref{fig:11.22.1.b}) does not affect the decoding because we can decode those upper $n_{A2}-n_{A1}$ levels ($a_{1},...,a_{n_{A2}-n_{A1}}$) first and cancel the effect of the repeated signals. If $a_{1},...,a_{n_{A2}-n_{A1}}$ is repeated by $R_2$ on the next lower levels, in case that the repetitions of some streams from two relays are from the same level and are repeated on the same level, i.e., they create an equation as higher levels, $R_1$ does not do the repeating for those levels as was explained in the definition of the strategies, so there is no problem in decoding of forward direction channel. If any of them is repeated in any lower level by $R_2$, it is out of range and does not have effect on decoding of forward channel.

Now, take the case that the repetitions in favor of the forward and backward channels happen in the same relay ($R_1$), i.e., $(m_2,n_2)=(0,0)$. In this case, $R_1$ repeats in favor of both directions ($(m_1,m_1)\in \{(6,2), (6,6)\}$) and $R_2$ uses Relay Strategy $(0,0)$. Repeating $a_{n_{A2}-n_{A1}+1},...,a_{n_{1B}-n_{2B}+n_{A2}-n_{A1}}$ in favor of the backward channel does not affect the achievability of the forward channel because these signals have been decoded from the top levels of the received signal from the same relay ($R_1$). Also, repeating the $a_{n_{1B}-n_{2B}+n_{A2}-n_{A1}+1},...,a_{n_{1B}-n_{2B}+2(n_{A2}-n_{A1})}$ in the same level range that they are located (i.e., on blocks 2 and 3 in Figure \ref{fig:11.22.1.b}) does not affect the decoding since the signals can be decoded from top and their effect can be cancelled. Repeating any of them on the next $2(n_{A1}-n_{A2})+n_{2B}$ lower levels (i.e., on blocks 4 and 5 in Figure \ref{fig:11.22.1.b}) does not affect the decoding because node $B$ decodes the upper $n_{A2}-n_{A1}$ levels ($a_{n_{1B}-n_{2B}+n_{A2}-n_{A1}+1},...,a_{n_{1B}-n_{2B}+2(n_{A2}-n_{A1})}$) first and cancel the effect of the repeated signals. If $a_{n_{1B}-n_{2B}+n_{A2}-n_{A1}+1},...,a_{n_{1B}-n_{2B}+2(n_{A2}-n_{A1})}$ is repeated on the lowest $n_{A2}-n_{A1}$ levels (i.e., on blocks 8 and 9 in Figure \ref{fig:11.22.1.b}), it does not affect the decoding because as it was explained in the definition of the strategies, in case that the repetitions of some streams from two relays are from the same levels and are supposed to be repeated on the same level, the repeating relay ($R_1$) does not do the repeating for those levels and repeats those streams only one time. If any of them is repeated in any lower level by $R_1$, it is out of range and does not have effect on decoding of forward channel.

5. $(u_1,v_2,r_2,s_2)$: Figure \ref{fig:11.22.2} depicts the received signal at node $B$ (ignoring the effect of transmitted signal from $B$) assuming that both relays use Relay Strategy 0.

First, consider the case that the repetitions in favor of the forward and backward channels happen in different relays, i.e., $m_2, n_1=0$. In this case, $R_1$ repeats in favor of the forward channel (uses Relay Strategy $(6,0)$) and $R_2$ repeats in favor of backward channel (uses Relay Strategy $(0,2)$ or $(0,6)$). Repeating $a_{n_{A2}-n_{A1}+1},...,a_{n_{1B}-n_{2B}+n_{A2}-n_{A1}}$ by $R_2$ in favor of the backward direction communication does not affect the achievability of the forward channel because they have been decoded from the top levels of the received signal from the other relay ($R_1$) where there is no interference from $R_2$. Also, repeating the $a_{1},...,a_{n_{A2}-n_{A1}}$ by $R_2$  within the top $n_{A2}-n_{A1}$ streams (i.e., on blocks 2 and 3 in Figure \ref{fig:11.22.2}) does not affect the decoding since the signals can be decoded from top and their effect can be cancelled. Repeating any of them on the next $n_{1B}-n_{2B}$ lower levels by $R_2$ (i.e., on blocks 4 and 5 in Figure \ref{fig:11.22.2}) does not affect the decoding because we can decode those upper $n_{A2}-n_{A1}$ levels ($a_{1},...,a_{n_{A2}-n_{A1}}$) first and cancel the effect of the repeated signals. If $a_{1},...,a_{n_{A2}-n_{A1}}$ is repeated by $R_2$ on the next lower levels, in case that the repetitions of some streams from two relays are from the same level and are repeated on the same level, i.e., they create an equation as higher levels, $R_1$ does not do the repeating for those levels as was explained in the definition of the strategies, so there is no problem in decoding of forward direction channel. If any of them is repeated in any lower level by $R_2$, it is out of range and does not have effect on decoding of forward channel.

Now, take the case that the repetitions in favor of the forward and backward channels happen in the same relay ($R_1$), i.e., $(m_2,n_2)=(0,0)$. In this case, $R_1$ repeats in favor of both directions ($(m_1,m_1)\in \{(6,2), (6,6)\}$) and $R_2$ uses Relay Strategy $(0,0)$. Repeating $a_{n_{A2}-n_{A1}+1},...,a_{n_{1B}-n_{2B}+n_{A2}-n_{A1}}$ in favor of the backward channel does not affect the achievability of the forward channel because these signals have been decoded from the top levels of the received signal from the same relay. Also, repeating the $a_{n_{1B}-n_{2B}+n_{A2}-n_{A1}+1},...,a_{n_{1B}-n_{2B}+2(n_{A2}-n_{A1})}$ in the same level range that they are located (i.e., on blocks 2 and 3 in Figure \ref{fig:11.22.2}) does not affect the decoding since the signals can be decoded from top and their effect can be cancelled. Repeating any of them on the next $2(n_{A1}-n_{A2})+n_{2B}$ lower levels does not affect the decoding because node $B$ decodes the upper $n_{A2}-n_{A1}$ levels ($a_{n_{1B}-n_{2B}+n_{A2}-n_{A1}+1},...,a_{n_{1B}-n_{2B}+2(n_{A2}-n_{A1})}$) first and cancel the effect of the repeated signals. If $a_{n_{1B}-n_{2B}+n_{A2}-n_{A1}+1},...,a_{n_{1B}-n_{2B}+2(n_{A2}-n_{A1})}$ is repeated on the next $n_{A2}-n_{A1}$ levels, it does not affect the decoding because as it was explained in the definition of the strategies, in case that the repetitions of some streams from two relays are from the same levels and are supposed to be repeated on the same level, the repeating relay ($R_1$) does not do the repeating for those levels and repeats those streams only one time. If any of them is repeated in any lower level by $R_1$, it is out of range and does not have effect on decoding of forward channel.

6. $(u_2,w_1)$: Figure \ref{fig:22.1} depicts the received signal at node $B$ (ignoring the effect of transmitted signal from $B$) assuming that both relays use Relay Strategy 0.

First, consider the case that the repetitions in favor of the forward and backward channels happen in different relays, i.e., $m_2, n_1=0$. In this case, $R_1$ repeats in favor of the forward channel (uses Relay Strategy $(6,0)$) and $R_2$ repeats in favor of the backward channel (uses Relay Strategy $(0,2)$ or $(0,6)$). Repeating $a_{n_{A2}-n_{A1}+1},...,a_{n_{2B}}$ by $R_2$ in favor of the backward direction communication does not affect the achievability of the forward channel because they have been decoded from the top levels of the received signal from the other relay ($R_1$) where there is no interference from $R_2$. Also, repeating the $a_{1},...,a_{n_{2B}+n_{A1}-n_{A2}}$ by $R_2$  within the top $n_{2B}-n_{A2}+n_{A1}$ streams (i.e., on block 2 in Figure \ref{fig:22.1}) does not affect the decoding since the signals can be decoded from top and their effect can be cancelled. Repeating any of them on the next $2(n_{A2}-n_{A1})-n_{2B}$ lower levels by $R_2$ (i.e., on block 3 in Figure \ref{fig:22.1}) does not affect the decoding because node $B$ decodes those upper $n_{2B}+n_{A1}-n_{A2}$ levels ($a_{1},...,a_{n_{2B}+n_{A1}-n_{A2}}$) first and cancel the effect of the repeated signals. If $a_{1},...,a_{n_{2B}+n_{A1}-n_{A2}}$ is repeated by $R_2$ on the next lower $n_{2B}+n_{A1}-n_{A2}$ levels (i.e., on block 4 in Figure \ref{fig:22.1}), in case that the repetitions of some streams from two relays are from the same level and are repeated on the same level, i.e., they create an equation as higher levels, $R_1$ does not do the repeating for those levels as was explained in the definition of the strategies, so there is no problem in decoding of forward direction channel. If $a_{n_{2B}-(n_{A2}-n_{A1})+1},...,a_{n_{A2}-n_{A1}}$ is repeated within the range that they are located (i.e., on block 3 in Figure \ref{fig:22.1}), it does not affect the decoding since the signals can be decoded from highest level of it and their effect can be cancelled. Also if $a_{n_{2B}-(n_{A2}-n_{A1})+1},...,a_{n_{A2}-n_{A1}}$ is repeated on the $n_{2B}-(n_{A2}-n_{A1})$ lower levels (i.e., on block 4 in Figure \ref{fig:22.1}), it does not affect the decoding since they are decoded first from upper block. If any of them is repeated in any lower level by $R_2$, it is out of range and does not have effect on decoding of forward channel.

Now, take the case that the repetitions in favor of the forward and backward channels happen in the same relay ($R_1$), i.e., $(m_2,n_2)=(0,0)$. In this case, $R_1$ repeats in favor of both directions ($(m_1,m_1)\in \{(6,2), (6,6)\}$) and $R_2$ uses the Relay Strategy $(0,0)$. Repeating $a_{n_{A2}-n_{A1}+1},...,a_{n_{1B}-n_{2B}+n_{A2}-n_{A1}}$ in favor of the backward channel does not affect the achievability of the forward channel because these signals have been decoded from the top levels of the received signal from the same relay. Also, repeating the $a_{n_{1B}-n_{2B}+n_{A2}-n_{A1}+1},...,a_{n_{1B}}$ in the same level range that they are located (i.e., on block 2 in Figure \ref{fig:22.1}) does not affect the decoding since the signals can be decoded from top and their effect can be cancelled. Repeating any of them on the next $2(n_{A2}-n_{A1})-n_{2B}$ lower levels (i.e., on block 3 in Figure \ref{fig:22.1}) does not affect the decoding because node $B$ decodes the upper $n_{2B}+n_{A1}-n_{A2}$ levels ($a_{n_{1B}-n_{2B}+n_{A2}-n_{A1}+1},...,a_{n_{1B}}$) first and cancel the effect of the repeated signals. If $a_{n_{1B}-n_{2B}+n_{A2}-n_{A1}+1},...,a_{n_{1B}}$ is repeated on the next $n_{2B}+n_{A1}-n_{A2}$ lower levels (i.e., on block 4 in Figure \ref{fig:22.1}), it does not affect the decoding because as it was explained in the definition of the strategies, in case that the repetitions of some streams from two relays are from the same level and are supposed to be repeated on the same level, the repeating relay ($R_1$) does not do the repeating for those levels and repeats those streams only one time. If any of them is repeated in any lower level by $R_1$, it is out of range and does not have effect on decoding of forward channel.

7. $(u_2,w_2)$: Figure \ref{fig:22.2} depicts the received signal at node $B$ (ignoring the effect of transmitted signal from $B$) assuming that both relays use Relay Strategy 0.

First, consider the case that the repetitions in favor of the forward and backward channels happen in different relays, i.e., $m_2, n_1=0$. In this case, $R_1$ repeats in favor of the forward channel (uses Relay Strategy $(6,0)$) and $R_2$ repeats in favor of the backward channel (uses Relay Strategy $(0,2)$ or $(0,6)$). Repeating $a_{n_{A2}-n_{A1}+1},...,a_{n_{A2}-n_{A1}+n_{1B}-n_{2B}}$ by $R_2$ in favor of the backward direction communication does not affect the achievability of the forward channel because they have been decoded from the top levels of the received signal from the other relay ($R_1$) where there is no interference from $R_2$. Also, repeating the $a_{1},...,a_{n_{2B}+n_{A1}-n_{A2}}$ by $R_2$  within the top $n_{2B}-n_{A2}+n_{A1}$ streams (i.e., on blocks 2 and 3 in Figure \ref{fig:22.2}) does not affect the decoding since the signals can be decoded from top and their effect can be cancelled. Repeating any of them on the next $2(n_{A2}-n_{A1})-n_{2B}$ lower levels by $R_2$ (i.e., on block 4 in Figure \ref{fig:22.2}) does not affect the decoding because node $B$ decodes the upper $n_{2B}+n_{A1}-n_{A2}$ levels ($a_{1},...,a_{n_{2B}+n_{A1}-n_{A2}}$) first and cancel the effect of the repeated signals. If $a_{1},...,a_{n_{2B}+n_{A1}-n_{A2}}$ is repeated by $R_1$ on the next $n_{1B}-n_{2B}$ lower levels (i.e., on block 5 in Figure \ref{fig:22.2}), in case that the repetitions of some streams from two relays are from the same level and are repeated on the same level, i.e., they create an equation as higher levels, $R_1$ does not do the repeating for those levels as was explained in the definition of the strategies, so there is no problem in decoding of forward direction channel. If any of them is repeated in any lower level by $R_2$, it is out of range and does not have effect on decoding of forward channel.

Now, take the case that the repetitions in favor of the forward and backward channels happen in the same relay ($R_1$), i.e., $(m_2,n_2)=(0,0)$. In this case, $R_1$ repeats in favor of both directions ($(m_1,m_1)\in \{(6,2), (6,6)\}$) and $R_2$ uses the Relay Strategy $(0,0)$. Repeating $a_{n_{A2}-n_{A1}+1},...,a_{n_{1B}-n_{2B}+n_{A2}-n_{A1}}$ in favor of the backward channel does not affect the achievability of the forward channel because these signals have been decoded from the top levels of the received signal from the same relay. Also, repeating the $a_{n_{1B}-n_{2B}+n_{A2}-n_{A1}+1},...,a_{n_{1B}}$ in the same level range that they are located does not affect the decoding since the signals can be decoded from top and their effect can be cancelled. Repeating any of them on the next $2(n_{A2}-n_{A1})-n_{2B}$ lower levels (i.e., on block 4 in Figure \ref{fig:22.2}), does not affect the decoding because node $B$ decodes the upper $n_{2B}+n_{A1}-n_{A2}$ levels ($a_{n_{1B}-n_{2B}+n_{A2}-n_{A1}+1},...,a_{n_{1B}}$) first and cancel the effect of the repeated signals. If $a_{n_{1B}-n_{2B}+n_{A2}-n_{A1}+1},...,a_{n_{1B}}$ is repeated on the next $n_{1B}-n_{2B}$ lower levels (i.e., on block 5 in Figure \ref{fig:22.2}), it does not affect the decoding because as it was explained in the definition of the strategies, in case that the repetitions of some streams from two relays are from the same levels and are supposed to be repeated on the same level, the repeating relay ($R_1$) does not do the repeating for those levels and repeats those streams only one time. If any of them is repeated in any lower level by $R_1$, it is out of range and does not have effect on decoding of forward channel.

\end{appendices}

\bibliographystyle{IEEETran}
\bibliography{bib,fdbib}

\end{document}